\documentclass[a4paper,twocolumn,10pt]{IEEEtran}
\pdfoutput=1
\usepackage[utf8]{inputenc}

\usepackage[english]{babel}
\usepackage[T1]{fontenc}
\usepackage{amsmath}
\usepackage[unicode=true,bookmarks=true,bookmarksnumbered=false,bookmarksopen=false,breaklinks=false,pdfborder={0 0 1}, backref=false,colorlinks=true]{hyperref}

\usepackage{xcolor}
\definecolor{myred}{rgb}{0.8,0.3,0.2}
\definecolor{myredfill}{rgb}{1,0.93,0.93}

\definecolor{myred2}{rgb}{0.6,0.5,0.3}
\definecolor{myredfill2}{rgb}{0.93,0.83,0.83}

\definecolor{myblue}{rgb}{0.05,0.25,0.7}
\definecolor{mybluefill}{rgb}{0.92,0.98,1}
\definecolor{mygreen}{rgb}{0.3,0.6,0.4}
\definecolor{mygreenfill}{rgb}{0.93,1,0.97}
\definecolor{mygrey}{rgb}{0.4,0.4,0.4}
\definecolor{mygreyfill}{rgb}{0.95,0.95,0.95}
\definecolor{mypurple}{rgb}{0.6,0.2,0.6}
\definecolor{mypurplefill}{rgb}{0.99,0.94,0.99}
\definecolor{myyellow}{rgb}{9,0.9,0.3}
\definecolor{myyellowfill}{rgb}{1,1,0.92}
\definecolor{myorange}{rgb}{9,0.7,0.3}
\definecolor{myorangefill}{rgb}{1,0.95,0.92}
\definecolor{mybrown}{rgb}{0.4,0.25,0.1}
\definecolor{mybrownfill}{rgb}{1,0.9,0.8}

\definecolor{mediumgrey}{rgb}{0.7,0.7,0.7}

\definecolor{linkred}{rgb}{0.6,0.1,0.1}
\definecolor{citeblue}{rgb}{0.2,0.35,0.75}
\definecolor{urlblue}{rgb}{0.2,0.25,0.45}

\hypersetup{
     linkcolor = linkred,
     anchorcolor = linkred,
     citecolor = citeblue,
     filecolor = linkred,
     urlcolor = urlblue
     }

\usepackage{bbm}
\usepackage{stmaryrd}
\usepackage{tabu,multirow}
\usepackage{float} 
\usepackage{lipsum}
\usepackage[autostyle]{csquotes}
\usepackage{tikz-cd}
\usetikzlibrary{decorations.pathreplacing,calligraphy}
\usepackage{amsmath,amsthm,amssymb,latexsym}
\usepackage{mathrsfs}
\usepackage{hhline}
\usepackage[caption=false]{subfig}
\usetikzlibrary{backgrounds}
\usepackage{graphicx}
\usepackage{authblk}

\usepackage[numbers,sort&compress]{natbib}

\DeclareMathOperator*{\functioncomposition}{\bigcirc}

\DeclareMathOperator*{\tensorcomposition}{\bigotimes}

\newtheorem{theorem}{Theorem}

\newtheorem{lemma}{Lemma}
\newtheorem{corollary}{Corollary}
\newtheorem{example}{Example}
\newtheorem{observation}{Observation}

\begin{document}

\title{\Large{Extracting Quantum Dynamical Resources: Consumption of Non-Markovianity for Noise Reduction}}

\author[1]{Graeme D. Berk$^*$}
\author[2]{Simon Milz}
\author[1]{Felix A. Pollock}
\author[1]{Kavan Modi}

\affil[1]{School of Physics and Astronomy, Monash University, Victoria 3800, Australia}
\affil[2]{Institute for Quantum Optics and Quantum Information\\ Austrian Academy of Sciences, Boltzmanngasse 3, A-1090 Vienna, Austria}

\date{\today}

\maketitle

\begin{abstract}
Noise is possibly the most formidable challenge for quantum technologies. As such, a great deal of effort is dedicated to developing methods for noise reduction. One remarkable achievement in this direction is dynamical decoupling; it details a clear set of instructions for counteracting the effects of quantum noise. Yet, the domain of its applicability remains limited to devices where exercising fast control is possible. In practical terms, this is highly limiting and there is a growing need for better noise reduction tools. Here we take a significant step in this direction, by identifying the crucial ingredients required for noise suppression and the development of methods that far outperform traditional dynamical decoupling techniques. Using resource theoretic methods, we show that the key resource responsible for the efficacy of dynamical decoupling, and related protocols, is non-Markovianity (or temporal correlations). Using this insight, we then propose two methods to identify optimal pulse sequences for noise reduction. With an explicit example, we show that our methods enable a more optimal exploitation of temporal correlations, and extend the timescales at which noise suppression is viable by at least two orders of magnitude. Importantly, the corresponding tools are built on operational grounds and are easily implemented in the current generation of quantum devices.
\end{abstract}

Even the most promising platforms for quantum computing~\cite{aquantumengineersguidetosuperconducting, trappedioncomputing} are inherently plagued with complex quantum noise~\cite{quantumnoise, nmrev, lili-hall-wiseman, quantumstochasticprocessesandquantumnon}. This noise must be significantly reduced to meet the threshold required for error corrected quantum computing~\cite{PhysRevLett.96.050504, quantumerrorcorrectionforbeginners, quantumerrorcorrectionforquantummemories}. Reduced noise will also significantly enhance the performance of current generation devices~\cite{arXiv:2101.08448}. These facts have lead to a flurry of techniques for noise characterisation and control~\cite{demonstrationofnonmarkovianprocess, PRXQuantum.2.030315}. While many different approaches exist, the general goal is to maximise the retained information between an input and the final output of the dynamics, for an arbitrary underlying noise process, by means of active experimental interventions.

A prototypical example of this category is dynamical decoupling (DD)~\cite{dynamicaldecouplingofopenquantumsystems, ddreview, dynamicaldecouplingofunbounded, demonstrationoffidelityimprovement,preservingelectronspincoherence, dynamicaldecouplingofasingleelectronspinatroom,suppressionofcrosstalk}, where a fixed sequence of unitary transformations is applied, having the effect of cancelling the detrimental influence of the environment. Achieving high-efficacy of such methods requires fast control, which is usually not practical. In such instances the feasibility of DD depends on the specific details of the noise process. For example, it is known that amenability to DD is linked~\cite{canquantummarkovevolutions} to the ability of the environment to carry a memory, also known as non-Markovianity~\cite{nonmarkovianquantumprocesses, 1367-2630-18-6-063032}. However, this connection is not well understood~\cite{nonmarkoviannoisethatcannotbe} and there is a dire need for a universal understanding of methods like DD, which is thus far missing~\cite{dynamicaldecouplingefficiencyversusquantumnonmarkovianity, nonmarkoviannoisethatcannotbe}.

In the search for a birdseye view of quantum control and noise mitigation techniques, quantum resource theories offer a promising lens. For example, in DD, memory (the resource) is expended by means of experimentally implementable (i.e., free) operations, to minimize the unwanted influence of the environment. More generally, a resource theory consists of a set of resource objects, and a set of transformations between those resource objects. The value and the inter-convertibility of resources can be quantified with monotones~\cite{review}. While resource theories were originally envisaged to quantify the utility of properties of \textit{static} resources, like quantum states~\cite{entanglementtheoryandthesecondlaw, operationalresourcetheoryofcoherence, theresourcetheoryofinformationalnonequilibriuminthermodynamics}, they have recently found manifold applications for dynamical objects, such as trace preserving mappings between quantum states, known as quantum channels~\cite{operationalresourcetheoryofquantumchannels, comparisonofquantumchannelsbysuperchannels, resourcetheoriesofquantumchannels}. 

Leveraging on this resource theoretic angle, our aim is to determine, quantify, and pin-point the amenability of noise processes to techniques like DD. Existing approaches, like those based solely on quantum channels, suffer from the problem that they cannot account for intermediate interventions between the input states and output states they relate. Supermaps, the most general transformations of a quantum channel~\cite{supermaps}, cannot bypass this inherent constraint of channels either. Indeed, it is exactly the \textit{complex multitime correlations}~\cite{witnessingquantummemory, genuinemultipartiteentanglementintime} -- not fully taken into consideration by previous approaches -- that enable DD in the first place. Here, we employ process tensors~\cite{quantumstochasticprocessesandquantumnon, nonmarkovianquantumprocesses} which are specifically designed to account for experimental interventions \emph{between} the input and output, as is the case in, for example, DD setups. In fact, \textit{any} quantum process can be represented as a process tensor~\cite{kolmogorov}, rendering the approach completely general. Moreover, resource theories for quantum processes (RTQP), built around the process tensor formalism, have recently been developed~\cite{resourcetheoriesofmultitime}, thus allowing for a comprehensive and systematic study of quantum control and noise mitigation techniques.

While the RTQP brings us a step closer to identifying the core dynamical assets for tasks like DD, it suffers from a problem with monotonicity; a free transformation of one process into another cannot increase its value. However, this is seemingly required in noise reduction scenarios. Overcoming this apparent paradox necessitates one more ingredient -- temporal coarse-graining. Whose irreversibility property, detailed in Obs.~\ref{obs:irreversibility}, accounts for how greater control at the short timescale translates to noise suppression at the long timescale. In turn, this identifies the second vital resource besides memory for methods like DD: the timing and speed of control operations. Combining coarse-graining with RTQP results in \emph{resource theories of temporal resolution}, which unambiguously account for the resources needed for noise suppression. Beyond DD, our results can be readily applied to identify and quantify the resources in other quantum dynamical phenomena, e.g. the quantum Zeno effect (QZE), decoherence-free subspace (DFS), and even quantum error correction (QEC). Finally, our work lays the ground work for a theory for distillation and formation of quantum channel information transmission.

Below, we introduce a prototypical resource theory of temporal resolution for information preservation tasks $\mathsf{Q}_{\hat{\mathbbm{W}}}$, containing a rich structure of operationally significant sub-theories, distinguished by different levels of experimental control, and thus, different sets of free operations. In particular, the sub-theory $\mathsf{D}_{\hat{\mathbbm{W}}} \subset \mathsf{Q}_{\hat{\mathbbm{W}}}$ allows the conversion of correlations in time (non-Markovianity) into system-level correlations, `decoupling' the system from its environment. Thus, $\mathsf{D}_{\hat{\mathbbm{W}}}$ contains DD, as well as an optimal information preservation protocol built up within the paradigm of semi-definite optimisation. We illustrate the latter by numerically showing its supremacy over DD at multiple timescales, and demonstrating the close connection between temporal resolution and decoupling success. This result adds to the growing body of work towards minimising noise by characterising and harnessing the underlying process tensor~\cite{demonstrationofnonmarkovianprocess, nonmarkovianquantumprocesstomography, diagnosingtemporalquantum}.

\section{Resource Theories of Temporal Resolution} \label{sec:resourcetheoriesoftemporalresolution}

To formalise and make sense of scenarios like DD from a resource theoretic perspective, there are two prerequisites. Firstly, noise processes must be cast as resources, and DD must be seen as transformations of those noise processes, which we present in Sec.~\ref{sec:resourcetheoriesofmultitimeprocesses}. Secondly, there must be a mechanism by which DD sequences can be perceived as resource increasing, despite being free transformations in the respective experimental setups. This is enabled by the irreversibility property of temporal coarse-graining, presented in Sec.~\ref{sec:temporalcoarsegraining}.

Any quantum noise can be modelled as an evolution operator $\mathcal{T}^{se}_{t:0}$ jointly acting on the system ($s$) of interest and its environment ($e$) for some time $t$. The noise then manifests as correlations between $s$ and $e$, and the goal of DD, and related methods, is to minimise the build-up of these $se$ correlations and maximise the input-output correlations. While we only have access to $s$, remarkably DD can do just this by `averaging' out the influence of the $se$ interaction on the system. This is achieved by breaking up the process $\mathcal{T}^{se}_{t:0}$ between a whole number $n \in \mathbbm{W}$ ($\mathbbm{W}=\mathbbm{N}\cup \{ 0\}$) of \emph{intermediate} times $\hat{n}=\{ t_1 , \dots , t_n \}$, and applying control operations on $s$ at each of those times. Fig.~\ref{fig:definingobjects} shows a single-qubit noise process, broken up into four segments, and then subjected to two different control sequences. Importantly, in both cases the noise process is the same `comb', independent of the applied controls. In other words, the noise processes contain dynamical resources and with control operations we can extract them.

The noise process, with $n$ intermediate interventions, has a concise representation as a quantum comb~\cite{quantumnetworks}, also known as a process tensor $\mathbf{T}_{\hat{n}}$~\cite{quantumstochasticprocessesandquantumnon, nonmarkovianquantumprocesses, 1367-2630-18-6-063032}, consisting of sequences of $se$ evolution maps
\begin{equation} \label{eq:processtensor}
    \mathbf{T}_{\hat{n}}:= \text{tr}_{e} \circ_e \mathcal{T}^{se}_{t:t_{n}} \circ_e \dots \circ_e \mathcal{T}^{se}_{t_{1}:0} \circ_e \rho^{e}_0,
\end{equation}
where $\rho^{e}_0$ is an initial environment state, and $\circ_e$ denotes composition only on the $e$ Hilbert spaces. This ensures that the process tensor describes the multitime noise on $s$ alone without knowing the details of $e$, and incorporates all pertinent memory effects between different points in time. A key result of the process tensor framework is a necessary and sufficient condition for quantum Markovianity and operationally meaningful measures for quantum memory~\cite{operationalmarkovcondition, 1367-2630-18-6-063032}.

Exercising control over this process, e.g. the DD sequences (on $s$ alone), amounts to contraction\footnote[1]{This contraction is written as $\llbracket \mathbf{T}_{\hat{n}} | \mathbf{A}_{\hat{n}} \rrbracket:=\text{tr}_{\text{int}}\big\{\mathbf{T}_{\hat{n}} (\mathbbm{1}^{\text{in}} \otimes \mathbf{A}_{\hat{n}}^{T} \otimes \mathbbm{1}^{\text{out}} ) \big\}$, where $T$ is the transpose, `in' is the input Hilbert space, `int' corresponds to all intermediate spaces, and `out' is the output space.} of the above tensor with an analogous control tensor $\mathbf{A}_{\hat{n}}$, at intermediate times $\hat{n}$, to yield the quantum channel
\begin{gather} \label{eq:PTcontrol}
\mathbf{T}'_{\emptyset} := \llbracket \mathbf{T}_{\hat{n}} | \mathbf{A}_{\hat{n}} \rrbracket.
\end{gather}
We use a square bra-ket notation~\cite{resourcetheoriesofmultitime} to denote this action, where $\mathbf{A}_{\hat{n}}$, in general, contains both the logical gates of a computation and the pulse sequences requisite for noise reduction, like, e.g. DD. The ultimate goal of noise reduction methods is to maximise the input-output correlations of $\mathbf{T}'_{\emptyset}$, i.e., the output of the computation should be highly correlated with the input. 

In practice, however, standard noise reduction methods lose their effectiveness when the temporal resolution of the control is too low. This raises the question (see Fig.~\ref{fig:definingobjects}): are these methods working at the fundamental limits set by quantum mechanics, or are there more dynamical resources available for extraction, by as of yet untapped by noise mitigation techniques? We will show below what explicit resources are at one's disposal, and that a more efficient conversion of non-Markovianity into system-level coherence results in a significant lengthening in the timescales for noise suppression.

\begin{figure*}[ht!]
\centering
\begin{tikzpicture}[scale=0.24]
\draw[myred,fill=myredfill, thick,solid,rounded corners=4] (-2,3) -- (-2,4) -- (28,4) -- (28,2+0.1) -- (26-0.1,2+0.1) -- (26-0.1,0+0.1) -- (24+0.1,0+0.1) -- (24+0.1,2+0.1)-- (18-0.1,2+0.1) -- (18-0.1,0+0.1) -- (16+0.1,0+0.1) -- (16+0.1,2+0.1) -- (10-0.1,2+0.1) -- (10-0.1,0+0.1) -- (8+0.1,0+0.1) -- (8+0.1,2+0.1) -- (2-0.1,2+0.1) -- (2-0.1,0+0.1) -- (0+0.1,0+0.1) -- (0+0.1,2+0.1) -- (-2,2+0.1) -- (-2,3)   ;

\draw[black, very thick,solid] (-1,3) -- (27,3);
\draw[black, very thick,solid] (26.8,2.7) -- (27.2,3.3);
\draw[black, very thick,solid] (-3,1) -- (29,1);


\draw[myred,fill=myredfill,very thick,solid,rounded corners=2] (-1.7,3.7) rectangle (-0.3,2.3);
\draw[myred,fill=myredfill,very thick,solid,rounded corners=2] (0.3,3.7) rectangle (1.7,0.3);
\draw[myred,fill=myredfill,very thick,solid,rounded corners=2] (8.3,3.7) rectangle (9.7,0.3);
\draw[myred,fill=myredfill,very thick,solid,rounded corners=2] (16.3,3.7) rectangle (17.7,0.3);
\draw[myred,fill=myredfill,very thick,solid,rounded corners=2] (24.3,3.7) rectangle (25.7,0.3);


\draw[] (-1,3) node {\footnotesize $\rho^{e}_0$};
\draw[] (1,2) node[rotate=0] {\small $\mathcal{T}$};
\draw[] (9,2) node[rotate=0] {\small $\mathcal{T}$};
\draw[] (17,2) node[rotate=0] {\small $\mathcal{T}$};
\draw[] (25,2) node[rotate=0] {\small $\mathcal{T}$};
\draw[] (-2.7,1.5) node {\small $\rho_\text{in}$};
\draw[] (29.1,1.5) node {\small $\rho_\text{out}$};

\draw[myblue,line width=3pt] (4.3,1) -- (5.7,1);

\draw[myblue,line width=3pt] (12.3,1) -- (13.7,1);

\draw[myblue,line width=3pt] (20.3,1) -- (21.7,1);

\draw[black, very thick,solid] (3+16+4,-2-1-1) -- (7+16+4,-2-1-1);
\draw[] (2.9-0.8+16+4,-2-1-1) node {\footnotesize $\rho_\text{in}$};
\draw[] (7.4+0.8+16+4,-2-1-1) node {\footnotesize $\rho_\text{out}$};

\draw[myred2,fill=myredfill2,very thick,solid,rounded corners=2] (4+16+4,2-3-1-1) rectangle (6+16+4,0-3-1-1);

\draw[] (5+0.1+16+4,0-1-0.5) node[rotate=90] {\Large $=$};

\draw[] (5+16+4,-2-1-1) node[rotate=0] { $\mathbf{T_{\emptyset}}$};

\draw[black, very thick,solid] (3+16+20-4,-2-1-1) -- (7+16+20-4,-2-1-1);
\draw[] (2.9-0.8+16+20-4,-2-1-1) node {\footnotesize $\rho_\text{in}$};
\draw[] (7.4+0.8+16+20-4,-2-1-1) node {\footnotesize $\rho_\text{out}$};

\draw[myred2,fill=myredfill2,very thick,solid,rounded corners=2] (4+16+20-4,2-3-1-1) rectangle (6+16+20-4,0-3-1-1);

\draw[] (5+0.1+16+20-4,0-1-0.5) node[rotate=90] {\Large $=$};

\draw[] (5+16+20-4,-2-1-1) node[rotate=0] { $\mathbf{T'_{\emptyset}}$};

\draw[] (13+0.1+18,-2-2) node[rotate=0] {\Huge $\mathbf{<}$};

\draw[] (13-0.1+18,-2-2-3) node[rotate=0] {(a) Process Inequality};




\draw[myred,fill=myredfill, thick,solid,rounded corners=4] (-2+36,3) -- (-2+36,4) -- (28+36,4) -- (28+36,2+0.1) -- (26-0.1+36,2+0.1) -- (26-0.1+36,0+0.1) -- (24+0.1+36,0+0.1) -- (24+0.1+36,2+0.1)-- (18-0.1+36,2+0.1) -- (18-0.1+36,0+0.1) -- (16+0.1+36,0+0.1) -- (16+0.1+36,2+0.1) -- (10-0.1+36,2+0.1) -- (10-0.1+36,0+0.1) -- (8+0.1+36,0+0.1) -- (8+0.1+36,2+0.1) -- (2-0.1+36,2+0.1) -- (2-0.1+36,0+0.1) -- (0+0.1+36,0+0.1) -- (0+0.1+36,2+0.1) -- (-2+36,2+0.1) -- (-2+36,3)   ;

\draw[black, very thick,solid] (-1+36,3) -- (27+36,3);
\draw[black, very thick,solid] (26.8+36,2.7) -- (27.2+36,3.3);
\draw[black, very thick,solid] (-3+36,1) -- (29+36,1);

\draw[myred,fill=myredfill,ultra thick,solid,rounded corners=2] (-1.7+36,3.7) rectangle (-0.3+36,2.3);
\draw[myred,fill=myredfill,ultra thick,solid,rounded corners=2] (0.3+36,3.7) rectangle (1.7+36,0.3);
\draw[myred,fill=myredfill,ultra thick,solid,rounded corners=2] (8.3+36,3.7) rectangle (9.7+36,0.3);
\draw[myred,fill=myredfill,ultra thick,solid,rounded corners=2] (16.3+36,3.7) rectangle (17.7+36,0.3);
\draw[myred,fill=myredfill,ultra thick,solid,rounded corners=2] (24.3+36,3.7) rectangle (25.7+36,0.3);

\draw[mypurple,fill=mypurplefill,very thick,solid,rounded corners=2] (-1.7+36,1.7) rectangle (-0.3+36,0.3);
\draw[mypurple,fill=mypurplefill,very thick,solid,rounded corners=2] (4.3-2+36,1.7) rectangle (5.7-2+36,0.3);
\draw[mypurple,fill=mypurplefill,very thick,solid,rounded corners=2] (4.3+2+36,1.7) rectangle (5.7+2+36,0.3);
\draw[mypurple,fill=mypurplefill,very thick,solid,rounded corners=2] (12.3-2+36,1.7) rectangle (13.7-2+36,0.3);
\draw[mypurple,fill=mypurplefill,very thick,solid,rounded corners=2] (12.3+2+36,1.7) rectangle (13.7+2+36,0.3);
\draw[mypurple,fill=mypurplefill,very thick,solid,rounded corners=2] (18.3+36,1.7) rectangle (19.7+36,0.3);
\draw[mypurple,fill=mypurplefill,very thick,solid,rounded corners=2] (20.3+2+36,1.7) rectangle (21.7+2+36,0.3);
\draw[mypurple,fill=mypurplefill,very thick,solid,rounded corners=2] (26.3+36,1.7) rectangle (27.7+36,0.3);

\draw[] (-1+36,3) node {\footnotesize $\rho^{e}_0$};
\draw[] (1+36,2) node[rotate=0] {\small $\mathcal{T}$};
\draw[] (9+36,2) node[rotate=0] {\small $\mathcal{T}$};
\draw[] (17+36,2) node[rotate=0] {\small $\mathcal{T}$};
\draw[] (25+36,2) node[rotate=0] {\small $\mathcal{T}$};
\draw[] (-1+36,1) node[rotate=0] {\small $\mathcal{I}$};
\draw[] (5-2+36,1) node[rotate=0] {\small $\mathcal{I}$};
\draw[] (5+2+36,1) node[rotate=0] {\small $\mathcal{X}$};
\draw[] (13-2+36,1) node[rotate=0] {\small $\mathcal{X}$};
\draw[] (13+2+36,1) node[rotate=0] {\small $\mathcal{Y}$};
\draw[] (21-2+36,1) node[rotate=0] {\small $\mathcal{Y}$};
\draw[] (21+2+36,1) node[rotate=0] {\small $\mathcal{Z}$};
\draw[] (27+36,1) node[rotate=0] {\small $\mathcal{Z}$};
\draw[] (-2.7+36,1.5) node {\small $\rho_\text{in}$};
\draw[] (29.1+36,1.5) node {\small $\rho_\text{out}$};



\draw[myblue,line width=3pt] (4.3+36,1) -- (5.7+36,1);

\draw[myblue,line width=3pt] (12.3+36,1) -- (13.7+36,1);

\draw[myblue,line width=3pt] (20.3+36,1) -- (21.7+36,1);





\end{tikzpicture}

\begin{minipage}{0.3\linewidth}
\vspace{-3.5cm}
        \centering
\begin{tikzpicture}[scale=0.75]


    \draw[] (0,3) node {\scriptsize $t_{16}$\ $t_0$};
    
    \draw[] (1.15,2.772) node {\scriptsize $t_1$};
    
    \draw[] (2.12,2.12) node {\scriptsize $t_2$};
    
    \draw[] (2.772,1.15) node {\scriptsize $t_3$};
    
    \draw[] (3,0) node {\scriptsize $t_4$};
    
    \draw[] (2.772,-1.15) node {\scriptsize $t_5$};
    
    \draw[] (2.12,-2.12) node {\scriptsize $t_6$};
    
    \draw[] (1.15,-2.772) node {\scriptsize $t_7$};
    
    \draw[] (0,-3) node {\scriptsize $t_8$};
    
    
    \draw[] (-1.15,2.772) node {\scriptsize $t_{15}$};
    
    \draw[] (-2.12,2.12) node {\scriptsize $t_{14}$};
    
    \draw[] (-2.772,1.15) node {\scriptsize $t_{13}$};
    
    \draw[] (-3,0) node {\scriptsize $t_{12}$};
    
    \draw[] (-2.772,-1.15) node {\scriptsize $t_{11}$};
    
    \draw[] (-2.12,-2.12) node {\scriptsize $t_{10}$};
    
    \draw[] (-1.15,-2.772) node {\scriptsize $t_9$};


    \draw [rotate=360/32] [mypurple,thick,domain=-0.8224:0.8224,smooth] plot ({(0.45*\x)}, {(\x)^2+2});
    
    \draw [rotate=360/32+360/16] [mypurple,thick,domain=-0.8224:0.8224,smooth] plot ({(0.45*\x)}, {(\x)^2+2});
    
    \draw [rotate=360/32+2*360/16] [mypurple,thick,domain=-0.8224:0.8224,smooth] plot ({(0.45*\x)}, {(\x)^2+2});
    
    \draw [rotate=360/32+3*360/16] [mypurple,thick,domain=-0.8224:0.8224,smooth] plot ({(0.45*\x)}, {(\x)^2+2});
    
    \draw [rotate=360/32+4*360/16] [mypurple,thick,domain=-0.8224:0.8224,smooth] plot ({(0.45*\x)}, {(\x)^2+2});
    
    \draw [rotate=360/32+5*360/16] [mypurple,thick,domain=-0.8224:0.8224,smooth] plot ({(0.45*\x)}, {(\x)^2+2});
    
    \draw [rotate=360/32+6*360/16] [mypurple,thick,domain=-0.8224:0.8224,smooth] plot ({(0.45*\x)}, {(\x)^2+2});
    
    \draw [rotate=360/32+7*360/16] [mypurple,thick,domain=-0.8224:0.8224,smooth] plot ({(0.45*\x)}, {(\x)^2+2});
    
    \draw [rotate=360/32+8*360/16] [mypurple,thick,domain=-0.8224:0.8224,smooth] plot ({(0.45*\x)}, {(\x)^2+2});
    
    \draw [rotate=360/32+9*360/16] [mypurple,thick,domain=-0.8224:0.8224,smooth] plot ({(0.45*\x)}, {(\x)^2+2});
    
    \draw [rotate=360/32+10*360/16] [mypurple,thick,domain=-0.8224:0.8224,smooth] plot ({(0.45*\x)}, {(\x)^2+2});
    
    \draw [rotate=360/32+11*360/16] [mypurple,thick,domain=-0.8224:0.8224,smooth] plot ({(0.45*\x)}, {(\x)^2+2});
    
    \draw [rotate=360/32+12*360/16] [mypurple,thick,domain=-0.8224:0.8224,smooth] plot ({(0.45*\x)}, {(\x)^2+2});
    
    \draw [rotate=360/32+13*360/16] [mypurple,thick,domain=-0.8224:0.8224,smooth] plot ({(0.45*\x)}, {(\x)^2+2});
    
    \draw [rotate=360/32+14*360/16] [mypurple,thick,domain=-0.8224:0.8224,smooth] plot ({(0.45*\x)}, {(\x)^2+2});
    
    \draw [rotate=360/32+15*360/16] [mypurple,thick,domain=-0.8224:0.8224,smooth] plot ({(0.45*\x)}, {(\x)^2+2});
    

    
    
    
    
    
    \draw [rotate=-0*360/16-360/32] [green,thick,domain=0.75:2,smooth] plot ({0}, {\x});
    
    \draw [rotate=-7*360/16-360/32] [green,thick,domain=0.75:2,smooth] plot ({0}, {\x});
    
    \draw [rotate=-11*360/16-360/32] [green,thick,domain=0.75:2,smooth] plot ({0}, {\x});
    
    
    
    \draw [rotate=-3*360/16] [mygreen,thick,domain=-1:1,smooth] plot ({(0.38*\x)}, {1.21*(\x)^2+0.75});
    
    \draw [rotate=-7*360/16] [mygreen,thick,domain=-1:1,smooth] plot ({(0.38*\x)}, {1.21*(\x)^2+0.75});
    
    \draw [rotate=-11*360/16] [mygreen,thick,domain=-1:1,smooth] plot ({(0.38*\x)}, {1.21*(\x)^2+0.75});
    
    \draw [rotate=-13*360/16] [mygreen,thick,domain=-1:1,smooth] plot ({(0.38*\x)}, {1.21*(\x)^2+0.75});

    
    \draw [rotate=-2*360/16-360/32] [myyellow,thick,domain=-1:1,smooth] plot ({(0.76*\x)}, {1.1*(\x)^2+0.75});
    
    \draw [rotate=-6*360/16-360/32] [myyellow,thick,domain=-1:1,smooth] plot ({(0.76*\x)}, {1.1*(\x)^2+0.75});
    
    \draw [rotate=-10*360/16-360/32] [myyellow,thick,domain=-1:1,smooth] plot ({(0.76*\x)}, {1.1*(\x)^2+0.75});
    
    \draw [rotate=-11*360/16-360/32] [myyellow,thick,domain=-1:1,smooth] plot ({(0.76*\x)}, {1.1*(\x)^2+0.75});
    
    \draw [rotate=-13*360/16-360/32] [myyellow,thick,domain=-1:1,smooth] plot ({(0.76*\x)}, {1.1*(\x)^2+0.75});
    
    \draw [rotate=-14*360/16-360/32] [myyellow,thick,domain=-1:1,smooth] plot ({(0.76*\x)}, {1.1*(\x)^2+0.75});

    
    \draw [rotate=-5*360/16] [myorange,thick,domain=-1:1,smooth] plot ({(1.1*\x)}, {0.92*(\x)^2+0.75});

     \draw [rotate=-10*360/16] [myorange,thick,domain=-1:1,smooth] plot ({(1.1*\x)}, {0.92*(\x)^2+0.75});
    
    \draw [rotate=-12*360/16] [myorange,thick,domain=-1:1,smooth] plot ({(1.1*\x)}, {0.92*(\x)^2+0.75});
    
    \draw [rotate=-13*360/16] [myorange,thick,domain=-1:1,smooth] plot ({(1.1*\x)}, {0.92*(\x)^2+0.75});
    
    \draw [rotate=-14*360/16] [myorange,thick,domain=-1:1,smooth] plot ({(1.1*\x)}, {0.92*(\x)^2+0.75});
    
    
    \draw [rotate=-8*360/16] [myred,thick,domain=-1:1,smooth] plot ({(1.67*\x)}, {0.36*(\x)^2+0.75});
    
    \draw [rotate=-10*360/16] [myred,thick,domain=-1:1,smooth] plot ({(1.67*\x)}, {0.36*(\x)^2+0.75});
    
    \draw [rotate=-11*360/16] [myred,thick,domain=-1:1,smooth] plot ({(1.67*\x)}, {0.36*(\x)^2+0.75});
    
    \draw [rotate=-13*360/16] [myred,thick,domain=-1:1,smooth] plot ({(1.67*\x)}, {0.36*(\x)^2+0.75});

    
    \draw [black,ultra thick,domain=10+90:350+90,smooth] plot ({2*cos(\x)}, {2*sin(\x)});

   \draw [black,ultra thick,domain=0:360,smooth] plot ({0.75*cos(\x)}, {0.75*sin(\x)});

   \draw[rotate=-1*360/16,blue,ultra thick,solid] (-0.183,2.68) -- (0.183,2.68);
   
   \draw[rotate=-2*360/16,blue,ultra thick,solid] (-0.183,2.68) -- (0.183,2.68);
   
   \draw[rotate=-3*360/16,blue,ultra thick,solid] (-0.183,2.68) -- (0.183,2.68);
   
   
   \draw[rotate=-5*360/16,blue,ultra thick,solid] (-0.183,2.68) -- (0.183,2.68);
   
   \draw[rotate=-6*360/16,blue,ultra thick,solid] (-0.183,2.68) -- (0.183,2.68);
   
   \draw[rotate=-7*360/16,blue,ultra thick,solid] (-0.183,2.68) -- (0.183,2.68);
   
   
   \draw[rotate=-9*360/16,blue,ultra thick,solid] (-0.183,2.68) -- (0.183,2.68);
   
   \draw[rotate=-10*360/16,blue,ultra thick,solid] (-0.183,2.68) -- (0.183,2.68);
   
   \draw[rotate=-11*360/16,blue,ultra thick,solid] (-0.183,2.68) -- (0.183,2.68);
   
   
   \draw[rotate=-13*360/16,blue,ultra thick,solid] (-0.183,2.68) -- (0.183,2.68);
   
   \draw[rotate=-14*360/16,blue,ultra thick,solid] (-0.183,2.68) -- (0.183,2.68);
   
   \draw[rotate=-15*360/16,blue,ultra thick,solid] (-0.183,2.68) -- (0.183,2.68);


   \draw[] (0,0) node {\footnotesize $e$};
   
   \draw[] (0,2) node {\footnotesize $s$};
   
   \draw[] (0,-4) node {(b) Noise Process};

\end{tikzpicture}

\end{minipage}
\begin{minipage}{0.3\linewidth}
        \centering
        \vspace{0.5cm}
\begin{tikzpicture}[scale=0.75]


    \draw[] (0,3) node {\scriptsize $t_{16}$\ $t_0$};
    
    \draw[] (1.15,2.772) node {\scriptsize $t_1$};
    
    \draw[] (2.12,2.12) node {\scriptsize $t_2$};
    
    \draw[] (2.772,1.15) node {\scriptsize $t_3$};
    
    \draw[] (3,0) node {\scriptsize $t_4$};
    
    \draw[] (2.772,-1.15) node {\scriptsize $t_5$};
    
    \draw[] (2.12,-2.12) node {\scriptsize $t_6$};
    
    \draw[] (1.15,-2.772) node {\scriptsize $t_7$};
    
    \draw[] (0,-3) node {\scriptsize $t_8$};
    
    
    \draw[] (-1.15,2.772) node {\scriptsize $t_{15}$};
    
    \draw[] (-2.12,2.12) node {\scriptsize $t_{14}$};
    
    \draw[] (-2.772,1.15) node {\scriptsize $t_{13}$};
    
    \draw[] (-3,0) node {\scriptsize $t_{12}$};
    
    \draw[] (-2.772,-1.15) node {\scriptsize $t_{11}$};
    
    \draw[] (-2.12,-2.12) node {\scriptsize $t_{10}$};
    
    \draw[] (-1.15,-2.772) node {\scriptsize $t_9$};


    \draw [rotate=360/32] [mypurple,thick,domain=-0.8224:0.8224,smooth] plot ({(0.45*\x)}, {(\x)^2+2});
    
    \draw [rotate=360/32+360/16] [mypurple,thick,domain=-0.8224:0.8224,smooth] plot ({(0.45*\x)}, {(\x)^2+2});
    
    \draw [rotate=360/32+2*360/16] [mypurple,thick,domain=-0.8224:0.8224,smooth] plot ({(0.45*\x)}, {(\x)^2+2});
    
    \draw [rotate=360/32+3*360/16] [mypurple,thick,domain=-0.8224:0.8224,smooth] plot ({(0.45*\x)}, {(\x)^2+2});
    
    \draw [rotate=360/32+4*360/16] [mypurple,thick,domain=-0.8224:0.8224,smooth] plot ({(0.45*\x)}, {(\x)^2+2});
    
    \draw [rotate=360/32+5*360/16] [mypurple,thick,domain=-0.8224:0.8224,smooth] plot ({(0.45*\x)}, {(\x)^2+2});
    
    \draw [rotate=360/32+6*360/16] [mypurple,thick,domain=-0.8224:0.8224,smooth] plot ({(0.45*\x)}, {(\x)^2+2});
    
    \draw [rotate=360/32+7*360/16] [mypurple,thick,domain=-0.8224:0.8224,smooth] plot ({(0.45*\x)}, {(\x)^2+2});
    
    \draw [rotate=360/32+8*360/16] [mypurple,thick,domain=-0.8224:0.8224,smooth] plot ({(0.45*\x)}, {(\x)^2+2});
    
    \draw [rotate=360/32+9*360/16] [mypurple,thick,domain=-0.8224:0.8224,smooth] plot ({(0.45*\x)}, {(\x)^2+2});
    
    \draw [rotate=360/32+10*360/16] [mypurple,thick,domain=-0.8224:0.8224,smooth] plot ({(0.45*\x)}, {(\x)^2+2});
    
    \draw [rotate=360/32+11*360/16] [mypurple,thick,domain=-0.8224:0.8224,smooth] plot ({(0.45*\x)}, {(\x)^2+2});
    
    \draw [rotate=360/32+12*360/16] [mypurple,thick,domain=-0.8224:0.8224,smooth] plot ({(0.45*\x)}, {(\x)^2+2});
    
    \draw [rotate=360/32+13*360/16] [mypurple,thick,domain=-0.8224:0.8224,smooth] plot ({(0.45*\x)}, {(\x)^2+2});
    
    \draw [rotate=360/32+14*360/16] [mypurple,thick,domain=-0.8224:0.8224,smooth] plot ({(0.45*\x)}, {(\x)^2+2});
    
    \draw [rotate=360/32+15*360/16] [mypurple,thick,domain=-0.8224:0.8224,smooth] plot ({(0.45*\x)}, {(\x)^2+2});
    

    
    
    
    
    
    
    \draw [rotate=-7*360/16-360/32] [green,thick,domain=0.75:2,smooth] plot ({0}, {\x});

    \draw [black,ultra thick,domain=10+90:350+90,smooth] plot ({2*cos(\x)}, {2*sin(\x)});

   \draw [black,ultra thick,domain=0:360,smooth] plot ({0.75*cos(\x)}, {0.75*sin(\x)});


   
   \draw[rotate=-1*360/16,mybrown,fill=mybrownfill,ultra thick,solid,rounded corners=2] (-0.2,2.5-0.2) rectangle (0.2,2.5+0.2);
   
   \draw[rotate=-2*360/16,mybrown,fill=mybrownfill,ultra thick,solid,rounded corners=2] (-0.2,2.5-0.2) rectangle (0.2,2.5+0.2);
   
   \draw[rotate=-3*360/16,mybrown,fill=mybrownfill,ultra thick,solid,rounded corners=2] (-0.2,2.5-0.2) rectangle (0.2,2.5+0.2);
   
   
   \draw[rotate=-5*360/16,mybrown,fill=mybrownfill,ultra thick,solid,rounded corners=2] (-0.2,2.5-0.2) rectangle (0.2,2.5+0.2);
   
   \draw[rotate=-6*360/16,mybrown,fill=mybrownfill,ultra thick,solid,rounded corners=2] (-0.2,2.5-0.2) rectangle (0.2,2.5+0.2);
   
   \draw[rotate=-7*360/16,mybrown,fill=mybrownfill,ultra thick,solid,rounded corners=2] (-0.2,2.5-0.2) rectangle (0.2,2.5+0.2);
   
   
   \draw[rotate=-9*360/16,mybrown,fill=mybrownfill,ultra thick,solid,rounded corners=2] (-0.2,2.5-0.2) rectangle (0.2,2.5+0.2);
   
   \draw[rotate=-10*360/16,mybrown,fill=mybrownfill,ultra thick,solid,rounded corners=2] (-0.2,2.5-0.2) rectangle (0.2,2.5+0.2);
   
   \draw[rotate=-11*360/16,mybrown,fill=mybrownfill,ultra thick,solid,rounded corners=2] (-0.2,2.5-0.2) rectangle (0.2,2.5+0.2);
   
   
   \draw[rotate=-13*360/16,mybrown,fill=mybrownfill,ultra thick,solid,rounded corners=2] (-0.2,2.5-0.2) rectangle (0.2,2.5+0.2);
   
   \draw[rotate=-14*360/16,mybrown,fill=mybrownfill,ultra thick,solid,rounded corners=2] (-0.2,2.5-0.2) rectangle (0.2,2.5+0.2);
   
   \draw[rotate=-15*360/16,mybrown,fill=mybrownfill,ultra thick,solid,rounded corners=2] (-0.2,2.5-0.2) rectangle (0.2,2.5+0.2);

   
   \draw[] (0,0) node {\footnotesize $e$};
   
   \draw[] (0,2) node {\footnotesize $s$};
   
   \draw[] (0,1.3) node {\color{myred}\huge ?};
   
   \draw[] (0,-4) node {(d) Optimal Dynamical Decoupling};
   
\end{tikzpicture}

    \end{minipage}
    \begin{minipage}{0.3\linewidth}
        \centering
        \vspace{-3.5cm}
        \begin{tikzpicture}[scale=0.75]


    \draw[] (0,3) node {\scriptsize $t_{16}$\ $t_0$};
    
    \draw[] (1.15,2.772) node {\scriptsize $t_1$};
    
    \draw[] (2.12,2.12) node {\scriptsize $t_2$};
    
    \draw[] (2.772,1.15) node {\scriptsize $t_3$};
    
    \draw[] (3,0) node {\scriptsize $t_4$};
    
    \draw[] (2.772,-1.15) node {\scriptsize $t_5$};
    
    \draw[] (2.12,-2.12) node {\scriptsize $t_6$};
    
    \draw[] (1.15,-2.772) node {\scriptsize $t_7$};
    
    \draw[] (0,-3) node {\scriptsize $t_8$};
    
    
    \draw[] (-1.15,2.772) node {\scriptsize $t_{15}$};
    
    \draw[] (-2.12,2.12) node {\scriptsize $t_{14}$};
    
    \draw[] (-2.772,1.15) node {\scriptsize $t_{13}$};
    
    \draw[] (-3,0) node {\scriptsize $t_{12}$};
    
    \draw[] (-2.772,-1.15) node {\scriptsize $t_{11}$};
    
    \draw[] (-2.12,-2.12) node {\scriptsize $t_{10}$};
    
    \draw[] (-1.15,-2.772) node {\scriptsize $t_9$};


    \draw [rotate=360/32] [mypurple,thick,domain=-0.8224:0.8224,smooth] plot ({(0.45*\x)}, {(\x)^2+2});
    
    \draw [rotate=360/32+360/16] [mypurple,thick,domain=-0.8224:0.8224,smooth] plot ({(0.45*\x)}, {(\x)^2+2});
    
    \draw [rotate=360/32+2*360/16] [mypurple,thick,domain=-0.8224:0.8224,smooth] plot ({(0.45*\x)}, {(\x)^2+2});
    
    \draw [rotate=360/32+3*360/16] [mypurple,thick,domain=-0.8224:0.8224,smooth] plot ({(0.45*\x)}, {(\x)^2+2});
    
    \draw [rotate=360/32+4*360/16] [mypurple,thick,domain=-0.8224:0.8224,smooth] plot ({(0.45*\x)}, {(\x)^2+2});
    
    \draw [rotate=360/32+5*360/16] [mypurple,thick,domain=-0.8224:0.8224,smooth] plot ({(0.45*\x)}, {(\x)^2+2});
    
    \draw [rotate=360/32+6*360/16] [mypurple,thick,domain=-0.8224:0.8224,smooth] plot ({(0.45*\x)}, {(\x)^2+2});
    
    \draw [rotate=360/32+7*360/16] [mypurple,thick,domain=-0.8224:0.8224,smooth] plot ({(0.45*\x)}, {(\x)^2+2});
    
    \draw [rotate=360/32+8*360/16] [mypurple,thick,domain=-0.8224:0.8224,smooth] plot ({(0.45*\x)}, {(\x)^2+2});
    
    \draw [rotate=360/32+9*360/16] [mypurple,thick,domain=-0.8224:0.8224,smooth] plot ({(0.45*\x)}, {(\x)^2+2});
    
    \draw [rotate=360/32+10*360/16] [mypurple,thick,domain=-0.8224:0.8224,smooth] plot ({(0.45*\x)}, {(\x)^2+2});
    
    \draw [rotate=360/32+11*360/16] [mypurple,thick,domain=-0.8224:0.8224,smooth] plot ({(0.45*\x)}, {(\x)^2+2});
    
    \draw [rotate=360/32+12*360/16] [mypurple,thick,domain=-0.8224:0.8224,smooth] plot ({(0.45*\x)}, {(\x)^2+2});
    
    \draw [rotate=360/32+13*360/16] [mypurple,thick,domain=-0.8224:0.8224,smooth] plot ({(0.45*\x)}, {(\x)^2+2});
    
    \draw [rotate=360/32+14*360/16] [mypurple,thick,domain=-0.8224:0.8224,smooth] plot ({(0.45*\x)}, {(\x)^2+2});
    
    \draw [rotate=360/32+15*360/16] [mypurple,thick,domain=-0.8224:0.8224,smooth] plot ({(0.45*\x)}, {(\x)^2+2});
    

    
    
    
    
    
    \draw [rotate=-0*360/16-360/32] [green,thick,domain=0.75:2,smooth] plot ({0}, {\x});
    
    \draw [rotate=-7*360/16-360/32] [green,thick,domain=0.75:2,smooth] plot ({0}, {\x});
    
    \draw [rotate=-11*360/16-360/32] [green,thick,domain=0.75:2,smooth] plot ({0}, {\x});
    
    
    
    
    
    

    
    
    
    
    
    

    
    \draw [rotate=-5*360/16] [myorange,thick,domain=-1:1,smooth] plot ({(1.1*\x)}, {0.92*(\x)^2+0.75});

    
    
    
    
    
    \draw [rotate=-8*360/16] [myred,thick,domain=-1:1,smooth] plot ({(1.67*\x)}, {0.36*(\x)^2+0.75});
    
    
    

    
    \draw [black,ultra thick,domain=10+90:350+90,smooth] plot ({2*cos(\x)}, {2*sin(\x)});

   \draw [black,ultra thick,domain=0:360,smooth] plot ({0.75*cos(\x)}, {0.75*sin(\x)});


   
   \draw[rotate=-1*360/16,myblue,fill=mybluefill,ultra thick,solid,rounded corners=2] (-0.2,2.5-0.2) rectangle (0.2,2.5+0.2);
   
   \draw[rotate=-2*360/16,myblue,fill=mybluefill,ultra thick,solid,rounded corners=2] (-0.2,2.5-0.2) rectangle (0.2,2.5+0.2);
   
   \draw[rotate=-3*360/16,myblue,fill=mybluefill,ultra thick,solid,rounded corners=2] (-0.2,2.5-0.2) rectangle (0.2,2.5+0.2);
   
   
   \draw[rotate=-5*360/16,myblue,fill=mybluefill,ultra thick,solid,rounded corners=2] (-0.2,2.5-0.2) rectangle (0.2,2.5+0.2);
   
   \draw[rotate=-6*360/16,myblue,fill=mybluefill,ultra thick,solid,rounded corners=2] (-0.2,2.5-0.2) rectangle (0.2,2.5+0.2);
   
   \draw[rotate=-7*360/16,myblue,fill=mybluefill,ultra thick,solid,rounded corners=2] (-0.2,2.5-0.2) rectangle (0.2,2.5+0.2);
   
   
   \draw[rotate=-9*360/16,myblue,fill=mybluefill,ultra thick,solid,rounded corners=2] (-0.2,2.5-0.2) rectangle (0.2,2.5+0.2);
   
   \draw[rotate=-10*360/16,myblue,fill=mybluefill,ultra thick,solid,rounded corners=2] (-0.2,2.5-0.2) rectangle (0.2,2.5+0.2);
   
   \draw[rotate=-11*360/16,myblue,fill=mybluefill,ultra thick,solid,rounded corners=2] (-0.2,2.5-0.2) rectangle (0.2,2.5+0.2);
   
   
   \draw[rotate=-13*360/16,myblue,fill=mybluefill,ultra thick,solid,rounded corners=2] (-0.2,2.5-0.2) rectangle (0.2,2.5+0.2);
   
   \draw[rotate=-14*360/16,myblue,fill=mybluefill,ultra thick,solid,rounded corners=2] (-0.2,2.5-0.2) rectangle (0.2,2.5+0.2);
   
   \draw[rotate=-15*360/16,myblue,fill=mybluefill,ultra thick,solid,rounded corners=2] (-0.2,2.5-0.2) rectangle (0.2,2.5+0.2);

   
   \draw[] (0,0) node {\footnotesize $e$};
   
   \draw[] (0,2) node {\footnotesize $s$};
   
   \draw[] (0,-4) node {(c) Dynamical Decoupling};
   
\end{tikzpicture}

    \end{minipage}

\caption{{\bf(a)} A scenario where an experimenter has a noise process $\mathbf{T}_{\hat{n}}$ (red), which they can interact with at intermediate times $\hat{n}=\{ t_1 , \dots , t_n \}$. The experimenter's goal is to use these intermediate interventions to maximise the mutual information $I$ they obtain after temporally coarse-graining their process tensor to the corresponding channel. In the case where the experimenter chooses not to act, this channel is $\mathbf{T}_{\emptyset}:=\llbracket \mathbf{T}_{\hat{n}} | \mathbf{I}_{\hat{n}} \rrbracket$. However, a DD sequence $\{ \mathcal{I},\mathcal{X},\mathcal{Z},\mathcal{X},\mathcal{Z} \}$ can be cast as a transformation $\mathbf{T}_{\hat{n}} \mapsto \llbracket \mathbf{T}_{\hat{n}} |\mathbf{Z}_{\hat{n}}$, where $\mathbf{Z}_{\hat{n}}=\{ \mathcal{I},\mathcal{I},\mathcal{X},\mathcal{X},\mathcal{Y},\mathcal{Y},\mathcal{Z},\mathcal{Z} \}$ such that $\mathbf{T}'_{\emptyset}:=\llbracket \mathbf{T}_{\hat{n}} |\mathbf{Z}_{\hat{n}} | \mathbf{I}_{\hat{n}} \rrbracket$ might have greater mutual information than $\mathbf{T}_{\emptyset}$. For any general superprocess to provide $I(\mathbf{T}_{\emptyset}) < I(\mathbf{T}'_{\emptyset})$, coarse-grainings must be irreversible in the resource theory the experimenter operates under, as described in Obs.~\ref{obs:irreversibility}. {\bf(b)} If the system $s$ is left to freely interact with the environment $e$, correlations build (visually depicted by coloured arcs between $s$ and $e$ at different times), causing decoherence to occur at an accelerating rate. {\bf(c)} Dynamical decoupling shifts multitime correlations mediated by the environment, to system-level correlations, preventing a build-up of system-environment correlations. {\bf(d)} Multitimescale optimal dynamical decoupling (see Sec.\ref{sec:optimisation}) performs this conversion more efficiently, resulting in better noise reduction over longer timescales, raising the question of how far this can be taken. Here, we identify and quantify yet untapped resources that can be employed for further improvement.}  \label{fig:definingobjects}
\end{figure*}
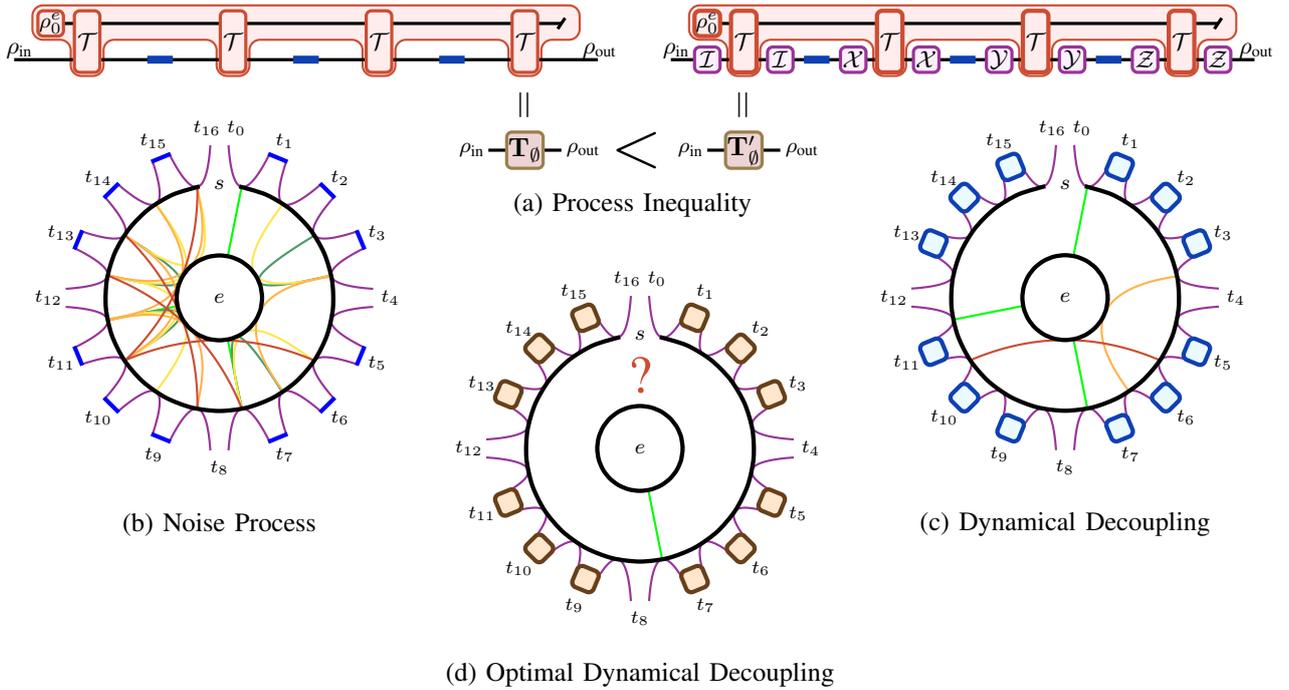

\subsection{Resource Theories for Quantum Processes} \label{sec:resourcetheoriesofmultitimeprocesses}

With this in mind, we now embark on a slightly longer path -- separating noise suppression techniques into two distinct steps -- to quantify temporal correlations in the language of RTQPs~\cite{resourcetheoriesofmultitime}. Here, noise processes, represented by process tensors $\mathbf{T}_{\hat{n}}$, are the resource objects. Experimental control can be cast as a resource-transformation $\mathbf{Z}_{\hat{n}}$, mapping $\mathbf{T}_{\hat{n}}$ to another process\footnote[2]{Observe that this picture is equivalent to Eq.~\eqref{eq:PTcontrol} by absorbing non-trivial control into the process tensor via the superprocess $\mathbf{T}'_{\emptyset} = \llbracket \mathbf{T}_{\hat{n}} | \mathbf{A}_{\hat{n}} \rrbracket  = \llbracket \mathbf{T}_{\hat{n}} | \mathbf{Z}_{\hat{n}} | \mathbf{A}'_{\hat{n}} \rrbracket = \llbracket \mathbf{T}'_{\hat{n}} | \mathbf{A}'_{\hat{n}} \rrbracket$, where $\mathbf{A}'_{\hat{n}}$ is taken to be a trivial `do-nothing' sequence.} $\mathbf{T}'_{\hat{n}}$:
\begin{gather} \label{eq:superprocess}
    \mathbf{T}'_{\hat{n}} = \llbracket \mathbf{T}_{\hat{n}} | \mathbf{Z}_{\hat{n}}.
\end{gather}

Above, the transformation $\mathbf{Z}_{\hat{n}}$ is called a superprocess~\cite{resourcetheoriesofmultitime} and consists of pre- and post- processing operations to the $s$ part of each evolution map in Eq.~\eqref{eq:processtensor}:
\begin{gather}
    \mathcal{T}^{se}_{t_j:t_{j-1}} \to \mathcal{T}'^{se}_{t_j:t_{j-1}}:= \mathcal{V}^{sa}_{\alpha} \circ_{sa}
    \mathcal{T}^{se}_{t_j:t_{j-1}} \circ_{sa}
    \mathcal{W}^{sa}_{\alpha}.
\end{gather}
The superprocess may potentially make use of additional ancillary system $a$. The form and connectivity of $\mathcal{V}^{sa}_{\alpha}$ and $\mathcal{W}^{sa}_{\alpha}$ correspond to experimental constraints. For instance, for QZE and DFS there is no $a$, for DD $a$ is restricted to be a classical clock, while for QEC $a$ would correspond to the measuring systems that detect the syndrome. For $n=0$, we have the limiting case of a channel $\mathbf{T}_\emptyset$, and supermap $\mathbf{Z}_\emptyset$~\cite{supermaps}.

Resource theories of quantum processes (RTQP) $\mathsf{S}_{\hat{n}}, \ n \in \mathbbm{W}$ account for the transformations between process tensors, $\mathbf{T}_{\hat{n}} \in \mathsf{T}_{\hat{n}}$, under a family of constrained superprocesses $\mathbf{Z}_{\hat{n}} \in \mathsf{Z}_{{\hat{n}}}$, which constitute the free transformations of the respective theory. This gives rise to a set of monotones for $\mathsf{S}_{\hat{n}}$, which are non-increasing functions $\mathbf{T}_{\hat{n}} \to \mathbbm{R}_{\geq0}$ under the action of $\mathsf{Z}_{{\hat{n}}}$. It also defines a set of free resources~\cite{review} $\mathsf{T}^{\text{F}}_{\hat{n}}$ that can be obtained starting from any process $\mathbf{T}_{\hat{n}}$ via some $\mathbf{Z}_{\hat{n}}$. Ref.~\cite{resourcetheoriesofmultitime} placed restrictions on the connectivity (but not form) of pre- and post- processing to construct a family of RTQPs. In particular, memory was found to be a resourceful quantity in a number of these theories, suggesting that they may be used to examine the role of non-Markovianity in DD. However, it turns out the DD pulse sequence is an isometric superprocess, i.e., the non-Markovianity of $\mathbf{T}'_{\hat{n}}$ and $\mathbf{T}_{\hat{n}} $ are identical. No non-Markovianity is expended when considering DD as a transformation of process tensors. But, as we will show, DD can be quantified in terms of an expenditure of the non-Markovianity under the transformation $\mathbf{T}_{\hat{n}} \to \mathbf{T}'_{\hat{n}}$ \textit{in combination} with a subsequent temporal coarse-graining procedure, which we illustrate in Fig.~\ref{fig:coarsegraining}.

\subsection{Temporal Coarse-Graining} \label{sec:temporalcoarsegraining}

In a quantum computation, ultimately we are not interested in the multitime correlations -- we only care about the aforementioned input-output correlations given in Eq.~\eqref{eq:PTcontrol}. However, the noise process $\mathbf{T}_{\hat{n}}$ in a typical NISQ device is correlated across multiple times. An experimenter will contract this process with the logical gates of the computer program, together with the noise-reduction pulses, with the aim of obtaining a high-fidelity computation. The contraction of noise-reduction pulses is an operationally meaningful notion of temporal coarse-graining, mapping an $n$-intervention process into an $m$-intervention process: $\mathbf{T}_{\hat{n}} \to \mathbf{T}_{\hat{m}}$. Ideally, the latter process should possess no multitime correlations, i.e., it should be a Markovian process $\mathbf{T}_{\hat{m}} =  \mathbf{T}_\emptyset^{\otimes m}$, with each $\mathbf{T}_\emptyset$ being as close to a unitary process as possible. This way, by contracting the remaining slots of the process with $m$ logical gates, we can perform the desired computation.

No resource theory where noiselessness is a resource allows for converting a noisy channel into a noiseless one under free transformations. Yet, noise suppression can be described as a free transformation when the multitime correlated process $\mathbf{T}_{\hat{n}}$ is converted into a single-time noiseless channel $\mathbf{T}_{\emptyset}$. The crucial point to note is that temporal coarse-graining, in general, is resource decreasing. Thus, having a priori access to only an $m$-intervention process is \textit{not} the same as having access to an $n$-intervention process (with $n>m$), which is then coarse-grained to $m$-interventions. In particular, the two are the same when a \textit{trivial} coarse-graining procedure is implemented, i.e., when intermediate times are closed off by using a collection of $n-m$ identity maps $\mathcal{I}^{s}_{i}$ on $s$
\begin{equation} \label{eq:coarsegrainingfunctor}
    \mathbf{I}_{{\hat{n}} \setminus {\hat{m}}}:=\bigotimes_{i \in {\hat{n}}\setminus {\hat{m}}} \mathcal{I}^{s}_{i} \quad \mbox{for} \quad \emptyset \subseteq \hat{m} \subseteq \hat{n},
\end{equation}
yielding the (trivially coarse-grained) process $ \mathbf{T}_{\hat{m}} = \llbracket \mathbf{T}_{\hat{n}} |\mathbf{I}_{{\hat{n}} \setminus {\hat{m}}} \rrbracket$. Thus our goal when exploiting the resources in $\mathbf{T}_{\hat{n}}$, will be to search for the optimal overall control sequence -- containing both the superprocess and coarse-graining. Given that all non-trivial allowed control can be delegated to the free superprocesses, all possible experimental control can be represented as
\begin{gather}\label{eq:moreoptimal}
 \mathbf{T}'_{\hat{m}}  = \llbracket \mathbf{T}_{\hat{n}} |\mathbf{Z}_{\hat{n}} |\mathbf{I}_{{\hat{n}} \setminus {\hat{m}}} \rrbracket,
\end{gather}
which amounts to trivial coarse-graining when $\mathbf{Z}_{\hat{n}}$ is the trivial, do-nothing superprocess. The above leads us to our first key observation.

\begin{observation}[Irreversibility of temporal coarse-graining] \label{obs:irreversibility} 
A process $\mathbf{T}'_{\hat{m}}$, that is transformed by a superprocess and then coarse-grained, has a larger range than a process $\mathbf{T}_{\hat{m}}$, that is coarse-grained and then transformed, for all non-trivial superprocesses $\mathbf{Z}_{\hat{n}}$ and $\mathbf{Z}_{\hat{m}}$, such that $\hat{n} \supseteq \hat{m}$.
\end{observation}
The proof of this is in the methods Sec.~\ref{sec:methods}, and also see Fig.~\ref{fig:irreversibility}. In other words, a fine-grained process $\mathbf{T}_{\hat{n}}$ can be transformed to a larger set of coarse-grained processes $\mathbf{T}'_{\hat{m}}$ than those reachable by transforming a trivially coarse-grained process $\mathbf{T}_{\hat{m}}$. This \emph{irreversibility} of coarse-graining implies that $\mathbf{T}'_{\hat{m}}$ will have a higher monotone value than $\mathbf{T}_{\hat{m}}$. When specifically concerned with the task of information preservation, a stronger version of this statement holds (Thm.~\ref{thm:nonmonotonicity} in Sec.~\ref{sec:irreversibilitynonmonotonicity}): a coarse-grained experimenter (i.e., one that only has access to times in $\hat m$) can perceive the free actions of a fine-grained experimenter to be resource-increasing if and only if coarse-graining a process \emph{strictly reduces} its mutual information. Due to the scaling of mutual information with the dimension of its argument (which coarse-graining reduces), Thm.~\ref{thm:nonmonotonicity} implies that temporal coarse-graining will almost always be able to produce non-monotonicity -- hence enabling information preservation. In other words, temporal resolution is almost always a resource for information preservation. 

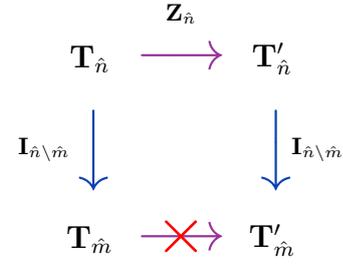
\begin{figure}[ht!]
\centering
\begin{tikzpicture}[scale=1.2]

\draw[] (0,2) node {\large $\mathbf{T}_{\hat{n}}$};
\draw[] (2,2) node {\large $\mathbf{T}'_{\hat{n}}  $};
\draw[] (0,0) node {\large $\mathbf{T}_{\hat{m}} $};
\draw[] (2,0) node {\large $\mathbf{T}'_{\hat{m}} $};

\draw[color=mypurple] (1,2) node {\huge $\longrightarrow$};
\draw[color=myblue] (0,1) node[rotate=270] {\huge $\longrightarrow$};
\draw[color=myblue] (2,1) node[rotate=270] {\huge $\longrightarrow$};
\draw[color=mypurple] (1,0) node {\huge $\longrightarrow$};
\draw[color=red] (1,0.04) node {\Huge $\times$};

\draw[] (1,2.5) node {\small $\mathbf{Z}_{\hat{n}}$};
\draw[] (-0.5,1) node {\small $\mathbf{I}_{\hat{n} \setminus \hat{m}}$};
\draw[] (2.5,1) node {\small $\mathbf{I}_{\hat{n} \setminus \hat{m}}$};

\end{tikzpicture}
\caption{Irreversibility as in Obs.~\ref{obs:irreversibility} states that for a process tensor $\mathbf{T}_{\hat{n}}$, applying a superprocess $\mathbf{Z}_{\hat{n}}$ followed by coarse-graining $\mathbf{I}_{\hat{n} \setminus \hat{m}}$, will have an image larger than that of coarse-graining followed by a superprocess. Any corresponding monotone in the former case can be made at least as large as in the latter, see Cor.~\ref{cor:irreversibilitymonotone}. This is necessary for non-monotonicity to occur from the perspective of the coarse-grained experimenter, as required by DD (see Sec.~\ref{sec:irreversibilitynonmonotonicity}). Throughout this work, an `increase' in a quantity of interest is measured as seen from a coarse perspective by comparing $\mathbf{T}'_{\hat{m}}$ to $\mathbf{T}_{\hat{m}}$. If one was to instead attempt comparing $\mathbf{T}'_{\hat{m}}$ or $\mathbf{T}'_{\hat{n}}$ to the original $\mathbf{T}_{\hat{n}}$, monotonicity of free transformations would ensure that no increases will ever be obtained, meaning that noise reduction would not be observed.}\label{fig:irreversibility}
\end{figure}

\begin{figure*}[ht!] 
        \centering
        \begin{tikzpicture}[scale=0.275]
\draw[ mediumgrey, thin,dashed] (0+0.0,-3) -- (0+0.0,-7);
\draw[ mediumgrey, thin,dashed] (2+0.0,-3) -- (2+0.0,-7);
\draw[ mediumgrey, thin,dashed] (0+8.0,-3) -- (0+8.0,-7);
\draw[ mediumgrey, thin,dashed] (2+8.0,-3) -- (2+8.0,-7);
\draw[ mediumgrey, thin,dashed] (0+16.0,-3) -- (0+16.0,-7);
\draw[ mediumgrey, thin,dashed] (2+16.0,-3) -- (2+16.0,-7);
\draw[ mediumgrey, thin,dashed] (0+24.0,-3) -- (0+24.0,-7);
\draw[ mediumgrey, thin,dashed] (2+24.0,-3) -- (2+24.0,-7);

\draw[myred,fill=myredfill, thick,solid,rounded corners=4] (-2,3-4) -- (-2,4-4) -- (28,4-4) -- (28,2+0.1-4) -- (26-0.1,2+0.1-4) -- (26-0.1,0+0.1-4) -- (24+0.1,0+0.1-4) -- (24+0.1,2+0.1-4)-- (18-0.1,2+0.1-4) -- (18-0.1,0+0.1-4) -- (16+0.1,0+0.1-4) -- (16+0.1,2+0.1-4) -- (10-0.1,2+0.1-4) -- (10-0.1,0+0.1-4) -- (8+0.1,0+0.1-4) -- (8+0.1,2+0.1-4) -- (2-0.1,2+0.1-4) -- (2-0.1,0+0.1-4) -- (0+0.1,0+0.1-4) -- (0+0.1,2+0.1-4) -- (-2,2+0.1-4) -- (-2,3-4)   ;

\draw[black, very thick,solid] (-1,3-4) -- (27,3-4);
\draw[black, very thick,solid] (26.8,2.7-4) -- (27.2,3.3-4);
\draw[black, very thick,solid] (-1.5,1-4) -- (27.5,1-4);


\draw[myred,fill=myredfill,very thick,solid,rounded corners=2] (-1.7,3.7-4) rectangle (-0.3,2.3-4);
\draw[myred,fill=myredfill,very thick,solid,rounded corners=2] (0.3,3.7-4) rectangle (1.7,0.3-4);
\draw[myred,fill=myredfill,very thick,solid,rounded corners=2] (8.3,3.7-4) rectangle (9.7,0.3-4);
\draw[myred,fill=myredfill,very thick,solid,rounded corners=2] (16.3,3.7-4) rectangle (17.7,0.3-4);
\draw[myred,fill=myredfill,very thick,solid,rounded corners=2] (24.3,3.7-4) rectangle (25.7,0.3-4);

\draw[white,fill=white,ultra thick,solid,rounded corners=2] (-1.7,1.7-4) rectangle (-0.3,-1.7-4);
\draw[white,fill=white,ultra thick,solid,rounded corners=2] (4.3-2,1.7-4) rectangle (5.7+2,-1.7-4);
\draw[white,fill=white,ultra thick,solid,rounded corners=2] (12.3-2,1.7-4) rectangle (13.7+2,-1.7-4);
\draw[white,fill=white,ultra thick,solid,rounded corners=2] (20.3-2,1.7-4) rectangle (21.7+2,-1.7-4);
\draw[white,fill=white,ultra thick,solid,rounded corners=2] (26.3,1.7-4) rectangle (27.7,-1.7-4);

\draw[] (-1,3-4) node {\small $\rho^{e}_0$};
\draw[] (1,2-4) node[rotate=0] {\small $\mathcal{T}$};
\draw[] (9,2-4) node[rotate=0] {\small $\mathcal{T}$};
\draw[] (17,2-4) node[rotate=0] {\small $\mathcal{T}$};
\draw[] (25,2-4) node[rotate=0] {\small $\mathcal{T}$};

\draw[] (-1,1-4) node { $0$};
\draw[] (5,1-4) node { $t_1$};
\draw[] (13,1-4) node { $t_2$};
\draw[] (21,1-4) node { $t_3$};
\draw[] (27,1-4) node { $t$};



\draw[] (13,5-4-0.25) node {\small Noise Process $\mathbf{T}_{\hat{n}}$};


\draw[mypurple,fill=mypurplefill,thick,solid,rounded corners=4] (-2,1-0.1-8) -- (-2,-2-8)

-- (4-0.1,-2-8) -- (4-0.1,0+0.1-8) -- (6+0.1,0+0.1-8) -- (6+0.1,-2-8) -- (12-0.1,-2-8) -- (12-0.1,0+0.1-8) -- (14+0.1,0+0.1-8) -- (14+0.1,-2-8) 

-- (20-0.1,-2-8)  -- (20-0.1,0-0.1-8)  -- (22+0.1,0-0.1-8) -- (22+0.1,-2-8)  -- (28,-2-8) -- (28,2-0.1-8)
-- (26+0.1,2-0.1-8) -- (26+0.1,0-0.1-8) -- (24-0.1,0-0.1-8) -- (24-0.1,2-0.1-8) 

-- (18+0.1,2-0.1-8) -- (18+0.1,0-0.1-8) -- (16-0.1,0-0.1-8) -- (16-0.1,2-0.1-8)

-- (10+0.1,2-0.1-8) -- (10+0.1,0-0.1-8) -- (8-0.1,0-0.1-8) -- (8-0.1,2-0.1-8)

-- (2+0.1,2-0.1-8) -- (2+0.1,0-0.1-8) -- (0-0.1,0-0.1-8) -- (0-0.1,2-0.1-8) -- (-2,2-0.1-8) -- (-2,1-0.1-8)  ;

\draw[black, very thick,solid] (-2.5,-1-8) -- (28.5,-1-8);
\draw[black, very thick,solid] (-1,1-8) -- (27,1-8);

\draw[white,fill=white,very thick,solid,rounded corners=2] (0.3,3.7-5-2-3) rectangle (1.7,0.3-5-3);
\draw[white,fill=white,very thick,solid,rounded corners=2] (8.3,3.7-5-2-3) rectangle (9.7,0.3-5-3);
\draw[white,fill=white,very thick,solid,rounded corners=2] (16.3,3.7-5-2-3) rectangle (17.7,0.3-5-3);
\draw[white,fill=white,very thick,solid,rounded corners=2] (24.3,3.7-5-2-3) rectangle (25.7,0.3-5-3);


\draw[mypurple,fill=mypurplefill,very thick,solid,rounded corners=2] (-1.7,1.7-8) rectangle (-0.3,-1.7-8);
\draw[mypurple,fill=mypurplefill,very thick,solid,rounded corners=2] (4.3-2,1.7-8) rectangle (5.7-2,-1.7-8);
\draw[mypurple,fill=mypurplefill,very thick,solid,rounded corners=2] (4.3+2,1.7-8) rectangle (5.7+2,-1.7-8);
\draw[mypurple,fill=mypurplefill,very thick,solid,rounded corners=2] (12.3-2,1.7-8) rectangle (13.7-2,-1.7-8);
\draw[mypurple,fill=mypurplefill,very thick,solid,rounded corners=2] (12.3+2,1.7-8) rectangle (13.7+2,-1.7-8);
\draw[mypurple,fill=mypurplefill,very thick,solid,rounded corners=2] (18.3,1.7-8) rectangle (19.7,-1.7-8);
\draw[mypurple,fill=mypurplefill,very thick,solid,rounded corners=2] (20.3+2,1.7-8) rectangle (21.7+2,-1.7-8);
\draw[mypurple,fill=mypurplefill,very thick,solid,rounded corners=2] (26.3,1.7-8) rectangle (27.7,-1.7-8);

\draw[white,fill=white,ultra thick,solid,rounded corners=2] (4.3,0.7-8-1) rectangle (5.7,-0.7-8-1);
\draw[white,fill=white,ultra thick,solid,rounded corners=2] (12.3,0.7-8-1) rectangle (13.7,-0.7-8-1);
\draw[white,fill=white,ultra thick,solid,rounded corners=2] (20.3,0.7-8-1) rectangle (21.7,-0.7-8-1);

\draw[] (-1,0-8) node[rotate=0] {\small $\mathcal{V}$};
\draw[] (5-2,0-8) node[rotate=0] {\small $\mathcal{W}$};
\draw[] (5+2,0-8) node[rotate=0] {\small $\mathcal{V}$};
\draw[] (13-2,0-8) node[rotate=0] {\small $\mathcal{W}$};
\draw[] (13+2,0-8) node[rotate=0] {\small $\mathcal{V}$};
\draw[] (21-2,0-8) node[rotate=0] {\small $\mathcal{W}$};
\draw[] (21+2,0-8) node[rotate=0] {\small $\mathcal{V}$};
\draw[] (27,0-8) node[rotate=0] {\small $\mathcal{W}$};







\draw[] (13,-5-0.25) node {\small Superprocess $\mathbf{Z}_{\hat{n}}$ acts on $\mathbf{T}_{\hat{n}}$};


\draw[black, very thick,solid] (-1,3-18+2) -- (27,3-18+2);
\draw[black, very thick,solid] (26.8,2.7-18+2) -- (27.2,3.3-18+2);
\draw[black, very thick,solid] (-2,1-18+2) -- (28,1-18+2);

\draw[myred,fill=myredfill, thick,solid,rounded corners=4] (-2,3-18+2) -- (-2,4-18+2) -- (28,4-18+2) -- (28,2+0.1-18+2) -- (26-0.1,2+0.1-18+2) -- (26-0.1,0+0.1-18+2) -- (24+0.1,0+0.1-18+2) -- (24+0.1,2+0.1-18+2)-- (18-0.1,2+0.1-18+2) -- (18-0.1,0+0.1-18+2) -- (16+0.1,0+0.1-18+2) -- (16+0.1,2+0.1-18+2) -- (10-0.1,2+0.1-18+2) -- (10-0.1,0+0.1-18+2) -- (8+0.1,0+0.1-18+2) -- (8+0.1,2+0.1-18+2) -- (2-0.1,2+0.1-18+2) -- (2-0.1,0+0.1-18+2) -- (0+0.1,0+0.1-18+2) -- (0+0.1,2+0.1-18+2) -- (-2,2+0.1-18+2) -- (-2,3-18+2)   ;



\draw[white,fill=white,ultra thick,solid,rounded corners=2] (4.3,1.7-18+2) rectangle (5.7,-1.7-18+2);
\draw[white,fill=white,ultra thick,solid,rounded corners=2] (12.3,1.7-18+2) rectangle (13.7,-1.7-18+2);
\draw[white,fill=white,ultra thick,solid,rounded corners=2] (20.3,1.7-18+2) rectangle (21.7,-1.7-18+2);





\draw[] (13,3-18+2) node {\small Transformed Process $\mathbf{T}'_{\hat{n}}$};



\draw[ mediumgrey, thin,dashed] (4+0.3,-15) -- (4+0.3,-17);
\draw[ mediumgrey, thin,dashed] (6-0.3,-15) -- (6-0.3,-17);
\draw[ mediumgrey, thin,dashed] (4+0.3+8,-15) -- (4+0.3+8,-17);
\draw[ mediumgrey, thin,dashed] (6-0.3+8,-15) -- (6-0.3+8,-17);
\draw[ mediumgrey, thin,dashed] (4+0.3+16,-15) -- (4+0.3+16,-17);
\draw[ mediumgrey, thin,dashed] (6-0.3+16,-15) -- (6-0.3+16,-17);

\draw[myblue,line width=3pt] (4.3,1-26+8) -- (5.7,1-26+8);

\draw[myblue,line width=3pt] (12.3,1-26+8) -- (13.7,1-26+8);

\draw[myblue,line width=3pt] (20.3,1-26+8) -- (21.7,1-26+8);

\draw[] (5,2-26+8) node {\small $\mathcal{I}$};
\draw[] (13,2-26+8) node {\small $\mathcal{I}$};
\draw[] (21,2-26+8) node {\small $\mathcal{I}$};


\draw[] (13,-18.5) node {\small Coarse-Graining $\mathbf{I}_{\hat{n} \setminus \hat{m}}$ acts on $\mathbf{T}'_{\hat{n}}$};


\draw[black, very thick,solid] (-1,3-26+2) -- (27,3-26+2);
\draw[black, very thick,solid] (26.8,2.7-26+2) -- (27.2,3.3-26+2);
\draw[black, very thick,solid] (-2,1-26+2) -- (28,1-26+2);

\draw[myred,fill=myredfill, thick,solid,rounded corners=4] (-2,3-26+2) -- (-2,4-26+2) -- (28,4-26+2) -- (28,2+0.1-26+2) -- (26-0.1,2+0.1-26+2) -- (26-0.1,0+0.1-26+2) -- (24+0.1,0+0.1-26+2) -- (24+0.1,2+0.1-26+2)-- (18-0.1,2+0.1-26+2) -- (18-0.1,0+0.1-26+2) -- (16+0.1,0+0.1-26+2) -- (16+0.1,2+0.1-26+2) -- (10-0.1,2+0.1-26+2) -- (10-0.1,0+0.1-26+2) -- (8+0.1,0+0.1-26+2) -- (8+0.1,2+0.1-26+2) -- (2-0.1,2+0.1-26+2) -- (2-0.1,0+0.1-26+2) -- (0+0.1,0+0.1-26+2) -- (0+0.1,2+0.1-26+2) -- (-2,2+0.1-26+2) -- (-2,3-26+2)   ;


\draw[myblue,line width=3pt] (4.3,1-26+2) -- (5.7,1-26+2);

\draw[myblue,line width=3pt] (12.3,1-26+2) -- (13.7,1-26+2);

\draw[myblue,line width=3pt] (20.3,1-26+2) -- (21.7,1-26+2);

\draw[] (13,3-26+2) node {\small Resultant Channel $\mathbf{T}'_{\emptyset}$};


\draw[] (29.5,2-4) node {\Large $\Leftrightarrow$};

\draw[] (29.5,2-5-5) node {\Large $\Leftrightarrow$};

\draw[] (29.5,-15) node {\Large $\Leftrightarrow$};

\draw[] (29.5,-22) node {\Large $\Leftrightarrow$};


\draw[] (45,5-4-0.25) node {\small Properties of Noise};

\draw[mygreyfill,fill=mygreyfill,ultra thick,solid,rounded corners=2] (31,4-4) rectangle (59,0-4);


\draw[myorange,fill=myorangefill,rounded corners=10, fill opacity = 1][rotate=0] (41-4-1,2-1.75-4) rectangle (41+4,2+1.75-4);
\draw[mygreen,fill=mygreenfill,rounded corners=10, fill opacity = 0.5][rotate=0] (49-6,2-1.75-4) rectangle (49+6,2+1.75-4);

\draw[myorange,rounded corners=10][rotate=0] (41-4-1,2-1.75-4) rectangle (41+4,2+1.75-4);
\draw[mygreen,rounded corners=10][rotate=0] (49-6,2-1.75-4) rectangle (49+6,2+1.75-4);

\draw[] (33,2-4) node {\small All};
\draw[] (39.5,2-4) node {\small Symmetries};
\draw[] (50,2-4) node {\small Non-Markovianity};
\draw[] (57,2-4) node {\small All};


\draw[mygreyfill,fill=mygreyfill,ultra thick,solid,rounded corners=2] (31,-22+2) rectangle (59,-26+2);

\draw[myorange,fill=myorangefill,rounded corners=10][rotate=0] (38-7,-24-1.75+2) rectangle (38+7,-24+1.75+2) ;
\draw[myred,fill=myredfill,rounded corners=10][rotate=0] (34-2.75,-24-1.5+2) rectangle (34+2.75,-24+1.5+2) ;

\draw[] (34,-24+2) node {\small Classical };
\draw[] (41,-24+2) node {\small\begin{tabular}{c} Quantum \\ Subspace \end{tabular}};
\draw[] (49,-24+2) node {\small Quantum};
\draw[] (56,-24+2) node {\small Quantum};

\draw[] (45,-21-0.25+2) node {\small Information Preserved in Resultant Channel};



\draw[] (45,-5-0.25) node {\small $+$ Superprocess  };

\draw[mypurple,fill=mypurplefill,very thick,solid,rounded corners=2] (31,-4-2) rectangle (59,-8-2);
\draw[mygrey,fill=mygreyfill,rounded corners=10][rotate=0] (45-7,-6-1.75-2) rectangle (45+7,-6+1.75-2);
\draw[myred,fill=myredfill,rounded corners=10][rotate=0] (33-1.5,-6-1.2-2) rectangle (33+1.5,-6+1.2-2);
\draw[myorange,fill=myorangefill,rounded corners=10][rotate=0] (41-2,-6-1.2-2) rectangle (41+2,-6+1.2-2);

\draw[] (33,-6-2) node {$\mathsf{C}_{\hat{\mathbbm{W}}}$};
\draw[] (41,-6-2) node {$\mathsf{P}_{\hat{\mathbbm{W}}}$};
\draw[] (49,-6-2) node {$\mathsf{D}_{\hat{\mathbbm{W}}}$};
\draw[] (57,-6-2) node {$\mathsf{Q}_{\hat{\mathbbm{W}}}$};


\draw[myblue,fill=mybluefill,very thick,solid,rounded corners=2] (31,-13) rectangle (59,-17);

\draw[] (45,-11-0.25-1) node {{\small $+$ Coarse-Graining}};

\draw[] (33,-15) node {{\small\begin{tabular}{c} Zeno \\ Effect \end{tabular}}};
\draw[] (41,-15) node {{\small\begin{tabular}{c} DFS \\ Creation \end{tabular}}};
\draw[] (49,-15) node {{\small\begin{tabular}{c} Dynamical \\ Decoupling \end{tabular}}};
\draw[] (56,-15) node {{\small\begin{tabular}{c} \ \ \ \ Error \\ Correction \end{tabular}}};




\draw[] (13,-25.5) node {(a)};
\draw[] (45,-25.5) node {(b)};

\end{tikzpicture}

\caption{The relationship between resource theories of temporal resolution and noise reduction methods. {\bf(a)} Resource objects are noise processes $\mathbf{T}_{\hat{n}}$, and free transformations can be represented as a superprocess $\mathbf{Z}_{\hat{n}}$, followed by temporal coarse-graining $\mathbf{I}_{{\hat{n}} \setminus {\hat{m}}}$, i.e., $\mathbf{T}_{\hat{n}} \mapsto \mathbf{T}'_{\hat{m}} := \llbracket \mathbf{T}_{\hat{n}} | \mathbf{Z}_{\hat{n}} | \mathbf{I}_{{\hat{n}} \setminus {\hat{m}}} \rrbracket $. The properties of the noise process are encoded in the $\mathcal{T}$ channels within $\mathbf{T}_{\hat{n}}$, while the capabilities of the experimenter are specified by the form and connectivity of the $\mathcal{V}$ (pre) and $\mathcal{W}$ (post) channels of $\mathbf{Z}_{\hat{n}}$. The diagram shows three intermediate interventions $\hat{n}=\{t_1, t_2, t_3 \}$ where the superprocess can be `plugged in', followed by complete coarse-graining, i.e., $\hat{m}=\emptyset$. {\bf(b)} The relationship between the properties of a noise process, the allowed superprocesses which specify the resource theory of temporal resolution, and the type of information a technique achievable within that resource theory can preserve. All techniques listed here can be performed by an experimenter possessing the full abilities of $\mathsf{Q}_{\hat{\mathbbm{W}}}$, but depending on the underlying noise process, a more constrained experimenter (corresponding to a sub-theory of $\mathsf{Q}_{\hat{\mathbbm{W}}}$) may still be capable of preserving some information, as detailed in Sec.~\ref{sec:subtheories}. Operating within the sub-theory $\mathsf{D}_{\hat{\mathbbm{W}}}$, DD harnesses non-Markovian noise to preserve a full quantum state. On the other hand, the creation of a DFS (performed by an experimenter operating in $\mathsf{P}_{\hat{\mathbbm{W}}} \subset \mathsf{D}_{\hat{\mathbbm{W}}}$, see Sec.~\ref{sec:subtheories}) harnesses symmetries in the interaction to preserve a subspace, without requiring non-Markovianity. However, there is nothing to suggest these two effects are mutually exclusive, and might be jointly harnessed. Both QEC and the QZE can in principle be performed for any kind of noise, but the extra ability of $\mathsf{Q}_{\hat{\mathbbm{W}}}$ over  $\mathsf{C}_{\hat{\mathbbm{W}}}$ (see Sec.~\ref{sec:subtheories}) means that QEC can preserve a full quantum state, rather than just classical information.} \label{fig:coarsegraining}
\end{figure*}
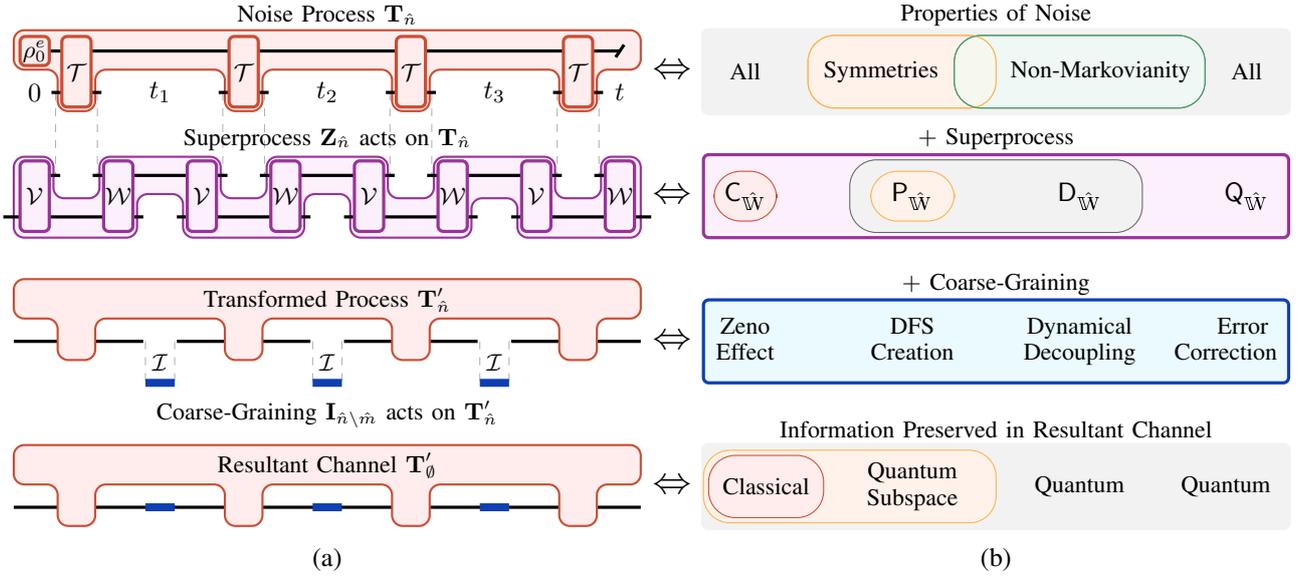

\subsection{Resource Theories of Temporal Resolution} \label{sec:subsecresourcetheoriesoftemporalresolution}

Combining the existing machinery of RTQPs with temporal coarse-graining yields \emph{resource theories of temporal resolution} (RTTR), enabling a unified view of temporal resources, illustrated in Fig.~\ref{fig:coarsegraining}. In a RTTR $\mathsf{S}_{\hat{\mathbbm{W}}}$, the resource objects $\mathbf{T}_{\hat{n}} \in \mathsf{T}_{\hat{\mathbbm{W}}}$ are still process tensors but each process tensor can have any whole number $n \in \mathbbm{W}$ of intermediate times for interventions, rather than one fixed set of times for the whole theory as in a RTQP. The generalised Kolmogorov extension theorem~\cite{kolmogorov} guarantees that such a fine grained description will always exist (although it need not be unique).

The set of free transformations are of the form $\mathbf{Z}_{\hat{n}} |  \mathbf{I}_{\hat{n} \setminus \hat{m}} \rrbracket$ for any $n,m \in \mathbbm{W}$ (see Lem.~\ref{lem:representation}), where
$\mathbf{Z}_{\hat{n}} \in \mathsf{Z}_{\hat{\mathbbm{W}}}$ are the free superprocesses derived from the corresponding RTQP for fixed ${\hat{n}}$. As discussed in Sec.~\ref{sec:resourcetheoriesofmultitimeprocesses}, the form and connectivity of allowed pre- and post- operations $\mathcal{V}^{sa}_{\alpha}$ and $\mathcal{W}^{sa}_{\alpha}$ within these superprocess specify the experimental constraints that define the resource theory. Due to the inclusions of temporal coarse-graining as a free transformation, the only free processes $\mathbf{T}^{\text{F}} \in \mathsf{T}^{\text{F}}$ (defined as those which can be reached from any other~\cite{review}) are zero capacity \emph{channels} $\mathbf{T}^{\text{F}} = \mathbbm{1} \otimes \beta$, where $\mathbbm{1}$ is the identity matrix, and $\beta$ is an arbitrary state on $s$. 

RTTRs have a well-defined tensor product structure for parallel and sequential composition of process tensors, which reduces to the channel notions after sufficient coarse-graining. However, making use of intermediate interventions, allows for resource transformations that cannot be consistently described by channels and their transformations alone -- opening the possibility of extending useful channel results (e.g.~\cite{fundementallimitationsondistillation}) beyond the limits of where they are currently applicable -- as illustrated by the ability to preserve information through protocols like DD. 

\section{Information Preservation} \label{sec:informationpreservation}

Resource theories of temporal resolution are useful for any scenario where one aims to control a quantum process -- including for the purpose of information preservation -- subject to constraints on the form and/or the timing of their actions. The scope of corresponding experimental scenarios is far broader than what can be presented here, so we shall focus on the goal of information preservation, under one particular umbrella of constraints.

\subsection{$\mathsf{Q}_{\hat{\mathbbm{W}}}$ Resource Theory} \label{sec:QNtheory}

The broadest member of the information preservation sub-theory structure -- denoted by $\mathsf{Q}_{\hat{\mathbbm{W}}}$ -- is the scenario where the experimenter can perform any pre-determined, memoryless sequence of quantum operations at times $\hat{m} \ , \ \emptyset \subseteq \hat{m} \subseteq \hat{n}$, where $\hat{n}$ is a `maximum resolution', treated as inherent to the process resource $\mathbf{T}_{\hat{n}} \in \mathsf{T}_{\hat{\mathbbm{W}}}$. We place no restriction on the types of processes we may consider, so $\mathsf{T}_{\hat{\mathbbm{W}}}$ is the full set of process tensors. 

The superprocesses in $\mathsf{Q}_{\hat{\mathbbm{W}}}$, on the other hand, are allowed to be any arbitrary quantum operation at each time, but constrained to have no memory correlating them. This, in turn, makes this type of resource theory difficult to work with, since the set of free resources does not form a convex set. For fixed numbers of times, superprocesses following this structure have been recently explored within the RTQP $(\emptyset,\mathscr{Q})$, where $\emptyset$ denotes the absence of memory and $\mathscr{Q}$ comprises all possible time-local experimental interventions. Here, we simply extend this set to include any $n \in \mathbbm{W}$, obtaining the set of free superprocesses in $\mathsf{Q}_{\hat{\mathbbm{W}}}$: $\mathbf{Z}_{\hat{n}} \in \mathsf{Z}^{(\emptyset,\mathscr{Q})}_{\hat{\mathbbm{W}}}$ with the resultant free \emph{transformations} (between different numbers of times) of the form $\mathbf{Z}_{\hat{n}} | \mathbf{I}_{\hat{n}\setminus \hat{m}}\rrbracket$. Like any other set of free transformations, these are always resource non-increasing. With this in mind, `noise reduction' corresponds to minimising the loss of information in coarse-graining via an appropriately chosen free superprocess.

\subsection{Monotones of $\mathsf{Q}_{\hat{\mathbbm{W}}}$}

Since applying free transformations of $\mathsf{Q}_{\hat{\mathbbm{W}}}$ can lead to noise reduction, it is important to pin down monotones whose changes indicate how properties of the process are affected, and what resources are expended.

We begin by noting two marginal processes of a given a process $\mathbf{T}_{\hat{n}}$:
\begin{gather} \label{eq:marginals}
     \mathbf{T}^{\text{Mkv}}_{\hat{n}} \! := \! \bigotimes_{j=1}^{n+1}  {\rm tr}_{\bar{j}} \{ \mathbf{T}_{\hat{n}}\}
     \quad \mbox{and} \quad 
     \mathbf{T}^{\text{marg}}_{\hat{n}} \! := \!\!\! \bigotimes_{k=1}^{2(n+1)} \!\! {\rm tr}_{\bar{k}} \{ \mathbf{T}_{\hat{n}}\}.
\end{gather}
The index $j$ enumerates the constituent channels $\mathcal{T}_j$ as in Eq.~\eqref{eq:processtensor}, and $k$ splits this further into each input and output Hilbert space of the process tensor. Both of these are processes in their own right. The former process, $\mathbf{T}^{\text{Mkv}}_{\hat{n}}$, has temporal correlations only between an output and its preceding input, which make it a Markov processes. The latter process, $\mathbf{T}^{\text{marg}}_{\hat{n}}$, has no temporal correlations whatsoever; it maps any input to a fixed output state.

With these two reference processes, along with quantum relative entropy, $S(x\|y) := \mbox{tr}\{x \log(x) - x \log(y)\}$,
we can readily define three relevant monotones.
\begin{theorem} \label{thm:Imonotone}
In $\mathsf{Q}_{\hat{\mathbbm{W}}}$, the total information $I$, non-Markovianity $N$, and Markov information $M$,
\begin{gather}
\begin{split}
    \label{eq:Imonotoneresolution}
    &I(\mathbf{T}_{\hat{n}}) \! := \! S\left(\mathbf{T}_{\hat{n}} \| \mathbf{T}^{\emph{marg}}_{\hat{n}} \right), \ \
    N(\mathbf{T}_{\hat{n}}) \! := \! S\left(\mathbf{T}_{\hat{n}} \| \mathbf{T}^{\emph{Mkv}}_{\hat{n}} \right),\\
    &\mbox{and} \ \ M(\mathbf{T}_{\hat{n}}):=S\left( \mathbf{T}^{\emph{Mkv}}_{\hat{n}} \| \mathbf{T}^{\emph{marg}}_{\hat{n}} \right),
\end{split}
\end{gather}
are all monotones under the free operations of $\mathsf{Q}_{\hat{\mathbbm{W}}}$.
\end{theorem}

This result follows from the contractivity of relative entropy under the free superprocesses of $(\emptyset,\mathscr{Q})$, as shown in Ref.~\cite{resourcetheoriesofmultitime}, combined with contractivity of relative entropy under coarse-grainings, outlined in Sec.~\ref{sec:contractivityoutline}. Ref.~\cite{operationalmarkovcondition} showed that $N$ has a clear operational interpretation as a measure of how well a Markov model could describe $\mathbf{T}_{\hat{n}})$. While $N$ vanishes when a process has no multitime (non-Markovian) correlations, $I$ only vanishes when the process has no correlations whatsoever.

Importantly, the total information $I$ can be straightforwardly be related to $N$ and $M$ (see Sec.~\ref{sec:splitI}).
\begin{equation} \label{eq:mutualinformationequality}
    I(\mathbf{T}_{\hat{n}}) =M(\mathbf{T}_{\hat{n}})+N(\mathbf{T}_{\hat{n}}). 
\end{equation}
In other words, any correlations present in a process tensor $\mathbf{T}_{\hat{n}}$ must be attributable to either $M(\mathbf{T}_{\hat{n}})$ -- corresponding to memory due to interactions with the environment -- or $N(\mathbf{T}_{\hat{n}})$  -- corresponding to the capability to transmit information between adjacent times -- with no overlap. Combining this fact with irreversibility opens the possibility that an appropriately chosen superprocess $\mathbf{Z}_{\hat{n}}$, in conjunction with coarse-graining $\mathbf{I}_{\hat{n} \setminus \hat{m}}$ might consume $N(\mathbf{T}_{\hat{n}})$ to obtain a high value of $M(\mathbf{T}'_{\hat{m}})$, where $\mathbf{T}'_{\hat{m}}:= \llbracket \mathbf{T}_{\hat{n}} | \mathbf{Z}_{\hat{n}} | \mathbf{I}_{\hat{n} \setminus \hat{m}} \rrbracket$, corresponding to `decoupling' from the environment. Intuitively, the degree to which $N(\mathbf{T}_{\hat{n}})$ can be transformed into $M(\mathbf{T}'_{\hat{m}})$ sets a limit on how well dynamical decoupling can be performed.

\subsection{Decoupling Mechanisms} \label{sec:dynamicaldecoupling}

From a resource theoretic perspective, it is exactly this expenditure of resources that underlies dynamical decoupling. One period of a (traditional) DD sequence can be represented by a superprocess from the resource theory $\mathsf{D}_{\hat{\mathbbm{W}}} \subset \mathsf{Q}_{\hat{\mathbbm{W}}}$ only containing memoryless sequences of unitaries rather than general quantum operations. These pulse sequences are detailed in Fig.~\ref{fig:definingobjects}. Our consideration is more general than the usual view of how DD leads to `decoupling'~\cite{dynamicaldecouplingofopenquantumsystems}; that is, the pulse sequence averages out the influence of the environment as long as the pulses are sufficiently rapid. Here, we provide a more detailed explanation, which can account for non-rapid sequences and views DD as a symbiosis of two distinct effects. Firstly, the consumption of non-Markovianity in coarse-graining reduces the loss of system-level information; and secondly, a Zeno-like slowdown is induced by the first effect when the interventions are sufficiently rapid.

To understand the first effect, consider starting with a process $\mathbf{T}_{\hat{n}}$, that has multitime correlations at all scales. If DD is successful, it will map this process to some other process $\mathbf{T}'_{\hat{n}}=\llbracket \mathbf{T}_{\hat{n}} | \mathbf{Z}^{\text{DD}}_{\hat{n}}$, such that $\mathbf{T}'_{\hat{n}}$ has limited temporal correlation range. If DD is effective, then coarse-graining $\mathbf{T}'_{\hat{n}}$ to $\mathbf{T}'_{\hat{m}}= \llbracket \mathbf{T}'_{\hat{n}} | \mathbf{I}_{\hat{n} \setminus \hat{m}} \rrbracket$ reduces the total information of the coarse-grained process to a lesser extent than it would have been without the the DD superprocess, i.e., $I(\mathbf{T}'_{\hat{m}}) >  I\big(\llbracket \mathbf{T}_{\hat{n}} | \mathbf{I}_{\hat{n} \setminus \hat{m}} \rrbracket \big)$. Achieving this requires that the non-Markovianity of the coarse-grained process is small, and the correlations to be predominantly in the form of $M$, corresponding to a large throughput of information between adjacent times, and an effective decoupling between the system and the environment. The crucial observation here is that, $N(\mathbf{T}'_{\hat{n}})$ is consumed to enhance $I(\mathbf{T}'_{\hat{m}})$. We emphasize that this explanation of DD is not explicitly dependent on the speed at which the decoupling sequences can be applied, and solely leverages on the correlations present in the underlying process that can used to maximize the input-output correlations of the resulting channel.

This first effect can be compounded by a Zeno-like effect~\cite{unificationofdynamicaldecoupling} induced by fast DD pulses. Since the instantaneous rate of formation of system-environment correlations is tied to the quantity of existing correlations~\cite{lazystates}, the conversion of non-Markovianity into system-level correlations slows the rate of the flow that needs to be corrected. However, when working with slow pulses, DD cannot benefit from this effect, and its efficacy reduces.

Once DD, or some other kind of noise suppression method, has been applied, the success of our original goal -- to preserve information between the input and output of a quantum process -- can be quantified by coarse-graining the resultant process tensor, and then measuring $I$
\begin{equation} \label{eq:I_Z}
I_{\hat{m}|\mathbf{Z}_{\hat{n}}}(\mathbf{T}_{\hat{n}}) := I(\mathbf{T}'_{\hat{m}}) = I\big(\llbracket \mathbf{T}_{\hat{n}} | \mathbf{Z}_{\hat{n}}| \mathbf{I}_{\hat{n} \setminus \hat{m}} \rrbracket \big). 
\end{equation}
Observe that setting $\hat{m}=\emptyset$ recovers the mutual information of the channel defined in Eq.~\eqref{eq:PTcontrol}). Naturally, the important figure of merit to gauge the success of a respective decoupling scheme is the comparison between the standard decoupling scheme $I_{\hat{m}|\mathbf{Z}^{\text{DD}}_{\hat{n}}}$, and the case where no decoupling scheme is applied: $I_{\hat{m}}(\mathbf{T}_{\hat{n}}):=I\big(\llbracket \mathbf{T}_{\hat{n}}| \mathbf{I}_{\hat{n} \setminus \hat{m}} \rrbracket \big)$.

\subsection{Multitimescale Optimal Dynamical Decoupling} 
\label{sec:optimisation}

Using this understanding of DD, the question of finding the best noise suppression method amounts to finding a control sequence $\mathbf{Z}_{\hat{n}}$, such that $I_{\emptyset|\mathbf{Z}_{\hat{n}}}$ is maximised. A pulse sequence that outperforms DD when it comes to conversion of correlations from $N$ into $M$, will fare better at information preservation, and might lead to satisfactory decoupling even in cases where the respective controls are significantly spaced out in time. To demonstrate this performance enhancement, we search for such a pulse sequence by means of a semidefinite program (SDP)~\cite{Watrous11} and call the optimal procedure \textit{optimal dynamical decoupling (ODD)}. Specifically, while the maximization of the input-output mutual information is not directly amenable to SDP techniques, we find the sequence of operations that maximizes the maximal eigenvalue of the resulting channel. Intuitively, this is a proxy for optimal mutual information, and we use the corresponding control sequences to compare the figure of merit for the three cases $I_{\emptyset}$,
$I_{\emptyset|\mathbf{Z}^{\text{DD}}_{\hat{n}}}$, and $I_{\emptyset|\mathbf{Z}^{\text{ODD}}_{\hat{n}}}$.

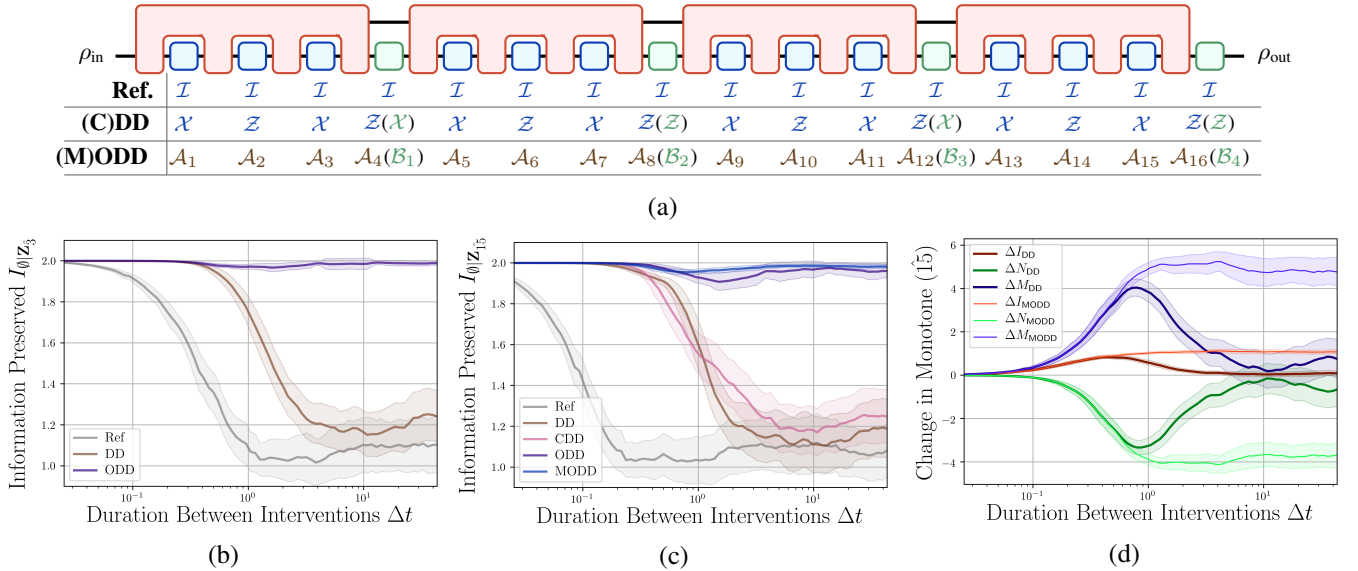
\begin{figure*}[ht!] 
\centering
\begin{tikzpicture}[scale=0.45]

\draw[black, very thick,solid] (0.5,1.5) -- (30.5,1.5);
\draw[black, very thick,solid] (-0.5,0.5) -- (32.5,0.5);

\draw[myred,fill=myredfill, thick,solid,rounded corners=2] (0+0.1,1+0.1) -- (0+0.1,0) -- (1-0.1,0) -- (1-0.1,1+0.1) -- (2-0.1,1+0.1) --
(0+0.1+2,1+0.1) -- (0+0.1+2,0) -- (1-0.1+2,0) -- (1-0.1+2,1+0.1) -- (2-0.1+2,1+0.1) --
(0+0.1+4,1+0.1) -- (0+0.1+4,0) -- (1-0.1+4,0) -- (1-0.1+4,1+0.1) -- (2-0.1+4,1+0.1) --
(0+0.1+6,1+0.1) -- (0+0.1+6,0) -- (1-0.1+6,0) -- (1-0.1+6,1+0.1) 
-- (1-0.1+6,2) -- (0+0.1,2) -- (0+0.1,1+0.1);

\draw[myred,fill=myredfill, thick,solid,rounded corners=2]
(0+0.1+8,1+0.1) -- (0+0.1+8,0) -- (1-0.1+8,0) -- (1-0.1+8,1+0.1) -- (2-0.1+8,1+0.1) --
(0+0.1+10,1+0.1) -- (0+0.1+10,0) -- (1-0.1+10,0) -- (1-0.1+10,1+0.1) -- (2-0.1+10,1+0.1) --
(0+0.1+12,1+0.1) -- (0+0.1+12,0) -- (1-0.1+12,0) -- (1-0.1+12,1+0.1) -- (2-0.1+12,1+0.1) --
(0+0.1+14,1+0.1) -- (0+0.1+14,0) -- (1-0.1+14,0) -- (1-0.1+14,1+0.1) 
-- (1-0.1+6+8,2) -- (0+0.1+8,2) -- (0+0.1+8,1+0.1);

\draw[myred,fill=myredfill, thick,solid,rounded corners=2]
(0+0.1+16,1+0.1) -- (0+0.1+16,0) -- (1-0.1+16,0) -- (1-0.1+16,1+0.1) -- (2-0.1+16,1+0.1) --
(0+0.1+18,1+0.1) -- (0+0.1+18,0) -- (1-0.1+18,0) -- (1-0.1+18,1+0.1) -- (2-0.1+18,1+0.1) --
(0+0.1+20,1+0.1) -- (0+0.1+20,0) -- (1-0.1+20,0) -- (1-0.1+20,1+0.1) -- (2-0.1+20,1+0.1) --
(0+0.1+22,1+0.1) -- (0+0.1+22,0) -- (1-0.1+22,0) -- (1-0.1+22,1+0.1) 
-- (1-0.1+6+16,2) -- (0+0.1+16,2) -- (0+0.1+16,1+0.1);

\draw[myred,fill=myredfill, thick,solid,rounded corners=2]
(0+0.1+24,1+0.1) -- (0+0.1+24,0) -- (1-0.1+24,0) -- (1-0.1+24,1+0.1) -- (2-0.1+24,1+0.1) --
(0+0.1+26,1+0.1) -- (0+0.1+26,0) -- (1-0.1+26,0) -- (1-0.1+26,1+0.1) -- (2-0.1+26,1+0.1) --
(0+0.1+28,1+0.1) -- (0+0.1+28,0) -- (1-0.1+28,0) -- (1-0.1+28,1+0.1) -- (2-0.1+28,1+0.1) --
(0+0.1+30,1+0.1) -- (0+0.1+30,0) -- (1-0.1+30,0) -- (1-0.1+30,1+0.1)  
-- (1-0.1+6+24,2) -- (0+0.1+24,2) -- (0+0.1+24,1+0.1);

\draw[myblue,fill=mybluefill,thick,solid,rounded corners=2] (1+0.1,0+0.1) rectangle (2-0.1,1-0.1);

\draw[myblue,fill=mybluefill,thick,solid,rounded corners=2] (1+0.1+2,0+0.1) rectangle (2-0.1+2,1-0.1);

\draw[myblue,fill=mybluefill,thick,solid,rounded corners=2] (1+0.1+4,0+0.1) rectangle (2-0.1+4,1-0.1);

\draw[mygreen,fill=mygreenfill,thick,solid,rounded corners=2] (1+0.1+6,0+0.1) rectangle (2-0.1+6,1-0.1);

\draw[myblue,fill=mybluefill,thick,solid,rounded corners=2] (1+0.1+8,0+0.1) rectangle (2-0.1+8,1-0.1);

\draw[myblue,fill=mybluefill,thick,solid,rounded corners=2] (1+0.1+10,0+0.1) rectangle (2-0.1+10,1-0.1);

\draw[myblue,fill=mybluefill,thick,solid,rounded corners=2] (1+0.1+12,0+0.1) rectangle (2-0.1+12,1-0.1);

\draw[mygreen,fill=mygreenfill,thick,solid,rounded corners=2] (1+0.1+14,0+0.1) rectangle (2-0.1+14,1-0.1);

\draw[myblue,fill=mybluefill,thick,solid,rounded corners=2] (1+0.1+16,0+0.1) rectangle (2-0.1+16,1-0.1);

\draw[myblue,fill=mybluefill,thick,solid,rounded corners=2] (1+0.1+18,0+0.1) rectangle (2-0.1+18,1-0.1);

\draw[myblue,fill=mybluefill,thick,solid,rounded corners=2] (1+0.1+20,0+0.1) rectangle (2-0.1+20,1-0.1);

\draw[mygreen,fill=mygreenfill,thick,solid,rounded corners=2] (1+0.1+22,0+0.1) rectangle (2-0.1+22,1-0.1);

\draw[myblue,fill=mybluefill,thick,solid,rounded corners=2] (1+0.1+24,0+0.1) rectangle (2-0.1+24,1-0.1);

\draw[myblue,fill=mybluefill,thick,solid,rounded corners=2] (1+0.1+26,0+0.1) rectangle (2-0.1+26,1-0.1);

\draw[myblue,fill=mybluefill,thick,solid,rounded corners=2] (1+0.1+28,0+0.1) rectangle (2-0.1+28,1-0.1);

\draw[mygreen,fill=mygreenfill,thick,solid,rounded corners=2] (1+0.1+30,0+0.1) rectangle (2-0.1+30,1-0.1);

\draw[] (-1.2,0.5) node {\small $\rho_\text{in}$};
\draw[] (33.4,0.5) node {\small $\rho_\text{out}$};

\draw[] (0,-0.5) node {\small \textbf{Ref.}};

\draw[] (1.5,-0.5) node {\footnotesize ${\color{myblue}\mathcal{I}}$};
\draw[] (1.5+2,-0.5) node {\footnotesize ${\color{myblue}\mathcal{I}}$};
\draw[] (1.5+4,-0.5) node {\footnotesize ${\color{myblue}\mathcal{I}}$};
\draw[] (1.5+6,-0.5) node {\footnotesize ${\color{myblue}\mathcal{I}}$};

\draw[] (1.5+8,-0.5) node {\footnotesize ${\color{myblue}\mathcal{I}}$};
\draw[] (1.5+2+8,-0.5) node {\footnotesize ${\color{myblue}\mathcal{I}}$};
\draw[] (1.5+4+8,-0.5) node {\footnotesize ${\color{myblue}\mathcal{I}}$};
\draw[] (1.5+6+8,-0.5) node {\footnotesize ${\color{myblue}\mathcal{I}}$};

\draw[] (1.5+16,-0.5) node {\footnotesize ${\color{myblue}\mathcal{I}}$};
\draw[] (1.5+2+16,-0.5) node {\footnotesize ${\color{myblue}\mathcal{I}}$};
\draw[] (1.5+4+16,-0.5) node {\footnotesize ${\color{myblue}\mathcal{I}}$};
\draw[] (1.5+6+16,-0.5) node {\footnotesize ${\color{myblue}\mathcal{I}}$};

\draw[] (1.5+24,-0.5) node {\footnotesize ${\color{myblue}\mathcal{I}}$};
\draw[] (1.5+2+24,-0.5) node {\footnotesize ${\color{myblue}\mathcal{I}}$};
\draw[] (1.5+4+24,-0.5) node {\footnotesize ${\color{myblue}\mathcal{I}}$};
\draw[] (1.5+6+24,-0.5) node {\footnotesize ${\color{myblue}\mathcal{I}}$};

\draw[mygrey,thin,solid] (-2,-1) -- (33,-1);

\draw[] (-0.5,-0.5-1) node {\small \textbf{(C)DD}};

\draw[] (1.5,-0.5-1) node {\footnotesize ${\color{myblue}\mathcal{X}}$};
\draw[] (1.5+2,-0.5-1) node {\footnotesize ${\color{myblue}\mathcal{Z}}$};
\draw[] (1.5+4,-0.5-1) node {\footnotesize ${\color{myblue}\mathcal{X}}$};
\draw[] (1.5+6,-0.5-1) node {\footnotesize ${\color{myblue}\mathcal{Z}}$({\color{mygreen}$\mathcal{X}$})};

\draw[] (1.5+8,-0.5-1) node {\footnotesize ${\color{myblue}\mathcal{X}}$};
\draw[] (1.5+2+8,-0.5-1) node {\footnotesize ${\color{myblue}\mathcal{Z}}$};
\draw[] (1.5+4+8,-0.5-1) node {\footnotesize ${\color{myblue}\mathcal{X}}$};
\draw[] (1.5+6+8,-0.5-1) node {\footnotesize ${\color{myblue}\mathcal{Z}}$({\color{mygreen}$\mathcal{Z}$})};

\draw[] (1.5+16,-0.5-1) node {\footnotesize ${\color{myblue}\mathcal{X}}$};
\draw[] (1.5+2+16,-0.5-1) node {\footnotesize ${\color{myblue}\mathcal{Z}}$};
\draw[] (1.5+4+16,-0.5-1) node {\footnotesize ${\color{myblue}\mathcal{X}}$};
\draw[] (1.5+6+16,-0.5-1) node {\footnotesize ${\color{myblue}\mathcal{Z}}$({\color{mygreen}$\mathcal{X}$})};

\draw[] (1.5+24,-0.5-1) node {\footnotesize ${\color{myblue}\mathcal{X}}$};
\draw[] (1.5+2+24,-0.5-1) node {\footnotesize ${\color{myblue}\mathcal{Z}}$};
\draw[] (1.5+4+24,-0.5-1) node {\footnotesize ${\color{myblue}\mathcal{X}}$};
\draw[] (1.5+6+24,-0.5-1) node {\footnotesize ${\color{myblue}\mathcal{Z}}$({\color{mygreen}$\mathcal{Z}$})};

\draw[mygrey,thin,solid] (-2,-1-1) -- (33,-1-1);

\draw[] (-0.9,-0.5-1-1) node {\small \textbf{(M)ODD}};

\draw[] (1.5,-0.5-1-1) node {\footnotesize ${\color{mybrown}\mathcal{A}_{1}}$};
\draw[] (1.5+2,-0.5-1-1) node {\footnotesize ${\color{mybrown}\mathcal{A}_{2}}$};
\draw[] (1.5+4,-0.5-1-1) node {\footnotesize ${\color{mybrown}\mathcal{A}_{3}}$};
\draw[] (1.5+6,-0.5-1-1) node {\footnotesize ${\color{mybrown}\mathcal{A}_{4}}$({\color{mygreen}$\mathcal{B}_{1}$})};

\draw[] (1.5+8,-0.5-1-1) node {\footnotesize ${\color{mybrown}\mathcal{A}_{5}}$};
\draw[] (1.5+2+8,-0.5-1-1) node {\footnotesize ${\color{mybrown}\mathcal{A}_{6}}$};
\draw[] (1.5+4+8,-0.5-1-1) node {\footnotesize ${\color{mybrown}\mathcal{A}_{7}}$};
\draw[] (1.5+6+8,-0.5-1-1) node {\footnotesize ${\color{mybrown}\mathcal{A}_{8}}$({\color{mygreen}$\mathcal{B}_{2}$})};

\draw[] (1.5+16,-0.5-1-1) node {\footnotesize ${\color{mybrown}\mathcal{A}_{9}}$};
\draw[] (1.5+2+16,-0.5-1-1) node {\footnotesize ${\color{mybrown}\mathcal{A}_{10}}$};
\draw[] (1.5+4+16,-0.5-1-1) node {\footnotesize ${\color{mybrown}\mathcal{A}_{11}}$};
\draw[] (1.5+6+16,-0.5-1-1) node {\footnotesize ${\color{mybrown}\mathcal{A}_{12}}$({\color{mygreen}$\mathcal{B}_{3}$})};

\draw[] (1.5+24,-0.5-1-1) node {\footnotesize ${\color{mybrown}\mathcal{A}_{13}}$};
\draw[] (1.5+2+24,-0.5-1-1) node {\footnotesize ${\color{mybrown}\mathcal{A}_{14}}$};
\draw[] (1.5+4+24,-0.5-1-1) node {\footnotesize ${\color{mybrown}\mathcal{A}_{15}}$};
\draw[] (1.5+6+24,-0.5-1-1) node {\footnotesize ${\color{mybrown}\mathcal{A}_{16}}$({\color{mygreen}$\mathcal{B}_{4}$})};

\draw[mygrey,thin,solid] (-2,-1-1-1) -- (33,-1-1-1);

\draw[mygrey,thin,solid] (1,0) -- (1,-1-1-1);

\draw[] (1.5+6+8,-0.5-1-1-1.5) node {(a)};

\end{tikzpicture}

\begin{minipage}{0.32\linewidth}
        \centering
        \includegraphics[width=\linewidth]{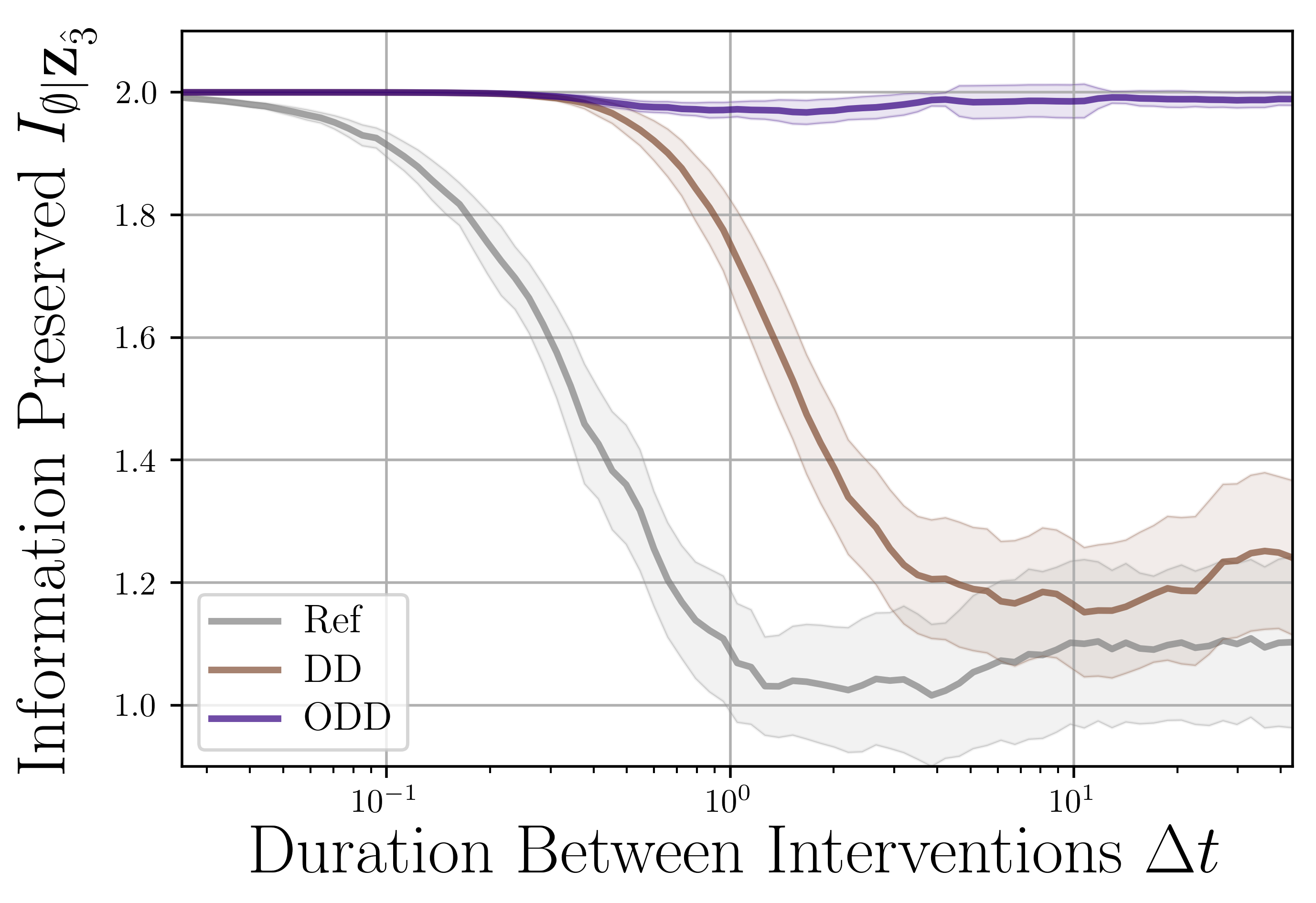}
        
        (b)
    \end{minipage}
\begin{minipage}{0.32\linewidth}
        \centering
        \includegraphics[width=\linewidth]{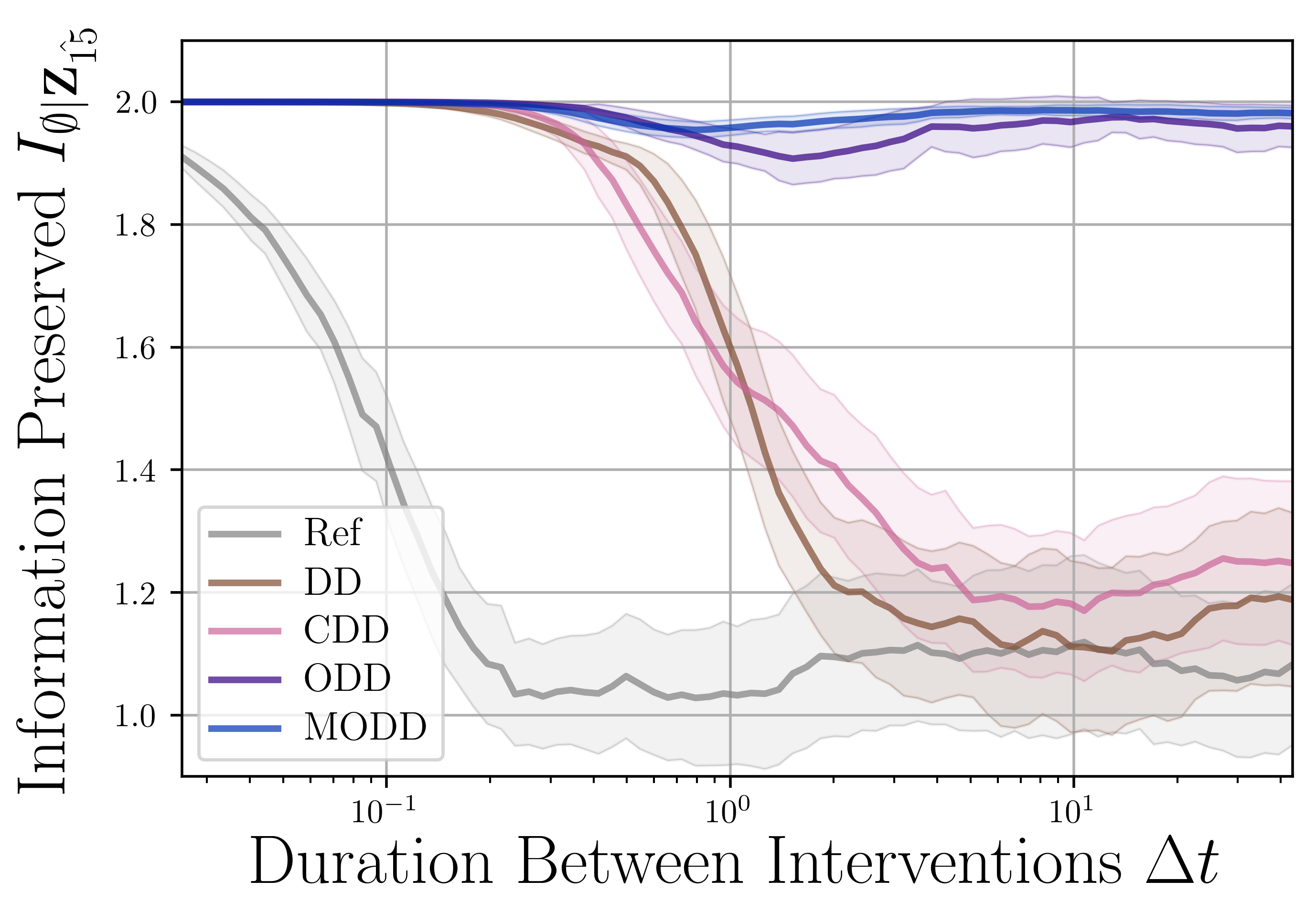}
        
        (c)
    \end{minipage}
    \begin{minipage}{0.32\linewidth}
        \centering
        \includegraphics[width=\linewidth]{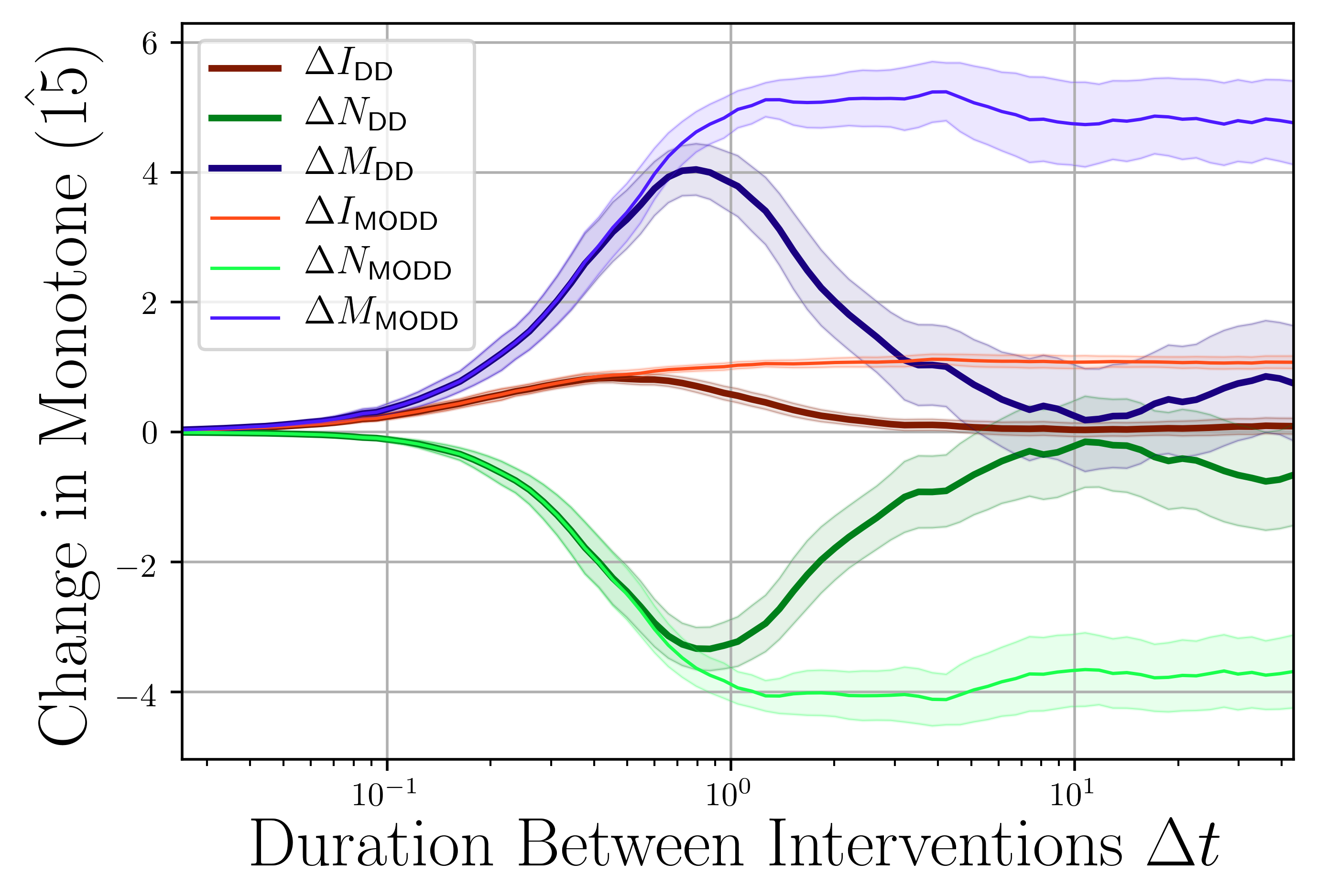}
        
        (d)
    \end{minipage}

\caption{Numerical study based on the model described in Sec.~\ref{sec:methods}. All free evolutions last a duration $\Delta t$, in units of Planck's constant. The data is smoothed in the time domain to average over oscillatory behaviour at long durations, and a $2\sigma$ confidence interval of the mean is shaded. {\bf(a)} A process with 15 intermediate times has either no control applied (Ref), DD, CDD, ODD, or MODD applied.  DD, CDD, or multitimescale optimisation applied. For CDD and MODD, the additional operations applied at the longer timescale are in green. {\bf(b)} Effectiveness at preserving channel-level information $I_{\emptyset|\mathbf{Z}_{\hat{n}}}$ and $I_{\emptyset}$ for the first 3 intermediate interventions of (a). A single period of DD can correct noise for roughly one order of magnitude in $\Delta t$, while ODD is effective for the whole range of $\Delta t$ values tested. {\bf(c)} Effectiveness at preserving channel-level information $I_{\emptyset|\mathbf{Z}_{\hat{n}}}$ and $I_{\emptyset}$ for all 15 intermediate interventions of (a). CDD shows a marginal improvement over DD. Both are significantly outperformed by ODD, which comes close to saturating mutual information for the whole set of $\Delta t$ values. Finally, MODD, obtains a small improvement over the already near-ideal result of ODD. {\bf(d)} Changes in monotones, e.g. $\Delta I= I_{\hat{m}|\mathbf{Z}_{n}}(\mathbf{T}_{\hat{n}})-I_{\hat{m}}(\mathbf{T}_{\hat{n}})$, after coarse-graining from $n=15$ intermediate interventions to $m=3$. For both traditional DD and MODD, $N$ is always lower with the protocol than without, indicating a conversion of multitime correlations to system-level correlations $M$. This effect, compounded by a Zeno-like slowdown in the formation of new system-environment correlations causes $I$ to increase as well. The sharp decline in the effectiveness of DD in (c) begins at the peak in (d), suggesting that DD loses effectiveness when it can no longer convert between $N$ and $M$.} \label{fig:numerics}
\end{figure*}

We test this resource theoretic characterisation of DD using a prototypical model (see Sec.~\ref{sec:methods} for details). Our first result is to optimally decouple a 4-intervention process. The left panel of Fig.~\ref{fig:numerics} shows that ODD achieves significant noise reduction over standard DD, especially at long timescales. It is well-known that DD ceases to be effective if the pulses are too far separated. Yet, if indeed it is possible to reduce $N$ then effective decoupling from the environment should still be possible, as demonstrated by the efficacy of ODD. What these results suggest is that DD works primarily due the second decoupling mechanism discussed above, limiting it to short timescales. In contrast, ODD is primarily utilising the first decoupling mechanism and thus able be effective for longer timescales.

Searching for the ODD pulses requires characterising the noise process $\mathbf{T}_{\hat{n}}$. The complexity of a process grows with the number of interventions $n$. This is a severe limiting factor to the scalability of ODD. However, by employing a see-saw SDP, we can iteratively nest superprocesses to optimise an $n$-intervention pulse by optimising $m$-interventions at a time (with $m<n$). The total complexity, thus, grows linearly. The middle panel of Fig.~\ref{fig:numerics} shows that, indeed iterative ODD remains highly effective at long times, while DD does not.

Yet, there still more resources remaining untapped. Once we find $n$ pulses for iterative ODD, we may further optimise these pulses at a higher timescale, e.g. apply an additional pulse for every fourth intervention. This allows for reducing the non-Markovian correlations $N$ at the larger timescales. We call this layered approach \textit{multitimescale optimal dynamical decoupling} (MODD), which we detail in Sec.~\ref{sec:optimisation}. MODD is closely related (in spirit) to concatenated dynamical decoupling~\cite{faulttolerantquantumdynamicaldecoupling} (CDD). These methods operate at multiple timescales, and thus become imperative when the noise is complex, and temporal correlations exist at multiple timescales. The middle panel of Fig.~\ref{fig:numerics} shows that, while ODD is able to preserve most of information in a process, MODD still allows for further gains. 

Another distinct advantage of iterative ODD and MODD is that the consumption of non-Markovianity can be quantified at each transition to a shorter timescale. The right panel of Fig.~\ref{fig:numerics} shows how non-Markovianity $N$ is consumed in coarse-graining a 16 free evolutions into 4 free evolutions for each strategy above. This figure highlights the relationship between the non-Markovianity, temporal resolution, and information preservation, and single-timescale optimisation. This is the main message of this work; namely, there are plentiful of untapped quantum dynamical resources that have the potential to extract a great deal of quantum coherence from the current generation of NISQ devices.

\section{Discussion} \label{sec:discussion}

For a practical implementation of the methods described in the last section, one first requires characterising the multitime noise process $\mathbf{T}_{\hat{n}}$ itself. Indeed, this has recently been achieved on a commercial-grade device~\cite{demonstrationofnonmarkovianprocess}. Since then, non-Markovian characterisation has been refined~\cite{nonmarkovianquantumprocesstomography} and can even be automated~\cite{PhysRevA.102.062414}. Importantly, the non-Markovian noise characterisation yields prediction-fidelities limited only by the shot noise~\cite{demonstrationofnonmarkovianprocess}, far outperforming methods that make a Markov assumption. Moreover, Refs.~\cite{demonstrationofnonmarkovianprocess, nonmarkovianquantumprocesstomography, diagnosingtemporalquantum} used the characterisation information for noise reduction, which is a variant of optimal dynamical decoupling. 

The present results allow for a formal quantification of the resources expended in noise mitigation techniques, and extend their domain of applicability to arbitrary lengths. The real-device implementations listed above mean that our results can be readily integrated on state-of-the-art devices. Doing so will naturally enhance the quantum capabilities of these devices and, e.g., foster an increase in the quantum volume.

While the immediate advantages are clear, there remain several outstanding challenges. Perhaps the most consequential unresolved question is whether there exists a simple bound on how large the experimenter can make $I_{\emptyset|\mathbf{Z}_{\hat{n}}}$. The contractivity of relative entropy under free transformations implies that total mutual information $I(\mathbf{T}_{\hat{n}})$ will always be at least as large as that of any other process it can reach, including $\sup_{\mathbf{Z}_{\hat{n}} \in \mathsf{Z}_{\hat{\mathbbm{W}}} } I_{\hat{m}}(\llbracket \mathbf{T}_{\hat{n}} | \mathbf{Z}_{\hat{n}}| \mathbf{I}_{\hat{n} \setminus \hat{m}} \rrbracket )$, with $\hat{m}=\emptyset$ corresponding to the highest one can make the input-output mutual information under any allowed control. However, the bound is not tight; the number of terms in $I(\mathbf{T}_{\hat{n}})$ is proportional to $|\hat{n}|+1$, implying that there is a tendency for $I$ to be higher for finer grained process tensors independently of the underlying physics. However, simple numerical checks show that using the normalised value $I/(|\hat{n}|+1)$ does not properly bound $I_{\emptyset}$ in all cases. Obtaining a separate monotone that provides a tighter bound, or identifying a sub-class of process tensors where $I/(|\hat{n}|+1)$ is a valid bound would be a powerful result: one would then only need to characterise a process tensor to determine the best \textit{any} noise reduction scheme can perform. Trace distance and diamond distance~\cite{memorystrength} are promising alternatives to relative entropy to produce monotones, since they are normalised to unity. However, both suffer from a disadvantage compared to relative entropy, in that they require a non-trivial optimisation in order to find the nearest free process. How this trade-off plays out in practice is likely to have important consequences for efficient characterisation and minimisation of noise on real quantum devices.

\subsection{Other Noise Suppression Methods in $\mathsf{Q}_{\hat{\mathbbm{W}}}$} \label{sec:subtheories}

In this paper, while we have introduced and developed the core idea of resource theories of temporal resolution, we have only explored one facet of this rich structure. We have devoted most of our attention to the RTTR $\mathsf{Q}_{\hat{\mathbbm{W}}}$ and its subset $\mathsf{D}_{\hat{\mathbbm{W}}}$ to explore dynamical decoupling. Yet, there remain many other structures unexplored that are related to physically interesting phenomena such as the quantum Zeno effect, decoherence free subspaces, and quantum error correction.

In fact, quantum error correction (QEC) can be naturally framed within $\mathsf{Q}_{\hat{\mathbbm{W}}}$ by letting the system be composed of many qubits. In the case of QEC, the experimenter can employ a specified number of ancillary qubits to spread the information about the main system across many subsystems. This enables syndrome measurements and informs one about the corresponding corrections to be carried out. QEC relies on an experimenter who can, in addition to what is required for DD, perform measurements as underlined by the inclusion $\mathsf{D}_{\hat{\mathbbm{W}}} \subset \mathsf{Q}_{\hat{\mathbbm{W}}}$. What is gained with this extra ability is that Markovian processes are amenable to QEC despite not being amenable to DD in general~\cite{canquantummarkovevolutions}. A detailed exploration has the potential for discovering more untapped resources and refining the practical implementation of QEC.

Decoherence free subspaces (DFS)~\cite{ddreview, decoherencefreesubspaces} naturally fit in the RTTR $\mathsf{P}_{\hat{\mathbbm{W}}}$ which is a subset of the RTTR $ \mathsf{D}_{\hat{\mathbbm{W}}} \subset \mathsf{Q}_{\hat{\mathbbm{W}}}$. Here, even with Markovian processes and greater experimental constraints it may be possible to harness symmetries present in the system-environment interaction. The sub-theory $\mathsf{P}_{\hat{\mathbbm{W}}}$ is more restrictive than $\mathsf{D}_{\hat{\mathbbm{W}}}$ because inducing a DFS requires only the repetition of a \emph{single} unitary, rather than a complex sequence of different unitaries. The fact that $\mathsf{P}_{\hat{\mathbbm{W}}} \subset \mathsf{D}_{\hat{\mathbbm{W}}}$ also opens the possibility of new techniques to harness both effects simultaneously.

The quantum Zeno effect (QZE)~\cite{analysisofthequantumzeno} can be cast as the RTTR $\mathsf{C}_{\hat{\mathbbm{W}}} \subset \mathsf{Q}_{\hat{\mathbbm{W}}}$, where the repeated action is a measurement rather than a unitary, and one can still preserve classical information. The QZE -- as we have laid out -- shares a commonality~\cite{unificationofdynamicaldecoupling} with DD in that both types of actions remove correlations between the system and environment to slow down decay process. The difference lies in the fact that DD harnesses non-Markovianity to do so, and consequentially preserves full quantum information, rather than just classical information. The sub-theory structure of $\mathsf{Q}_{\hat{\mathbbm{W}}}$ is summarised in panel (b) of Fig.~\ref{fig:coarsegraining}.

Finally, let us bring the discussion back to DD. Uhrig dynamical decoupling (UDD)~\cite{keepingaqubitalive} and other variants of DD can be examined within the framework we have provided. It is known that UDD achieves high order decoupling using relatively few pulses by optimising the timing between pulses within a specified interval. Interestingly, from the perspective of RTTRs, this optimised pulse spacing can be interpreted as demanding a greater resource. However, since process tensors with the same number of differently spaced interventions are not comparable in the resource \emph{preorder}, different monotones used to impose \emph{total orders} onto the resource theory may disagree about the cost of this extra requirement. In simpler terms, one could define the `temporal resolution' required to perform UDD as the shortest gap between pulses used (which is the interval between $t_n$ and t), and create a process tensor with $~\lfloor t/(t-t_n) \rfloor -1$ instants for intermediate interventions. In such a scenario, the resource requirement is not the number of pulses which \textit{are} used, but the number of pulses which \emph{could be} used, resulting in the perception of UDD having a significantly poorer scaling. Clearly, whether this is experimentally the case depends on the physical constraints of the apparatus -- is the limiting factor the total number of actions, or their rapidity?

\subsection{Application of RTTRs Beyond Noise Suppression} \label{sec:suptheories}

It is conceivable that resource theories could be devised where the abilities of the experimenter (the form and connectivity of free superprocesses) change as the timescale changes. Interesting trade-offs between speed and quality might appear in such theories, meaning that they might not satisfy the notion of irreversibility under temporal coarse-grainings we use in our theories. As such, these interesting structures are beyond the scope of this work and subject to future work.

A more powerful resource theory than $\mathsf{Q}_{\hat{\mathbbm{W}}}$ would be one where the free superprocesses are those of $(\mathscr{B},\mathscr{Q})$~\cite{resourcetheoriesofmultitime} -- corresponding to an experimenter who can carry a classical memory that is fed forward between control operations. It remains unexplored how this additional power can be utilised for information preservation. However, from a technical standpoint this theory may be easier to study, since free processes and allowed controls are convex sets.

The discussion in this paper has been concerned with the conversion of single copies of process tensors, which -- one might argue -- is the most operationally relevant scenario as they usually used to depict a single experiment. However, the idea of distilling noiseless channels -- akin to magic state or entanglement distillation -- is likely to be immensely useful for quantum technologies. Clearly any result for channel distillation~\cite{reversibleframework, fundementallimitationsondistillation} applies here too, since any process tensor can be coarse-grained into a channel. However, as Obs.~\ref{obs:irreversibility} shows, more channels can be reached if one has the additional temporal resolution of a process tensor, raising the question of whether existing rates of channel distillation can be improved by accounting for this extra resource. Such a theory will set the rules for what transformations are possible between processes with different entanglement structures~\cite{witnessingquantummemory,genuinemultipartiteentanglementintime}, which can form the basis for a resource theory of spatiotemporal entanglement~\cite{entanglementtheoryandthesecondlaw}.

\section{Methods} \label{sec:methods}

Here we present supporting details behind claims made throughout the paper, as well as a brief description of our numerical model. We begin by outlining a proof of the contractivity of relative entropy under temporal coarse-graining -- as required for $I$ to be a monotone in $\mathsf{Q}_{\hat{\mathbbm{W}}}$. Next we provide a proof of Obs.~\ref{obs:irreversibility} via a supporting lemma concerning the representation of free transformations. We follow this up with a monotone notion of irreversibility, and the theorem linking it to perceived non-monotonicity by a coarse-grained experimenter. Next, we prove that $I$ can be partitioned into $M+N$, and show that both are individually monotones in $\mathsf{Q}_{\hat{\mathbbm{W}}}$. The final section provides a brief description of the numerical model used for the results of  Fig.~\ref{fig:numerics}.

Also, see the supplementary material for a detailed discussion of the parallel and sequential product structure of these theories, as well as the sub-theory structure of $\mathsf{Q}_{\hat{\mathbbm{W}}}$. The supplementary material also contains a discussion of the Markovianisation of noise via dynamical decoupling, and a summary of the notation used throughout this text.

\subsection{Contractivity of Relative Entropy Under Coarse-Grainings} \label{sec:contractivityoutline}

As a pre-requisite for $\mathsf{Q}_{\hat{\mathbbm{W}}}$ to be considered a useful resource theory for describing information preservation, we require that mutual information is respected as a monotone. It is already known that free superprocesses respect this quantity~\cite{resourcetheoriesofmultitime}. What we seek to show here is that mutual information is contractive under \emph{temporal coarse-graining}. To do this, it is sufficient to show that temporal coarse-graining does not increase the relative entropy $S(\mathbf{T}_{\hat{n}}\Vert \mathbf{R}_{\hat{n}}):=\text{tr} \{ \mathbf{T}_{\hat{n}} \log(\mathbf{T}_{\hat{n}}) \} - \text{tr} \{ \mathbf{T}_{\hat{n}} \log(\mathbf{R}_{\hat{n}}) \}$ between any two process tensors $\mathbf{T}_{\hat{n}}$ and $\mathbf{R}_{\hat{n}}$. 

The proof (Sup.~\ref{sup:contractivity}) involves separately proving that all pre-requisites of a theorem , provided in Ref.~\cite{monotonicityofthequantumrelative}, that guarantees contractivity of relative entropy under a mapping, are satisfied. These conditions are: positivity, trace preservation between the relevant domain and image sets, and linearity. Trace preservation, and positivity are shown by explicitly writing the Choi state of a process, before and after coarse-graining, in terms of the composition of channels acting on maximally entangled states. Linearity is shown directly.

\subsection{Representation of Free Transformations}

Here we prove a lemma used in the proof that coarse-grainings are irrerversible.
\begin{lemma} \label{lem:representation}
Given a resource theory of temporal resolution $ \mathsf{S}_{\hat{\mathbbm{W}}}$, and sub-theories for fixed temporal resolution $ \mathsf{S}_{\hat{m}} \subseteq \mathsf{S}_{\hat{n}} \subseteq \mathsf{S}_{\hat{\mathbbm{W}}}$, all free transformations from $\mathsf{S}_{\hat{n}}$ to $\mathsf{S}_{\hat{m}}$ can be represented as
\begin{equation}
   \mathbf{Y}_{\hat{n} \setminus \hat{m}} = \mathbf{Z}_{\hat{n}} | \mathbf{I}_{\hat{n} \setminus \hat{m}} \rrbracket.
\end{equation}
\end{lemma}
\begin{proof}
A sequence of transformations in a resource theory of temporal resolution can be written explicitly as
\begin{equation} \label{eq:functorchain}
      \mathbf{Z}_{\hat{z}} | \mathbf{I}_{\hat{z} \setminus \hat{y}} \rrbracket |  \mathbf{Z}_{\hat{y}} | \mathbf{I}_{\hat{y} \setminus \hat{x}} \rrbracket \dots |\mathbf{Z}_{\hat{b}} | \mathbf{I}_{\hat{b} \setminus \hat{a}} \rrbracket,
\end{equation}
with $\hat{z} \supseteq  \dots \supseteq \hat{a}$. A generalised Kolmogorov extension theorem~\cite{kolmogorov} implies that for any level of coarse-graining $\hat{\beta}$ there exists a fine grained view of the underlying process such that $\hat{\alpha} \supseteq \hat{\beta}$ such that any $\mathbf{Z}_{\hat{\beta}}$ can be expressed as $\mathbf{Z}_{\hat{\beta}} := \llbracket \mathbf{I}_{{\hat{\alpha}} \setminus {\hat{\beta}}} | \mathbf{Z}_{\hat{\alpha}} | \mathbf{I}_{\hat{\alpha} \setminus \hat{\beta}} \rrbracket$ for some fine grained superprocess $\mathbf{Z}_{\hat{\alpha}}$. The actions of $\mathbf{Z}_{\hat{\alpha}}$ at times ${\hat{\alpha}} \setminus {\hat{\beta}}$ are identities, ensuring that the physical situation is equivalent. Hence Eq.~\eqref{eq:functorchain} can be re-written as
\begin{equation} \label{eq:functorchain2}
      \mathbf{Z}_{\hat{z}} |\mathbf{Z}^{y}_{\hat{z}}| \dots |\mathbf{Z}^{a}_{\hat{z}} | \mathbf{I}_{\hat{z} \setminus \hat{a}} \rrbracket  := \mathbf{Z}'_{\hat{z}} | \mathbf{I}_{\hat{z} \setminus \hat{a}} \rrbracket,
\end{equation}
for a sequence of fine grained actions $\mathbf{Z}^{y}_{\hat{z}},\dots,\mathbf{Z}^{a}_{\hat{z}}$
\end{proof}

\subsection{Proof of Obs.~\ref{obs:irreversibility}}

To prove that more processes can be reached by applying a superprocess then coarse-graining, compared to coarse-graining and then applying a superprocess, we show that the latter can always be re-written in the form of Lem.~\ref{lem:representation}, but that the converse is false in general. 

The former statement follows straightforwardly from Lem.~\ref{lem:representation}. By this lemma, the transformation $ \mathbf{I}_{{\hat{n}} \setminus {\hat{m}}} \rrbracket \Big| \llbracket \mathbf{I}_{{\hat{n}} \setminus {\hat{m}}} | \mathbf{Z}_{\hat{n}} | \mathbf{I}_{{\hat{n}} \setminus {\hat{m}}} \rrbracket$ can be represented by $\mathbf{Z}'_{\hat{n}} | \mathbf{I}_{{\hat{n}} \setminus {\hat{m}}} \rrbracket$ for an appropriate choice of $\mathbf{Z}'_{\hat{n}}$. 

Secondly, since $\hat{n} \supset \hat{m}$ is a strict inclusion, $\mathbf{Z}'_{\hat{n}}$ is restricted in that no non-trivial actions may occur at times in $\hat{n} \setminus \hat{m}$, proving Obs.~\ref{obs:irreversibility}.

\subsection{Monotone Statement of Irreversibility} \label{sec:monotoneirreversibility}

It is possible to re-frame irreversibility in terms of monotones. The following inequality follows directly from Obs.~\ref{obs:irreversibility}.
\begin{corollary} \label{cor:irreversibilitymonotone}
For any valid monotone $M: \mathsf{T}_{\hat{\mathbbm{W}}} \rightarrow \mathbbm{R}_{\geq0}$ in a resource theory of temporal resolution $\mathsf{S}_{\hat{\mathbbm{W}}}$, and all process tensors $\mathbf{T}_{\hat{n}} \in \mathsf{T}_{\hat{\mathbbm{W}}}$, with $\emptyset \subseteq \hat{m} \subseteq \hat{n}$
\begin{equation} \label{eq:monotoneirreversibility}
    \sup_{\mathbf{Z}_{\hat{n}} \in \mathsf{Z}_{\hat{n}}} M\Big( \llbracket \mathbf{T}_{\hat{n}} | \mathbf{Z}_{\hat{n}} | \mathbf{I}_{{\hat{n}}-{\hat{m}}} \rrbracket  \Big) \geq \sup_{\mathbf{Z}_{\hat{m}} \in \mathsf{Z}_{\hat{n}}} M\big( \Big\llbracket \llbracket \mathbf{T}_{\hat{n}}  | \mathbf{I}_{{\hat{n}}-{\hat{m}}} \rrbracket \big|  \mathbf{Z}_{\hat{m}}  \big).
\end{equation}
\end{corollary}
This works because monotones in resource theories satisfy $a \rightarrow b \ \Rightarrow \ M(a) \geq M(b)$, and Obs.~\ref{obs:irreversibility} guarantees that the left hand side of Eq.~\eqref{eq:monotoneirreversibility} can reach at least as much as the right hand side. Given the same underlying dynamics, a finer grained process tensor description will always be preferable. 

\subsection{Irreversibility Leads to Perceived Non-Monotonicity} \label{sec:irreversibilitynonmonotonicity}

One of the crucial features of $\mathsf{Q}_{\hat{\mathbbm{W}}}$ is that, when information preservation as measured by a coarse experimenter $I_{\emptyset}$ can be improved by a fine-grained experimenter's superprocess, $I_{\emptyset}$ does not induce a valid monotone. It is this property which makes $\mathsf{Q}_{\hat{\mathbbm{W}}}$ suitable for quantifying the amenability of noise processes to noise reduction techniques. 
\begin{theorem} \label{thm:nonmonotonicity}
In the resource theory of temporal resolution $\mathsf{Q}_{\hat{\mathbbm{W}}}$, input-output mutual information 
\begin{equation} \label{eq:notmonotone}
   I_{\emptyset}:\mathsf{T}_{\hat{\mathbbm{W}}} \rightarrow \mathbbm{R}_{\geq 0} , \quad I_{\emptyset}(\mathbf{T}_{\hat{n}}) := I\big(\llbracket \mathbf{T}_{\hat{n}} | \mathbf{I}_{\hat{n}} \rrbracket \big)
\end{equation}
can be increased by free transformations (non-monotonicity) iff Cor.~\ref{cor:irreversibilitymonotone} is realised as a strict inequality for $\hat{m}=\emptyset$.
\end{theorem}
\begin{proof}
Increasing the input-output mutual information $I_{\emptyset}$ is only possible if it is not a monotone, i.e. for all $\mathbf{T}_{\hat{n}} \in \mathsf{T}_{\hat{\mathbbm{W}}}$ there exists a transformation $\mathbf{Z}_{\hat{n}}$ such that $I_{\emptyset}(\mathbf{T}_{\hat{n}}) < I_{\emptyset}\big( \llbracket \mathbf{T}_{\hat{n}} | \mathbf{Z}_{\hat{n}} \big)$. Consider Eq.~\eqref{eq:monotoneirreversibility} in the case of $M=I$ and $\hat{m}= \emptyset$. Then, Cor.~\ref{cor:irreversibilitymonotone} reduces to
\begin{equation} \label{eq:MInonmonotonicity}
\sup_{\mathbf{Z}_{\hat{n}} \in \mathsf{Z}_{\hat{n}}} I_{\emptyset}\big( \llbracket \mathbf{T}_{\hat{n}} | \mathbf{Z}_{\hat{n}} \big) \geq I_{\emptyset}(\mathbf{T}_{\hat{n}}).
\end{equation}
Observe that the supremum on the right hand side disappears because free superprocesses in $\mathsf{Q}_{\hat{\mathbbm{W}}}$ with no intermediate interventions $\llbracket \mathbf{I}_{{\hat{n}}}  | \mathbf{Z}_{\hat{n}} | \mathbf{I}_{{\hat{n}}} \rrbracket$ are just memoryless supermaps, and cannot increase mutual information, i.e.,
\begin{equation}
  \sup_{\mathbf{Z}_{\hat{n}} \in \mathsf{Z}_{\hat{n}}}I\Big(\Big\llbracket \llbracket \mathbf{T}_{\hat{n}} | \mathbf{I}_{\hat{n}} \rrbracket \Big| \llbracket \mathbf{I}_{{\hat{n}}}  | \mathbf{Z}_{\hat{n}} | \mathbf{I}_{{\hat{n}}} \rrbracket \Big) = I_{\emptyset}(\mathbf{T}_{\hat{n}}).  
\end{equation}
Thus, if Cor.~\ref{cor:irreversibilitymonotone} holds as a strict inequality, then there will exist a superprocess $\mathbf{Z}_{\hat{n}}$ satisfying the strict inequality for Eq.~\eqref{eq:MInonmonotonicity}. Conversely, if there exists any $\mathbf{Z}_{\hat{n}}$ such that $I_{\emptyset}(\mathbf{T}_{\hat{n}}) < I_{\emptyset}\big( \llbracket \mathbf{T}_{\hat{n}} | \mathbf{Z}_{\hat{n}} \big)$, Cor.~\ref{cor:irreversibilitymonotone} will be a strict inequality for the same reason that memoryless supermaps cannot increase mutual information.
\end{proof}

\subsection{Partitioning of Total Mutual Information Into $M$ and $N$} \label{sec:splitI}

To see that $I$ can be partitioned into two distinct contributions, as described in Eq.~\eqref{eq:mutualinformationequality}, consider the difference in the mutual information monotone between a process $\mathbf{T}_{\hat{n}}$ and its nearest Markovian one $\mathbf{T}^{\text{Mkv}}_{\hat{n}}$
\begin{equation}
\begin{aligned}
    I(\mathbf{T}_{\hat{n}}) -I(\mathbf{T}^{\text{Mkv}}_{\hat{n}}) =& S\left( \mathbf{T}_{\hat{n}} \Vert \mathbf{T}^{\text{marg}}_{\hat{n}} \right) -S\left( \mathbf{T}^{\text{Mkv}}_{\hat{n}} \| \mathbf{T}^{\text{marg}}_{\hat{n}} \right) \\ 
    =&  \left(S\left( \mathbf{T}^{\text{marg}}_{\hat{n}} \right) - S\left( \mathbf{T}_{\hat{n}} \right) \right)\\ & \ -\left(S\left( \mathbf{T}^{\text{marg}}_{\hat{n}} \right) - S\left( \mathbf{T}^\text{Mkv}_{\hat{n}} \right) \right)  \\ 
    = & S\left( \mathbf{T}_{\hat{n}} \| \mathbf{T}^{\text{Mkv}}_{\hat{n}} \right) = N(\mathbf{T}_{\hat{n}}).
    \end{aligned}
\end{equation}
The above holds because $\mathbf{T}^{\text{Mkv}}_{\hat{n}}$ is a product of marginals of $\mathbf{T}_{\hat{n}}$, while $\mathbf{T}^{\text{marg}}_{\hat{n}}$ is a product of marginals of $\mathbf{T}^{\text{Mkv}}_{\hat{n}}$. Hence, the mutual information monotone can be partitioned into two contributions
\begin{equation} 
    I(\mathbf{T}_{\hat{n}}) =M(\mathbf{T}_{\hat{n}})+N(\mathbf{T}_{\hat{n}}).
\end{equation}

\subsection{Monotonicity of $I$, $M$, and $N$ under Free Superprocesses of $\mathsf{Q}_{\hat{\mathbbm{W}}}$} \label{sec:I-N-monotonicity}

Here we show that $I$, $M$, and $N$ are monotonic under free superprocesses of $\mathsf{Q}_{\hat{\mathbbm{W}}}$. Moreover, Thm.~\ref{thm:contractivity} guarantees that they are also monotonic under coarse-grainings, and hence the transformations of $\mathsf{Q}_{\hat{\mathbbm{W}}}$ more broadly, since all free transformations can be represented as a combination thereof.

\begin{theorem} \label{thm:IMkvN_nonmonotonicity}
Markov information $M$ and Non-Markovianity $N$, as defined in Eq.~\eqref{eq:marginals} and Eq.~\eqref{eq:Imonotoneresolution}, respectively, are monotonic under the free superprocesses of $\mathsf{Q}_{\hat{\mathbbm{W}}}$.
\end{theorem}

\begin{proof}
In resource theory $\mathsf{Q}_{\hat{\mathbbm{W}}}$, all experimental interventions are temporally local, which means that they cannot increase correlations between temporally separated subsystems. To demonstrate this for the case of $M$, a free superprocess $\mathbf{Z}_{\hat{n}}$ is applied to $\mathbf{T}_{\hat{n}}$, and Markov information takes the form
\begin{equation}
\begin{aligned}
   &M\big( \llbracket \mathbf{T}_{\hat{n}} | \mathbf{Z}_{\hat{n}} \big) \\ =& S\Big( \left(\llbracket \mathbf{T}_{\hat{n}} | \mathbf{Z}_{\hat{n}} \right)^\text{Mkv} \Big\Vert  \left( \llbracket \mathbf{T}_{\hat{n}} | \mathbf{Z}_{\hat{n}} \right)^\text{marg} \Big) \\
    =& S\left( \bigotimes_{j=1}^{n+1} \text{tr}_{\bar{j}} \{ \llbracket \mathbf{T}_{\hat{n}} | \mathbf{Z}_{\hat{n}} \}  \Bigg\Vert   \bigotimes_{k=1}^{2(n+1)} \text{tr}_{\bar{k}} \{ \llbracket \mathbf{T}_{\hat{n}} | \mathbf{Z}_{\hat{n}} \}  \right) \\
    =& S\left(\left\llbracket \bigotimes_{j=1}^{n+1} \text{tr}_{\bar{j}} \{ \mathbf{T}_{\hat{n}} \} \right| \mathbf{Z}_{\hat{n}}    \Bigg\Vert \left\llbracket \bigotimes_{k=1}^{2(n+1)} \text{tr}_{\bar{k}} \{ \mathbf{T}_{\hat{n}} \} \right\vert \mathbf{Z}_{\hat{n}} \right) \\
    =&  S\Big(\llbracket \mathbf{T}_{\hat{n}}^\text{Mkv} | \mathbf{Z}_{\hat{n}} \Big\Vert \llbracket \mathbf{T}_{\hat{n}}^\text{marg} | \mathbf{Z}_{\hat{n}} \Big) \\
    \leq & S\big( \mathbf{T}_{\hat{n}}^\text{Mkv}  \big\Vert  \mathbf{T}_{\hat{n}}^\text{marg} \big) = M(\mathbf{T}_{\hat{n}}).
\end{aligned}
\end{equation}
$j$ indexes the free evolutions of the process tensor, while $h$ labels the input and output Hilbert spaces. The third line follows from the second because free superprocesses are temporally local in $\mathsf{Q}_{\hat{\mathbbm{W}}}$, implying that $\left(\llbracket \mathbf{T} | \mathbf{Z} \right)^{\text{Mkv}}=\llbracket \mathbf{T}^{\text{Mkv}} | \mathbf{Z}$, and $\left(\llbracket \mathbf{T} | \mathbf{Z} \right)^{\text{marg}}=\llbracket \mathbf{T}^{\text{marg}} | \mathbf{Z}$. The final line uses the contractivity of relative entropy under free superprocesses~\cite{resourcetheoriesofmultitime}. The same type of argument can be applied to show that $I$ and $N$ are monotones, so we will not repeat it. It should be noted that this feature does not hold for all resource theories of temporal resolution. For other theories where some degree of communication through time is allowed within a superprocess, the relevant monotone will be some other more restricted property like violation of direct cause relations.
\end{proof}

\subsection{Invariance of $I$, $M$, and $N$ under Free Superprocesses of $\mathsf{D}_{\hat{\mathbbm{W}}}$} \label{sec:I_invariance}

We present a brief proof that total mutual information, Markov information, and non-Markovianity are not only monotonic, but also \emph{invariant} under DD sequences, and the free superprocesses of $\mathsf{D}_{\hat{\mathbbm{W}}}$ more broadly. 
\begin{theorem}
$I$, $M$, and $N$ are invariant under the free superprocesses of $\mathsf{D}_{\hat{\mathbbm{W}}}$.
\end{theorem}
\begin{proof}
We show this invariance directly for $I$, using the same argument as in our proof of Thm.~\ref{thm:IMkvN_nonmonotonicity}:
\begin{equation}
\begin{aligned}
   I\big( \llbracket \mathbf{T}_{\hat{n}} | \mathbf{Z}_{\hat{n}} \big)  =& S\Big( \mathbf{T}_{\hat{n}}   \Big\Vert  \left( \llbracket \mathbf{T}_{\hat{n}} | \mathbf{Z}_{\hat{n}} \right)^\text{marg} \Big) \\
    =& S\left( \llbracket \mathbf{T}_{\hat{n}} | \mathbf{Z}_{\hat{n}}  \Bigg\Vert   \bigotimes_{h=1}^{2(n+1)} \text{tr}_{\bar{h}} \{ \llbracket \mathbf{T}_{\hat{n}} | \mathbf{Z}_{\hat{n}} \}  \right) \\
    =& S\left( \llbracket \mathbf{T}_{\hat{n}} | \mathbf{Z}_{\hat{n}}   \Bigg\Vert \left\llbracket \bigotimes_{h=1}^{2(n+1)} \text{tr}_{\bar{h}} \{ \mathbf{T}_{\hat{n}} \} \right\vert \mathbf{Z}_{\hat{n}} \right) \\
    =&  S\Big(\llbracket \mathbf{T}_{\hat{n}} | \mathbf{Z}_{\hat{n}} \Big\Vert \llbracket \mathbf{T}_{\hat{n}}^\text{marg} | \mathbf{Z}_{\hat{n}} \Big) \\
    = &  I(\mathbf{T}_{\hat{n}}).
\end{aligned}
\end{equation}
As with Thm.~\ref{thm:IMkvN_nonmonotonicity}, the second and third line are equal because the free superprocesses of $\mathsf{D}_{\hat{\mathbbm{W}}} \subset \mathsf{Q}_{\hat{\mathbbm{W}}}$ are temporally local. The invariance of $I^{Mkv}$ and $N$ follow from the same argument, so we will not repeat it here.
\end{proof}

\subsection{Numerical Model}

The model used to generate Fig.~\ref{fig:numerics} consists of a two-dimensional system $s$ and an environment $e$ with a (Haar) randomly sampled initial pure state $\rho^{e}_0$, undergoing evolution for duration $t$ under a randomly sampled $s$-$e$ Hamiltonian $H^{se}$ whose operator norm is normalised to unity. Specifically, these Hamiltonians are sampled by producing matrices $K$ whose entries are uniformily distributed in $[0,1]$, taking $H^{se}$ as the combination $K + K^\dagger$, and subsequently normalising it by its operator norm. We produce 20 samples of $\rho^{e}_0$ and $H^{se}$, generating an ensemble of 20 sets of underlying dynamics. For a given set of underlying dynamics, we study three levels of temporal resolution: $\mathbf{T}_{\hat{15}}$, $\hat{15}=\{ t/16,\dots,15t/16 \}$, $\mathbf{T}_{\hat{3}}$, $\hat{3}=\{ t/4,2t/4,3t/4 \}$, and $\mathbf{T}_{\emptyset}$, corresponding to a fine-grained process, an intermediate process, and a channel. DD can be applied at the fine-grained level, and/or at the coarse-grained level, with doing both corresponding to CDD. Using these three levels of temporal resolution, we can compute changes in monotones, e.g. $\Delta I= I_{\hat{3}|\mathbf{Z}_{15}}(\mathbf{T}_{\hat{15}})-I_{\hat{3}}(\mathbf{T}_{\hat{15}})$, as well as the increase in channel-level mutual information $ I_{\emptyset|\mathbf{Z}_{15}}(\mathbf{T}_{\hat{15}})-I_{\emptyset}(\mathbf{T}_{\hat{15}})$.

Similarly we can apply MODD with one level of concatenation. This optimisation begins with an SDP optimisation of control pulses (by maximizing the largest eigenvalue of the corresponding resulting channel) for a small process tensor with intermediate times $\{ t/16,2t/16,3t/16 \}$. Using this sequence, a maximally mixed state is placed into the input of the dynamics at time $=0$ to generate an environment state at $t/4$ and the procedure is repeated for intermediate times $\{ 5t/16,6t/16,7t/16 \}$ etc. until all of $\hat{15}$ has optimised operations. The multitimescale aspect appears when we coarse-grain and repeat the procedure for $\hat{3}$.

\section*{Acknowledgements}

G. B. is supported by an Australian Government Research Training Program (RTP) Scholarship. S. M. acknowledges funding from the European Union’s Horizon 2020 research and innovation programme under the Marie Sk{\l}odowska Curie grant agreement No 801110, and the Austrian Federal Ministry of Education, Science and Research (BMBWF). The opinions expressed in this publication are those of the authors, the EU Agency is not responsible for any use that may be made of the information it contains. K.M. is supported through Australian Research Council Future Fellowship FT160100073 and Discovery Project DP210100597. K.M. was recipients of the International Quantum U Tech Accelerator award by the US Air Force Research Laboratory.

\bibliographystyle{apsrev4-1_custom}
\bibliography{refs}

\onecolumn
\section{Supplementary Material} 

Here we detail a few of the more technical aspects of this investigation. We begin with a proof of the contractivity of relative entropy under temporal coarse-grainings in Sup.~\ref{sup:contractivity}. Subsequently in Sup.~\ref{sup:productstructure}, we discuss the parallel and sequentially product structure of process tensors. In Sup.~\ref{sup:markovianisation} we present an argument based on that of Ref.~\cite{canquantummarkovevolutions}, showing that dynamical decoupling removes all non-Markovianity from a process under ideal conditions. Finally, in Sup.~\ref{sup:subtheories} we consider a sub-theory hierarchy of $\mathsf{Q}_{\hat{\mathbbm{W}}}$ delineating which noise suppression techniques remain available after applying specific additional restrictions. We have also included a notation summary for quick reference in Sup.~\ref{sup:notationsummary}.

\subsection{Contractivity of Relative Entropy Under Coarse-Grainings} \label{sup:contractivity}

For temporal coarse-graining to be appended to a pre-existing resource theory without disturbing its useful structure, we expect that pre-existing monotones under the free superprocesses should also be monotones under the free coarse-grainings (and by Lem.~\ref{lem:representation} any combination thereof). Hence, to ensure that our mutual information based monotones $I$, $M$, and $N$ remain valid after the inclusion of coarse-grainings, we ask that temporal coarse-graining does not increase the relative entropy $S(\mathbf{T}_{\hat{n}}\Vert \mathbf{R}_{\hat{n}}):=\text{tr} \{ \mathbf{T}_{\hat{n}} \log(\mathbf{T}_{\hat{n}}) \} - \text{tr} \{ \mathbf{T}_{\hat{n}} \log(\mathbf{R}_{\hat{n}}) \}$ between any two process tensors $\mathbf{T}_{\hat{n}}$ and $\mathbf{R}_{\hat{n}}$.
\begin{theorem} \label{thm:contractivity} 
Given any two process tensors $\mathbf{T}_{\hat{n}}, \mathbf{R}_{\hat{n}} \in \mathsf{T}_{\hat{n}}$, temporal coarse-graining $\mathbf{I}_{\hat{n} \setminus \hat{m}} :\mathsf{S}_{\hat{n}} \rightarrow \mathsf{S}_{\hat{m}}$ for all $\emptyset \subseteq \hat{m} \subseteq \hat{n}$ satisfies
\begin{equation}
    S( \mathbf{T}_{\hat{n}} \Vert \mathbf{R}_{\hat{n}}) \geq S\big( \llbracket \mathbf{T}_{\hat{n}} | \mathbf{I}_{\hat{n} \setminus \hat{m}} \rrbracket \big\Vert \llbracket \mathbf{R}_{\hat{n}} | \mathbf{I}_{\hat{n} \setminus \hat{m}} \rrbracket \big) ,
\end{equation}
for relative entropy $S(\cdot \Vert \cdot )$. 
\end{theorem}
\begin{proof}
To show contractivity of relative entropy under a coarse-graining functor $\mathbf{I}_{\hat{n} \setminus \hat{m}}$, we use a previously proved theorem requiring only positivity (not complete positivity), trace preservation from the domain to the image, and linearity~\cite{monotonicityofthequantumrelative}.
\begin{lemma} \label{lem:conditionsforcontractivity}
Let $\Phi:\mathcal{B}(\mathcal{H}) \rightarrow \mathcal{B}(\mathcal{H'})$ be a positive trace-preserving linear map, where $\mathcal{H}$ and $\mathcal{H}'$ are separable Hilbert spaces. Then for any positive semidefinite operators $a,b \in \mathcal{B}_{\geq 0}(\mathcal{H})$,
\begin{equation}
    S( a \Vert b) \geq S\big( \Phi(a) \big\Vert \Phi(b) \big) .
\end{equation}
$\mathcal{B}(\mathcal{H})$ represents trace-class operators on $\mathcal{H}$, while $\mathcal{B}_{\geq 0}(\mathcal{H})$ represents the positive semi-definite ones.
\end{lemma}

The remainder of this proof is dedicated to showing that coarse-graining satisfies the conditions of Lem.~\ref{lem:conditionsforcontractivity}. First, we show positivity and trace preservation. Throughout this work, we have used only one symbol to represent each object, and whether the object is represented by a superoperator or a Choi state is implied by context. Here we explicitly write a process $\mathbf{T}_{\hat{n}}$ in the Choi representation (with unit normalisation) as
\begin{equation}  \label{eq:choistaten}
\mathbf{T}_{\hat{n}}= \text{tr}_e \Bigg\{  \functioncomposition\limits_{t_j \in \hat{n}} \Big( S^{s,{\mathtt{o}_{t_j}}} \circ \mathcal{T}_{{t_{j+1}}:{t_j}}^{s e} \circ S^{s,{\mathtt{o}_{t_j}}}  \Big)   \tensorcomposition\limits_{{t_j} \in \hat{n}}  \left( \psi^{{\mathtt{o}_{t_j}},{\mathtt{i}_{t_j}}} \right) \otimes \rho^{e}_{0}  \Bigg\},
\end{equation}
where $S^{\alpha,\beta}$ is a swap operation between subsystems $\alpha$ and $\beta$, $\psi$ is a maximally entangled bipartite state, $\mathtt{o}$ and $\mathtt{i}$ index the output and input Hilbert spaces of each leg of the process tensor respectively. With this indexing, $s={\mathtt{o}_0}$, hence $S^{(s,{\mathtt{o}_0})}=\mathcal{I}^{s}$. Each $\mathcal{T}_{{t_{j+1}}:{t_j}}^{s e}$ represents the free evolution from time $t_{j}$ to $t_{j+1}$ All operations used here are completely positive and trace preserving, hence the Choi state representation of $\mathbf{T}_{\hat{n}}$ is a valid quantum state.

Turning attention to $\mathbf{T}_{\hat{m}} = \llbracket \mathbf{T}_{\hat{n}} | \mathbf{I}_{\hat{n} \setminus \hat{m}} \rrbracket$, the Choi state representation of $\mathbf{T}_{\hat{m}}$ can be written as
\begin{equation} \label{eq:choistatem}
\mathbf{T}_{\hat{n}}= \text{tr}_e \Bigg\{  \functioncomposition\limits_{t_j \in \hat{m}} \left( S^{s,{\mathtt{o}_{t_j}}}  \functioncomposition\limits_{t_j < t_k <t_{j+1} \in \hat{n} \setminus \hat{m} } \left( \mathcal{T}_{{t_{k+1}}:{t_k}}^{se} \right) \circ \mathcal{T}_{{t_{j+1}}:{t_j}}^{s e} \functioncomposition\limits_{t_{j-1}< t_k <t_j \in \hat{n} \setminus \hat{m} } \left( \mathcal{T}_{{t_{k+1}}:{t_k}}^{se} \right) \circ S^{s,{\mathtt{o}_{t_j}}}  \right)   \tensorcomposition\limits_{{t_j} \in \hat{m}}  \left( \psi^{{\mathtt{o}_{t_j}},{\mathtt{i}_{t_j}}} \right) \otimes \rho^{e}_{0}   \Bigg\}.
\end{equation}
The difference between Eq.~\eqref{eq:choistaten} and Eq.~\eqref{eq:choistatem} is that in the former there is one $s$-$e$ evolution $\mathcal{T}_{{t_{j+1}}:{t_j}}^{s e}$ per step in $\hat{n}$, while the latter has more evolutions than steps. These correspond to the times $t_k$ in $\hat{n}$ but not $\hat{m}$, lying in the intervals $t_{j-1}< t_k <t_j$ and $t_j < t_k <t_{j+1}$. Still, all actions on the original maximally entangled states are completely positive and trace preserving, so the result is a valid Choi $m$-step state with the same trace (or appropriately re-scaled depending on normalisation convention). Hence, $\mathbf{I}_{\hat{n} \setminus \hat{m}} :\mathsf{T}_{\hat{n}} \rightarrow \mathsf{T}_{\hat{m}}$ is positive and trace preserving (or re-scaling) from its domain to its image.

Linearity can be seen using the Choi isomorphism in the opposite direction, taking two process tensors $\mathbf{T}_{\hat{n}}$ and $\mathbf{R}_{\hat{n}}$, with the same set of times for interventions
\begin{equation}
\begin{aligned}
    \llbracket  \mathbf{T}_{\hat{n}} + \mathbf{R}_{\hat{n}} | \mathbf{I}_{\hat{n} \setminus \hat{m}} \rrbracket & = \text{tr}_{i\in \hat{n} \setminus \hat{m}} \left\{ \left(\mathbf{T}_{\hat{n}} + \mathbf{R}_{\hat{n}} \right) \left( \mathbbm{1}^{\text{i}_0}  \otimes \mathbf{I}_{\hat{n} \setminus \hat{m}} \otimes \mathbbm{1}^{\text{o}_n}  \right) \right\} \\
    &= \text{tr}_{i\in \hat{n} \setminus \hat{m}} \left\{ \mathbf{T}_{\hat{n}}  \left( \mathbbm{1}^{\text{i}_0}  \otimes \mathbf{I}_{\hat{n} \setminus \hat{m}} \otimes \mathbbm{1}^{\text{o}_n}  \right) \right\} + \text{tr}_{i\in \hat{n} \setminus \hat{m}} \left\{ \mathbf{R}_{\hat{n}}  \left( \mathbbm{1}^{\text{i}_0}  \otimes \mathbf{I}_{\hat{n} \setminus \hat{m}} \otimes \mathbbm{1}^{\text{o}_n}  \right) \right\} \\
    &= \llbracket  \mathbf{T}_{\hat{n}}  | \mathbf{I}_{\hat{n} \setminus \hat{m}} \rrbracket + \llbracket   \mathbf{R}_{\hat{n}} | \mathbf{I}_{\hat{n} \setminus \hat{m}} \rrbracket.
 \end{aligned}
\end{equation}

Using Lem.~\ref{lem:conditionsforcontractivity}, positivity, trace preservation, and linearity are sufficient to conclude the contractivity of relative entropy under $\mathbf{I}_{\hat{n} \setminus \hat{m}}$.
\end{proof}

\subsection{Parallel and Sequential Product Structures} \label{sup:productstructure}

An important aspect of channel resource theories is that they have a notions of combining and discarding channels. This inclusion broadens the scope of what can be achieved in those theories, enabling tasks like catalytic conversion, asymptotic conversion, and much more. Here we investigate the consequences of including analogous notions for process tensor resources. Combining process tensors could mean to take two separate experiments and consider them concurrently, or subsequently. Similarly, discarding process tensors might correspond to ignoring the results of an experiment, or terminating an experiment early. We should expect that these notions reduce to the channel notions after process tensors are coarse-grained to have no intermediate interventions.

Aside from basic considerations of closedness on the set of resource objects, there aren't actually any restrictions on what can or cannot be defined as free transformations in a resource theory, since this will simply result in different sets of free resource objects, and different monotones. However, one must still be careful when deciding what to include in the set of free transformations, because not all mathematically valid resource theories will be useful for solving physical problems. The purpose of this section is to verify that combining and discarding process tensors are indeed sensible inclusions to resource theories of temporal resolution.

\subsubsection{Parallel Composition} 

Consider the example of two experimenters in two different laboratories performing their own experiments. In this situation, each experimenter has their own process tensor, and the global `experiment' is a tensor product of what happens in the two laboratories. This is what the operation of parallel composition of process tensors physically corresponds to. Depending on the specifics of the given resource theory, these two experimenters may be able to communicate, or exchange resources. 

The two experiments can be expressed as $\mathbf{T}^1_{\hat{n}}=\text{tr}_{e_1} \left\{ \mathcal{T}^{s_1 e_1}_{t:t_{n}} \circ_{e_1} \dots \circ_{e_1} \mathcal{T}^{s_1 e_1}_{t_{1}:0} \circ_{e_1} \tau^{e_1} \right\} $ and $\mathbf{T}^2_{\hat{m}}=\text{tr}_{e_2} \left\{ \mathcal{T}^{s_2 e_2}_{t:t'_{m}} \circ_{e_2} \dots \circ_{e_2} \mathcal{T}^{s_2 e_2}_{t'_1:0} \circ^{e_2} \tau^{e_2} \right\}$, parallel composition is defined as 
\begin{equation} \label{eq:parallelcomp}
    \mathbf{T}^1_{\hat{n}} \otimes \mathbf{T}^2_{\hat{m}}:= \text{tr}_{e_1} \left\{ \mathcal{T}^{s_1 e_1}_{t:t_{n}} \circ_{e_1} \dots \circ_{e_1} \mathcal{T}^{s_1 e_1}_{t_{1}:0} \circ_{e_1} \tau^{e_1} \right\} \otimes \text{tr}_{e_2} \left\{ \mathcal{T}^{s_2 e_2}_{t:t'_{m}} \circ_{e_2} \dots \circ_{e_2} \mathcal{T}^{s_2 e_2}_{t'_1:0} \circ^{e_2} \tau^{e_2} \right\},
\end{equation}
where $\mathcal{T}^{s_i e_i}_{t_{\alpha+1}:t_{\alpha}}$ is free evolution on the $i$th process tensor from the $\alpha$th intervention to the $\alpha+1$th intervention, $\tau^e$ and $\sigma^e$ are initial environment states, and $\circ_{e_i}$ is composition over the environment $i$ alone. An example of parallel compositon is shown in Fig.\ref{fig:parallelcompositionexample}. Observe that setting each experiment to have no intermediate interventions $\hat{n}=\hat{m}=\emptyset$ reduces Eq.~\eqref{eq:parallelcomp} to the regular notion of parallel composition of maps.
\begin{figure}[ht!]
\centering
\begin{tikzpicture}[scale=0.3]

\draw[black, very thick,solid] (-2.5,1) -- (4.25,1);
\draw[black, very thick,solid] (5.75,1) -- (20.25,1);
\draw[black, very thick,solid] (21.75,1) -- (28.5,1);

\draw[black, very thick,solid] (-2.5,-3) -- (12.25,-3);
\draw[black, very thick,solid] (13.75,-3) -- (20.25,-3);
\draw[black, very thick,solid] (21.75,-3) -- (28.5,-3);

  \draw[mypurple,fill=mypurplefill, thick,solid,rounded corners=4] (0-0.1,1) -- (0-0.1,0-0.1) -- (2+0.1,0-0.1) -- (2+0.1,2-0.1) -- (4-0.1,2-0.1) -- (4-0.1,0-0.1)  -- (6+0.1,0-0.1) -- (6+0.1,2-0.1) -- (8-0.1,2-0.1) -- (8-0.1,0-0.1) -- (18+0.1,0-0.1) -- (18+0.1,2-0.1) -- (20-0.1,2-0.1) -- (20-0.1,0-0.1) -- (22+0.1,0-0.1) -- (22+0.1,2-0.1) -- (24-0.1,2-0.1)  -- (24-0.1,0-0.1) -- (26+0.1,0-0.1) -- (26+0.1,2-0.1) -- (28-0.1,2-0.1) -- (28-0.1,-4+0.1) -- (26+0.1,-4+0.1) -- (26+0.1,-2+0.1) -- (24-0.1,-2+0.1) -- (24-0.1,-4+0.1) -- (22+0.1,-4+0.1) -- (22+0.1,-2+0.1) -- (20-0.1,-2+0.1) -- (20-0.1,-4+0.1) -- (18+0.1,-4+0.1) -- (18+0.1,-2+0.1) -- (16-0.1,-2+0.1) -- (16-0.1,-4+0.1) -- (14+0.1,-4+0.1) -- (14+0.1,-2+0.1) -- (12-0.1,-2+0.1) -- (12-0.1,-4+0.1) -- (10+0.1,-4+0.1) -- (10+0.1,-2+0.1) -- (0-0.1,-2+0.1) -- (0-0.1,-4+0.1) -- (-2+0.1,-4+0.1) -- (-2+0.1,2-0.1) -- (0-0.1,2-0.1) -- (0-0.1,1)    ; 

\draw[myred,fill=myredfill, thick,solid,rounded corners=4] (0+0.1,3) -- (0+0.1,4) -- (26-0.1,4) -- (26-0.1,2+0.1) -- (26-0.1,0+0.1) -- (24+0.1,0+0.1) -- (24+0.1,2+0.1)-- (18-0.1,2+0.1) -- (18-0.1,0+0.1) -- (8+0.1,0+0.1) -- (8+0.1,2+0.1) -- (2-0.1,2+0.1) -- (2-0.1,0+0.1) -- (0+0.1,0+0.1) -- (0+0.1,3)   ;

\draw[myred,fill=myredfill, thick,solid,rounded corners=4] (0+0.1,-2-3) -- (0+0.1,-2-4) -- (26-0.1,-2-4) -- (26-0.1,-2-2-0.1) -- (26-0.1,-2-0.1) -- (24+0.1,-2-0.1) -- (24+0.1,-2-2-0.1)-- (18-0.1,-2-2-0.1) -- (18-0.1,-2-0-0.1) -- (16+0.1,-2-0-0.1) -- (16+0.1,-2-2-0.1) -- (10-0.1,-2-2-0.1) -- (10-0.1,-2-0-0.1) --  (0+0.1,-2-0-0.1) -- (0+0.1,-2-3)   ;

\draw[] (13,2) node[rotate=0] {\large $\mathbf{T}_{\hat{n}}$};
\draw[] (13,-5) node[rotate=0] {\large $\mathbf{T}_{\hat{m}}$};
\draw[] (13,-1) node[rotate=0] {\large $\mathbf{Z}_{\hat{n} \otimes \hat{m}}$};

\draw[] (5,1) node[rotate=0] { $t_1$};
\draw[] (13,-3) node[rotate=0] { $t'_1$};

\draw[] (21,1) node[rotate=0] { $t_2$};
\draw[] (21,-3) node[rotate=0] { $t'_2$};

\end{tikzpicture}
\caption{An example of a bipartite experiment, where the experimenter is given a tensor product $\mathbf{T}_{\hat{n}} \otimes \mathbf{S}_{\hat{m}}$, and can interact with each subsystem at non-identical sets of times $\hat{n} =\{ t_1,t_2 \}$ and $\hat{m}=\{ t'_1,t'_2 \}$. In this case $t_1 \neq t'_1$ but $t_2 = t'_2$. Without the experimenter's actions these two process tensors are independent, but in general superprocesses can create a correlated global process tensor.}  \label{fig:parallelcompositionexample}
\end{figure}
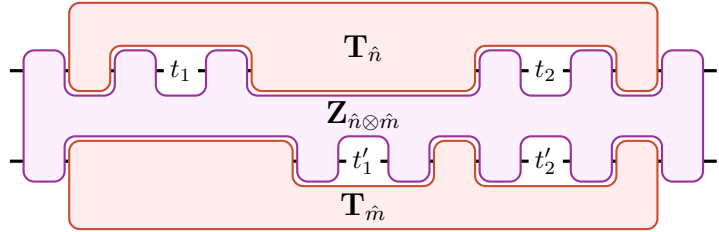

For this to be considered parallel composition as opposed to sequential composition, both process tensors must occur concurrently, i.e. in the window $(0,t)$. If one process is shorter than the other, sequential composition of a free process (Sec.~\ref{sec:sequentialcomp}) for that duration is required to make the two compatible. 

Having defined what is meant by parallel composition, we can now show that induces a well-defined tensor product structure that also respects monotones. Consider the impact of defining parallel composition with a free resource -- $\mathbf{T}_{\hat{n}} \mapsto \mathbf{T}_{\hat{n}} \otimes \mathbf{S}_{\hat{m}}$ for $\mathbf{S}_{\hat{m}} \in \mathsf{T}^{\text{F}}_{\hat{m}}$ -- as a class of free resource transformation. Observe that $\mathbf{S}_{\hat{m}}$ is a free resource in the sub-theory for fixed $\hat{m}$, which reduces to the free resources of the full theory if we specify $\hat{m}=\emptyset$. 

To verify that this addition respects the structure of resource theories of temporal resolution, we show that trace distance and relative entropy remain as appropriate monotones.
\begin{theorem} \label{thm:parallelcompositiontracedistance}
For any process tensor $\mathbf{T}_{\hat{n}} \in \mathsf{T}_{\hat{\mathbbm{W}}}$, trace distance to the nearest free process in the sub-theory for fixed ${\hat{n}}$, $\inf_{\mathbf{T}^{\text{F}}_{\hat{n}} \in \mathsf{T}^\text{F}_{\hat{n}}}D(\mathbf{T}_{\hat{n}},\mathbf{T}^{\text{F}}_{\hat{n}})$ is contractive under the parallel composition with free resources.
\end{theorem}
\begin{proof}
Let $\mathbf{T}_{\hat{n}} \in \mathsf{T}_{\hat{n}}$ and $\mathbf{S}_{\hat{m}} \in \mathsf{T}^\text{F}_{\hat{m}}$. The trace distance of the parallel composition $\mathbf{T}_{\hat{n}} \otimes \mathbf{S}_{\hat{m}}$ to its nearest two-laboratory free process $\mathbf{R}^{\text{F}}_{\hat{n} \otimes \hat{m}}$ with $\hat{n}$ and $\hat{m}$ intermediate interventions is
\begin{equation}
\begin{aligned}
     \inf_{\mathbf{R}^{\text{F}}_{\hat{n} \otimes \hat{m}}   \in \mathsf{T}^{\text{F}}_{\hat{n} \otimes \hat{m}} }& D(\mathbf{T}_{\hat{n}} \otimes \mathbf{S}_{\hat{m}},\mathbf{R}^{\text{F}}_{\hat{n} \otimes \hat{m}} ) \\
    \leq  \inf_{\mathbf{T}^{\text{F}}_{\hat{n}} \in \mathsf{T}^\text{F}_{\hat{n}}}& D(\mathbf{T}_{\hat{n}} \otimes \mathbf{S}_{\hat{m}},\mathbf{T}^{\text{F}}_{\hat{n}} \otimes \mathbf{S}_{\hat{m}}) \\
    \leq  \inf_{\mathbf{T}^{\text{F}}_{\hat{n}} \in \mathsf{T}^\text{F}_{\hat{n}}} & D(\mathbf{T}_{\hat{n}},\mathbf{T}^{\text{F}}_{\hat{n}}) . 
\end{aligned}
\end{equation}
This parallel composition is equivalent jointly considering two separate experiments. The second line relaxes the infinum from being over both experiments jointly, to only the one corresponding to $\mathsf{T}_{\hat{n}}$. Hence, the second line is greater than the first. The final line follows from subadditivity of trace distance. 
\end{proof}
A similar argument can be used for relative entropy, but invoking additivity under independent subsystems, rather than subadditivity.
\begin{theorem} \label{parallelcompositionrelativeentropy}
For any process tensor $\mathbf{T}_{\hat{n}} \in \mathsf{T}_{\hat{\mathbbm{W}}}$, relative entropy to the nearest free process in the sub-theory for fixed ${\hat{n}}$, $\inf_{\mathbf{T}^{\text{F}}_{\hat{n}} \in \mathsf{T}^\text{F}_{\hat{n}}}S(\mathbf{T}_{\hat{n}}\|\mathbf{T}^{\text{F}}_{\hat{n}})$ is contractive under the parallel composition with free resources.
\end{theorem}
\begin{proof}
Let $\mathbf{T}_{\hat{n}} \in \mathsf{T}_{\hat{n}}$ and $\mathbf{S}_{\hat{m}} \in \mathsf{T}^\text{F}_{\hat{m}}$. The relative entropy of the parallel composition $\mathbf{T}_{\hat{n}} \otimes \mathbf{S}_{\hat{m}}$ to its nearest two-laboratory free process $\mathbf{R}^{\text{F}}_{\hat{n} \otimes \hat{m}}$ with $\hat{n}$ and $\hat{m}$ intermediate interventions is
\begin{equation}
\begin{aligned}
     \inf_{\mathbf{R}^{\text{F}}_{\hat{n} \otimes \hat{m}}   \in \mathsf{T}^{\text{F}}_{\hat{n} \otimes \hat{m}} }& S(\mathbf{T}_{\hat{n}} \otimes \mathbf{S}_{\hat{m}}\|\mathbf{R}^{\text{F}}_{\hat{n} \otimes \hat{m}} ) \\
    \leq  \inf_{\mathbf{T}^{\text{F}}_{\hat{n}} \in \mathsf{T}^\text{F}_{\hat{n}}}& S(\mathbf{T}_{\hat{n}} \otimes \mathbf{S}_{\hat{m}}\|\mathbf{T}^{\text{F}}_{\hat{n}} \otimes \mathbf{S}_{\hat{m}}) \\
    =  \inf_{\mathbf{T}^{\text{F}}_{\hat{n}} \in \mathsf{T}^\text{F}_{\hat{n}}}& S(\mathbf{T}_{\hat{n}}\|\mathbf{T}^{\text{F}}_{\hat{n}}) . 
\end{aligned}
\end{equation}
The second line is an inequality in general but will be an equality for theories where the free processes have no correlations between steps. The equality between the second and third line is due to the additivity of entropy over independent subsystems. 
\end{proof}

In order to have a well-defined tensor product structure for process tensors, a notion of discarding process tensors is also required. Physically this would correspond to neglecting an experiment, and mathematically, it is taking the partial trace over a the process tensor which represents that experiment. Contractivity of relative entropy and trace distance under partial traces ensures that this still respects our monotones. The final requirement for a tensor product structure of process tensors is invariance under permutations of the tensor product, which trace distance and relative entropy also respect.

In most existing resource theories, the inclusion of a tensor product structure is only made meaningful by free operations which can cause subsystems to interact. Considering two totally isolated laboratories jointly is superfluous unless there is some kind of relation between them, such as sharing entanglement and allowing classical communication (LOCC). In our case, these relations between subsystems is not only spatial but also temporal. We have experimenters who can act on multiple subsystems at multiple times, jointly harnessing `spatiotemporal' correlations.

\subsubsection{Sequential Composition} \label{sec:sequentialcomp}

An experimenter may choose to extend the duration of their experiment, which corresponds to appending an additional process tensor before the beginning or after the conclusion of the original experiment. In doing this one must ensure that the post operation of the first experiment, and the pre-operation of the second are still possible, creating an intermediate intervention at the moment where the join occurs. This can be achieved by adjoining an identity operation to the end of the end of the process tensor which occurs first, and then composing the result together. By contrast the traditional notion of channel composition involves composing the output of the first channel with the input of the next, but without the possibility of an intermediate intervention at the moment of the join. However, the traditional notion of channel composition can still be recovered from the process tensor notion by coarse-graining away the intermediate intervention. Additionally, in the process tensor case like the channel case, no environment memory can transport information about the system between the two objects being composed.

Given two process tensors $\mathbf{T}_{\hat{n}}=\text{tr}_{e} \left\{ \mathcal{T}^{s e}_{t:t_{n}} \circ_{e} \dots \circ_{e} \mathcal{T}^{s e}_{t_{1}:t'} \circ_{e} \tau^{e}  \right\}$ and $\mathbf{S}_{{\hat{m}}}=\text{tr}_{e} \left\{ \mathcal{S}^{s e}_{t':t'_{m}} \circ_{e} \dots \circ_{e} \mathcal{S}^{s e}_{t_{1}:0} \circ_{e} \sigma^{e} \right\}$, sequential composition is defined as
\begin{equation} \label{eq:sequentialcomp}
\begin{aligned}
    \mathbf{T}_{\hat{n}} \circ \mathbf{S}_{\hat{m}}:&= \text{tr}_{e} \left\{ \mathcal{T}^{s e}_{t:t_{n}} \circ_{e} \dots \circ_{e} \mathcal{T}^{s e}_{t_1:t'} \circ_{e} \tau^{e} \right\} \circ  \text{tr}_{e} \left\{\mathcal{I}^{s e} \circ_{e} \mathcal{S}^{s e}_{t':t'_m} \circ_{e} \dots \circ_{e} \mathcal{S}^{s e}_{t'_1:0} \circ_{e} \sigma^{e} \right\}  \\
    &= \text{tr}_{e} \left\{ \mathcal{T}^{s e}_{t:t_{n}} \circ_{e} \dots \circ_{e} \mathcal{T}^{s e}_{t_1:t'} \circ_{e} \tau^{e} \right\} \otimes  \text{tr}_{e} \left\{ \mathcal{S}^{s e}_{t':t'_m} \circ_{e} \dots \circ_{e} \mathcal{S}^{s e}_{t'_1:0} \circ_{e} \sigma^{e} \right\},
    \end{aligned}
\end{equation}
where $\mathcal{I}^{se}$ is the identity channel on the system and environment, placed after the final channel within $\mathbf{S}_{\hat{m}}$ to make room for a post-operation to the $\mathbf{S}_{\hat{m}}$ (which is also a pre-operation for $\mathbf{T}_{\hat{n}}$). The first line is equal to the second line because the composition $\mathbf{T}_{\hat{n}} \circ \mathbf{S}_{\hat{m}}$ is Markovian about the join between $\mathbf{T}_{\hat{n}}$ and $\mathbf{S}_{\hat{m}}$. However, observe the difference between this second line and Eq.~\eqref{eq:parallelcomp}. In the parallel case, other than both occuring within the same window $(0,t)$, $\hat{n}$ and $\hat{n}$ are not related by any temporal order. In Eq.~\eqref{eq:sequentialcomp}, all times in $\hat{n}$ subsequent to all times in $\hat{m}$. An example of sequential composition is given in Fig.~\ref{fig:sequentialcompositionexample}.
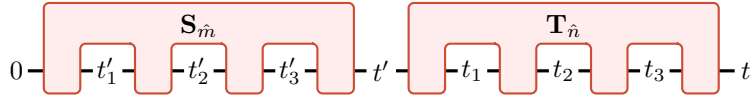
\begin{figure}[ht!]
\centering
\begin{tikzpicture}[scale=0.6]

\draw[black, very thick,solid] (-0.25,0.5) -- (15.25,0.5);

\draw[myred,fill=myredfill, thick,solid,rounded corners=2] (0+0.1,1+0.1) -- (0+0.1,0) -- (1-0.1,0) -- (1-0.1,1+0.1) -- (2-0.1,1+0.1) --
(0+0.1+2,1+0.1) -- (0+0.1+2,0) -- (1-0.1+2,0) -- (1-0.1+2,1+0.1) -- (2-0.1+2,1+0.1) --
(0+0.1+4,1+0.1) -- (0+0.1+4,0) -- (1-0.1+4,0) -- (1-0.1+4,1+0.1) -- (2-0.1+4,1+0.1) --
(0+0.1+6,1+0.1) -- (0+0.1+6,0) -- (1-0.1+6,0) -- (1-0.1+6,1+0.1) 
-- (1-0.1+6,2) -- (0+0.1,2) -- (0+0.1,1+0.1);

\draw[myred,fill=myredfill, thick,solid,rounded corners=2]
(0+0.1+8,1+0.1) -- (0+0.1+8,0) -- (1-0.1+8,0) -- (1-0.1+8,1+0.1) -- (2-0.1+8,1+0.1) --
(0+0.1+10,1+0.1) -- (0+0.1+10,0) -- (1-0.1+10,0) -- (1-0.1+10,1+0.1) -- (2-0.1+10,1+0.1) --
(0+0.1+12,1+0.1) -- (0+0.1+12,0) -- (1-0.1+12,0) -- (1-0.1+12,1+0.1) -- (2-0.1+12,1+0.1) --
(0+0.1+14,1+0.1) -- (0+0.1+14,0) -- (1-0.1+14,0) -- (1-0.1+14,1+0.1) 
-- (1-0.1+6+8,2) -- (0+0.1+8,2) -- (0+0.1+8,1+0.1);

\draw[white,fill=white,thick,solid,rounded corners=2] (1+0.2,0+0.2) rectangle (2-0.2,1-0.1);

\draw[white,fill=white,thick,solid,rounded corners=2] (1+0.2+2,0+0.2) rectangle (2-0.2+2,1-0.1);

\draw[white,fill=white,thick,solid,rounded corners=2] (1+0.2+4,0+0.2) rectangle (2-0.2+4,1-0.1);

\draw[white,fill=white,thick,solid,rounded corners=2] (1+0.2+6,0+0.2) rectangle (2-0.2+6,1-0.1);

\draw[white,fill=white,thick,solid,rounded corners=2] (1+0.2+8,0+0.2) rectangle (2-0.2+8,1-0.1);

\draw[white,fill=white,thick,solid,rounded corners=2] (1+0.2+10,0+0.2) rectangle (2-0.2+10,1-0.1);

\draw[white,fill=white,thick,solid,rounded corners=2] (1+0.2+12,0+0.2) rectangle (2-0.2+12,1-0.1);

\draw[] (-0.5,0.5) node[rotate=0] { $0$};
\draw[] (1.5,0.5) node[rotate=0] { $t'_1$};
\draw[] (3.5,0.5) node[rotate=0] { $t'_2$};
\draw[] (5.5,0.5) node[rotate=0] { $t'_3$};
\draw[] (7.5,0.5) node[rotate=0] { $t'$};

\draw[] (1.5+8,0.5) node[rotate=0] { $t_1$};
\draw[] (3.5+8,0.5) node[rotate=0] { $t_2$};
\draw[] (5.5+8,0.5) node[rotate=0] { $t_3$};
\draw[] (7.5+8,0.5) node[rotate=0] { $t$};

\draw[] (3.5,1.5) node[rotate=0] { $\mathbf{S}_{\hat{m}}$};
\draw[] (3.5+8,1.5) node[rotate=0] { $\mathbf{T}_{\hat{n}}$};

\end{tikzpicture}
\caption{An example of sequential composition $\mathbf{T}_{\hat{n}} \circ \mathbf{S}_{\hat{m}}$, with $\hat{n}=\{t_1,t_2,t_3 \}$ and $\hat{m}=\{t'_1,t'_2,t'_3 \}$. The input of $\mathbf{T}_{\hat{n}}$ overlaps with the output of $\mathbf{S}_{\hat{m}}$ at $t'$, where the corresponding pre/post operations become a new intermediate intervention. }  \label{fig:sequentialcompositionexample}
\end{figure}

We seek to formalise the idea that if an experimenter can always freely extend the duration of their experiment (either by starting earlier or finishing later), so long as the process that they are extending their experiment is a free one. 

Let $\mathbf{T}_{\hat{n}} \in \mathsf{T}_{\hat{n}}$ be some (potentially non-free) process the experimenter is given, and let $\mathbf{S}_{\hat{m}} \in \mathsf{T}^\text{F}_{\hat{m}}$ be a free one they seek to append. Define the transformations $\mathbf{T}_{\hat{n}} \mapsto  \mathbf{S}_{\hat{m}} \circ \mathbf{T}_{\hat{n}}$ and $\mathbf{T}_{\hat{n}} \mapsto  \mathbf{T}_{\hat{n}} \circ \mathbf{S}_{\hat{m}}$ both to be free operations. As before, it is not guaranteed that existing monotones will be respected. However, we show that relative entropy and trace distance both remain monotones after including these free operations.
\begin{theorem} \label{thm:sequentialtracedist}
For any process tensor $\mathbf{T}_{\hat{n}} \in \mathsf{T}_{\hat{\mathbbm{W}}}$, trace distance to the nearest free process in the sub-theory for fixed ${\hat{n}}$, $\inf_{\mathbf{T}^{\text{F}}_{\hat{n}} \in \mathsf{T}^\text{F}_{\hat{n}}}D(\mathbf{T}_{\hat{n}},\mathbf{T}^{\text{F}}_{\hat{n}})$ is contractive under the pre and post sequential composition with free resources.
\end{theorem}
\begin{proof}
The proof here closely resembles the parallel case, which is no accident. Since the Choi isomorphism treats different time steps as different spatial subsystems, and the sequential composition of two process tensors is Markovian about the partitioning induced by the join, in this context $\circ$ is indistinguishable from the regular $\otimes$ except for one key difference. In the parallel case, the infinum is taken over processes with totally independent sets of times for intermediate interventions $\hat{n}\otimes \hat{m}$. In the sequential case, these two sets of times must respect temporal ordering, i.e. all times in $\hat{n}$ before all times in $\hat{m}$. This changes the range of processes which can be considered in the infinum when finding the nearest free process. Let $\mathbf{T}_{\hat{n}} \in \mathsf{T}_{\hat{n}}$ and $\mathbf{S}_{\hat{m}} \in \mathsf{T}^\text{F}_{\hat{m}}$. The trace distance to the nearest free process $\mathbf{R}^{\text{F}}_{\hat{n} \cup {\hat{m}}}$ for fixed ${\hat{n} \cup {\hat{m}}}$ after sequential composition is
\begin{equation}
\begin{aligned}
      \inf_{\mathbf{R}^{\text{F}}_{\hat{n} \cup {\hat{m}}} \in \mathsf{T}^\text{F}_{\hat{n} \cup \hat{m}}}  &D(\mathbf{T}_{\hat{n}} \circ \mathbf{S}_{\hat{m}},\mathbf{R}^{\text{F}}_{\hat{n}\cup {\hat{m}}} ) \\
    =  \inf_{\mathbf{R}^{\text{F}}_{\hat{n} \cup {\hat{m}}} \in \mathsf{T}^\text{F}_{\hat{n} \cup \hat{m}}}  &D(\mathbf{T}_{\hat{n}} \otimes \mathbf{S}_{\hat{m}},\mathbf{R}^{\text{F}}_{\hat{n}\cup {\hat{m}}} ) \\
    \leq  \inf_{\mathbf{T}^{\text{F}}_{\hat{n}} \in \mathsf{T}^\text{F}_{\hat{n}}}  &D(\mathbf{T}_{\hat{n}} \otimes \mathbf{S}_{\hat{m}},\mathbf{T}^{\text{F}}_{\hat{n}} \otimes \mathbf{S}_{\hat{m}}) \\
    \leq  \inf_{\mathbf{T}^{\text{F}}_{\hat{n}} \in \mathsf{T}^\text{F}_{\hat{n}}}  &D(\mathbf{T}_{\hat{n}},\mathbf{T}^{\text{F}}_{\hat{n}}) . 
\end{aligned}
\end{equation}
The first and second line are equal, since sequential and parallel composition can both be represented with tensor products, but over different temporal orderings of $\hat{n}$ and $\hat{m}$. The third line is less optimal than the second because we are specifying the prior part of $\mathbf{R}^{\text{F}}_{\hat{n} \cup {\hat{m}}}$ to be $\mathbf{S}_{\hat{m}}$. Finally, as in Thm.~\ref{thm:parallelcompositiontracedistance}, the last line follows from subadditivity of trace distance. This argument proves pre-composition with free resources, but post-composition follows from an identical proof.
\end{proof}
Again, a similar agrument can be used for relative entropy.
\begin{theorem} \label{thm:sequentialcompositionrelativeentropy}
For any process tensor $\mathbf{T}_{\hat{n}} \in \mathsf{T}_{\hat{\mathbbm{W}}}$, relative entropy to the nearest free process in the sub-theory for fixed ${\hat{n}}$, $\inf_{\mathbf{T}^{\text{F}}_{\hat{n}} \in \mathsf{T}^\text{F}_{\hat{n}}}S(\mathbf{T}_{\hat{n}}\|\mathbf{T}^{\text{F}}_{\hat{n}})$ is contractive under the pre and post sequential composition with free resources.
\end{theorem}
\begin{proof}
Let $\mathbf{T}_{\hat{n}} \in \mathsf{T}_{\hat{n}}$ and $\mathbf{S}_{\hat{m}} \in \mathsf{T}^\text{F}_{\hat{m}}$. Relative entropy in the pre-composition case can be written as
\begin{equation}
\begin{aligned}
      \inf_{\mathbf{R}^{\text{F}}_{\hat{n} \cup {\hat{m}}} \in \mathsf{T}^\text{F}_{\hat{n} \cup \hat{m}}}  &S(\mathbf{T}_{\hat{n}} \circ \mathbf{S}_{\hat{m}} \| \mathbf{R}^{\text{F}}_{\hat{n}\cup {\hat{m}}} ) \\
    =  \inf_{\mathbf{R}^{\text{F}}_{\hat{n} \cup {\hat{m}}} \in \mathsf{T}^\text{F}_{\hat{n} \cup \hat{m}}}  &S(\mathbf{T}_{\hat{n}} \otimes \mathbf{S}_{\hat{m}} \| \mathbf{R}^{\text{F}}_{\hat{n}\cup {\hat{m}}} ) \\
    \leq  \inf_{\mathbf{T}^{\text{F}}_{\hat{n}} \in \mathsf{T}^\text{F}_{\hat{n}}}  &S(\mathbf{T}_{\hat{n}} \otimes \mathbf{S}_{\hat{m}} \| \mathbf{T}^{\text{F}}_{\hat{n}} \otimes \mathbf{S}_{\hat{m}}) \\
    =  \inf_{\mathbf{T}^{\text{F}}_{\hat{n}} \in \mathsf{T}^\text{F}_{\hat{n}}}  &S(\mathbf{T}_{\hat{n}} \| \mathbf{T}^{\text{F}}_{\hat{n}}) . 
\end{aligned}
\end{equation}
The first and second line are equal for the same reason given in the proof of Thm.~\ref{thm:sequentialtracedist}. The remainder follows the same argument as in Thm.~\ref{parallelcompositionrelativeentropy}
\end{proof}
It should be stressed that Thm.~\ref{thm:sequentialtracedist} and Thm.~\ref{thm:sequentialcompositionrelativeentropy} only prove that adding \emph{free} processes is a resource non-increasing operation. In a physical laboratory, when the process the experimenter has access to is non-free, waiting longer might turn out to be equivalent to acquiring an additional valuable resource. Increasing resource value by waiting is harnessing the a resource flux from the environment in an analogous manner to how solar panels or wind turbines collect energy.

Having shown that resource theories of temporal resolution have notions of parallel and sequential composition which respect relevant monotones, the door is open to study concepts like asymptotic conversion, catalytic conversion, experiments with non-fixed durations, and much more. Since the set of process tensors subsumes the set of channels, and channel resource theories have already seen significant success~\cite{reversibleframework, fundementallimitationsondistillation}, we expect there are many novel results for these structures waiting to be uncovered.

\subsection{Markovianisation for Idealised DD} \label{sup:markovianisation}

In general, one cannot decide if noise is decouplable by looking at its master equation~\cite{canquantummarkovevolutions}, which is our motivation for considering process tensors, and their multitime correlations. However, in the case where the Hamiltonian can be decoupled, we show that the residual dynamics after DD are fully Markovian, given frequent enough interventions. This is evidence to support the claim that DD is expending non-Markovianity as a resource. This argument is based on an a similar existing result~\cite{canquantummarkovevolutions}.

Consider the process tensor $\mathbf{T}_{\hat{n}}$ generated by the Lindblad evolution on both $s$ and $e$, with sufficiently short gaps between evenly spaced interventions $\tau=t/(|\hat{n}|+1)$ such that it can be expressed as
\begin{equation} \label{eq:lindbladevo}
    \Phi^{se}(\rho^{se}):=e^{\tau\mathcal{L}^{se}} \rho^{se} e^{-\tau\mathcal{L}^{se}}=\rho^{se}-i\tau[H^{se},\rho^{se}]+ \tau \sum_k \gamma_k \left( L^{se}_k \rho^{se} L^{se\dag}_k - \frac{1}{2}\left\{ L^{se\dag}_k L^{se}_k, \rho^{se} \right\} \right).
\end{equation}
We ask that $H^{se}$ is decouplable, but make no restrictions on $L^{se}_k$. After applying an element $\mathcal{V}^{s}_l \in V$ from the DD group~\cite{dynamicaldecouplingofopenquantumsystems} for the system Hilbert space $s$, the evolution becomes
\begin{equation}
\begin{aligned}
     \Phi^{se}_l (\rho^{se}) =&  \frac{1}{|V|}\mathcal{V}^{s\dag}_l \circ \Phi^{se}  \circ \mathcal{V}^{s}_l (\rho^{se}) \\ 
    = &  \mathcal{V}^{s\dag}_l \circ \mathcal{V}^{s}_l (\rho^{se}) \\ -& i\tau \mathcal{V}^{s\dag}_l \left( [H^{se}, \mathcal{V}^{s}_l(\rho^{se})] \right) \\
    + & \tau \sum_k \gamma_k \mathcal{V}^{s\dag}_l\left( L^{se}_k \mathcal{V}^{s}_l(\rho^{se}) L^{se\dag}_k - \frac{1}{2}\left\{ L^{se\dag}_k L^{se}_k, \mathcal{V}^{s}_l(\rho^{se}) \right\} \right).
\end{aligned}
\end{equation}
Expanding out $\mathcal{V}^{s}_l( \ \cdot \ )=v^{s}_l ( \ \cdot \ ) v^{s\dag}_l$ resolves the expression to
\begin{equation}
\begin{aligned}
     \Phi^{se}_l (\rho^{se}) =&  \rho^{se} \\
    -& i\tau  \left( \left[ (v^{s\dag}_l H^{se} v^{s}_l), \rho^{se}\right] \right) \\
    + & \tau \sum_k \gamma_k \Big( (v^{s\dag}_l L^{se}_k v^{s}_l)\rho^{se}(v^{s\dag}_l L^{se\dag}_k v^{s}_l)  \\
    - & \frac{1}{2}\left( (v^{s\dag}_l L^{se\dag}_k v^{s}_l) (v^{s\dag}_l L^{se}_k v^{s}_l)\rho^{se}(v^{s\dag}_l v^{s}_l) + (v^{s\dag}_l v^{s}_l )\rho^{se}(v^{s\dag}_l L^{se\dag}_k v^{s}_l) (v^{s\dag}_l L^{se}_k v^{s}_l)  \Big) \right) ,
\end{aligned}
\end{equation}
which is merely a change of basis on the operators. Define $v^{s\dag}_l H^{se} v^{s}_l:=H^{se}_l$ and $v^{s\dag}_l L^{se}_k v^{s}_l:= L^{se}_{kl}$ to re-express the evolution as
\begin{equation}
\begin{aligned}
     \Phi^{se}_l (\rho^{se}) =  \rho^{se} 
    -i\tau[H^{se}_l,\rho^{se}] 
    + \tau \sum_k \gamma_k \left( L^{se}_{kl} \rho^{se} L^{se\dag}_{kl} - \frac{1}{2}\left\{ L^{se\dag}_{kl} L^{se}_{kl}, \rho^{se} \right\} \right).
\end{aligned}
\end{equation}

For a full DD sequence and coarse-graining $\llbracket \mathbf{T}_{\hat{n}} | \mathbf{Z}_{\hat{n}} | \mathbf{I}_{\hat{n} \setminus \hat{m}} \rrbracket$ in the small $\tau$ limit where a 
Trotter approximation is valid, it is possible to write
\begin{equation}
    \bigcirc_{l=1}^{|V|}\Phi_l (\rho^{se}) \approx \exp \left(\sum_{k=1}^{|V|}i\tau\mathcal{L}_l^{se} \right) (\rho^{se}).
\end{equation}
The `$\approx$' is equating the decoupled process tensor with multiple distinct evolutions, to one evolving under a single average Hamiltonian. While $\mathbf{Z}_{\hat{n}}$ preserves all of $I$, $M$, and $N$, the application of DD causes the coarse-graining to have a drastically different effect. The dynamics becomes
\begin{equation}
\begin{aligned}
    \exp \left(\sum_{l=1}^{|V|}i\tau\mathcal{L}_l^{se} \right) (\rho^{se}) 
    = &\rho^{se}  -i[\sum_{l=1}^{|V|} \tau H^{se}_l,\rho^{se}] \\
    +& \sum_{l=1}^{|V|} \tau \sum_k \gamma_k \left( L^{se}_{kl} \rho^{se} L^{se\dag}_{kl} - \frac{1}{2}\left\{ L^{se\dag}_{kl} L^{se}_{kl}, \rho^{se} \right\} \right).
\end{aligned}
\end{equation}
The definition of the decoupling group $V$, $\sum_{l=1}^{|V|} v^{s}_l X^{se}_l v^{s\dag}_l=I^{s}\otimes B^e$ for some operator $B^e$ on the environment state, and any operator $X^{se}$, implies that the actions of the Hamiltonian terms on $s$ will cancel exactly in this scenario where $\tau$ is allowed to get arbitrarily small. However, the same rule cannot be applied to the dissipator term in the Lindblad superoperator, suggesting that cancellation is not guaranteed. The remaining dynamics from the perspective of $s$ alone only contains contributions from the dissipator $L^{se\dag}_{kl}$
\begin{equation}
\text{tr}_{e} \left\{ \bigcirc_{l=1}^{|V|}\Phi_l (\rho^{se}) \right\}
    = \rho^{s} - \frac{1}{2}\text{tr}_{e} \left\{ \sum_{l=1}^{|V|} \tau \sum_k \gamma_k \ \left\{ L^{se\dag}_{kl} L^{se}_{kl}, \rho^{se} \right\}  \right\},
\end{equation}
implying that the generator of the dynamics on $s$ is purely Markovian, meaning that $\mathbf{T}_{\hat{m}}$ will be too.

\subsection{Sub-theories of $\mathsf{Q}_{\hat{\mathbbm{W}}}$} \label{sup:subtheories}

$\mathsf{Q}_{\hat{\mathbbm{W}}}$ is capable of describing the actions of an experimenter performing DD, and how those actions are perceived after coarse-graining as preserving information at the system-level via the consumption of non-Markovianity. However, DD does not use the full capabilities of experimenters operating within $\mathsf{Q}_{\hat{\mathbbm{W}}}$. For example, QEC uses free transformations within $\mathsf{Q}_{\hat{\mathbbm{W}}}$ that DD does not call on. In QEC, information is redundantly encoded in a multipartite state such that when errors occur, syndrome measurements, followed by corresponding unitaries can be performed to correct those errors. The measurements required for QEC are allowed in $\mathsf{Q}_{\hat{\mathbbm{W}}}$, but not required by DD. As such we underline the lesser requirements to perform DD, from the full resource theory $\mathsf{Q}_{\hat{\mathbbm{W}}}$, using the restriction $\mathsf{D}_{\hat{\mathbbm{W}}} \subset \mathsf{Q}_{\hat{\mathbbm{W}}}$, constraining the free superprocesses to be sequences of unitaries. Consequently $I$, $M$, and $N$ will be invariant under the free superprocesses of this $\mathsf{D}_{\hat{\mathbbm{W}}}$.

Still, $\mathsf{Q}_{\hat{\mathbbm{W}}}$ encompasses many other noise suppression techniques. We illustrate two other restricted sub-theories of $\mathsf{Q}_{\hat{\mathbbm{W}}}$ with insufficient sets of free superprocesses to permit DD or QEC, but still allow other techniques: the inducement of decoherence free subspaces (DFS), and the quantum Zeno effect (QZE).

\subsubsection{Resource Theory for DFS Inducement} 

 For the case of DD, we saw in Sup.~\ref{sup:markovianisation} that an appropriately chosen sequence of unitary pulses performed by a fine-grained experimenter operating within $\mathsf{Q}_{\hat{\mathbbm{W}}}$ has the potential to eliminate unitary evolution at the system level for many (although not all) Hamiltonians, leaving only a residual Markovian evolution in tact. Here, we show that an experimenter under stricter constraints can potentially do the converse: apply an appropriately chosen sequence of pulses to eliminate dissipative evolution, while preserving the unitary evolution.
 
A DFS $P(\mathcal{L}^{se})$ can be defined~\cite{algebraicconditionsforconvergence} (at the $s$-$e$ level) as a set of states $\rho^{se}$ such that Lindblad evolution is non-dissipative 
\begin{equation}
    \Phi^{se}(\rho^{se})=e^{t\mathcal{L}^{se}}\rho^{se} e^{-t\mathcal{L}^{se}}=e^{tH^{se}}\rho^{se} e^{-tH^{se}},
\end{equation}
approximating brief duration $t$. For a decoherence free subspace to be useful, we also require it to at least be of dimension two, i.e. contains more than a single valid quantum state. It has been shown~\cite{algebraicconditionsforconvergence} that this definition is equivalent to the condition 
\begin{equation} \label{eq:DFScondition}
    P(\mathcal{L}^{se})= \left\{ \rho^{se} \ : \ \big[ \text{ad}^{l}_{H^{se}} (L^{se}_k)  ,\rho^{se} \big]=0, \ \big[ \text{ad}^{l}_{H^{se}} (L^{se \dag}_k)  ,\rho^{se} \big]=0, \     \ : \ k\geq 1, \  l\geq 0 \right\},
\end{equation}
where $\text{ad}_{H^{se}}$ is a map acting as $\text{ad}^0_{H^{se}}(a)=a^{se}$ and $\text{ad}^1_{H^{se}}(a^{se})=[H^{se},a^{se}]$. Consider the fine-grained picture, where a process tensor $\mathbf{T}_{\hat{n}} \in \mathsf{T}_{\hat{\mathbbm{W}}}$ is derived from a general $se$ Lindblad evolution. For adequately small $\tau$ the evolution $\Phi^{se} (\rho^{se})$ is the same as for Eq.~\eqref{eq:lindbladevo}. 

In order to perform DD, the experimenter needs to keep a clock-like memory of which operation to perform at which time. However, this is not required to induce a DFS, which is akin to a change of basis transformation. In order to delineate the abilities required to perform DD and to induce a DFS, we create a sub-theory $\mathsf{P}_{\hat{\mathbbm{W}}} \subset \mathsf{D}_{\hat{\mathbbm{W}}} \subset \mathsf{Q}_{\hat{\mathbbm{W}}}$ where the free superprocesses cannot use a clock-like memory resource of which operation to perform at which time. Hence, at any given moment, the operations performed $\mathcal{V}^{s}$ and $\mathcal{W}^{s}$ are identical,
\begin{equation}
    \llbracket \mathbf{T}_n | \mathbf{Z}_n = \mathcal{V}^{s} \circ \mathcal{W}^{s} \circ \mathcal{T}^{se}_{n:n-1} \circ_{e} \mathcal{V}^{s} \circ_{e} \dots \circ_{e} \mathcal{W}^{s} \circ \mathcal{T}^{se}_{1:0} \circ \mathcal{V}^{s} \circ \mathcal{W}^{s}.
\end{equation}
Physically, $\mathcal{V}^{s}$ and $\mathcal{W}^{s}$ are part of the same operation, since they occur at the same instant in time. The free process tensors for sub-theories of fixed $\hat{n}$ are stationary processes. However, once coarse-graining is also considered, the free resources are still zero capacity channels.

Clearly, not all processes will have the necessary symmetries to induce a DFS. However, it is not hard to see what is going on when this technique does work. 
\begin{example}
As a prototypical example where a DFS can be induced, take $\mathcal{L}$ for an qubit under going simultaneous rotation $H= \sigma_{x}$, and dephasing $L= \sqrt{\gamma}\sigma_z$. A evolution of small duration $\tau$ can be written as 
\begin{equation} \label{eq:rotateanddephase}
\begin{aligned}
       \Phi (\rho)   =  \rho 
    -i\tau [\sigma_{x},\rho] 
     + \tau \gamma (  \sigma_z  \rho \sigma_z   -  \rho).
\end{aligned}
\end{equation}
Acting jointly, the rotation and dephasing destroys all mutual information. However, if the experimenter applies identical pulses $\mathcal{V}(\rho)= \sigma_z \rho \sigma_z$ spaced $\tau$ apart, an evolution of $2\tau$ becomes
\begin{equation} 
\begin{aligned}
       \mathcal{V} \circ \Phi  \circ \mathcal{V} \circ \Phi  \circ \rho  =  \rho 
    + 2\tau \gamma (    \sigma_z  \rho \sigma_z    -  \rho^{se} ),
\end{aligned}
\end{equation}
which corresponds to dephasing alone. This is known to have a 1D DFS along the $z$-axis of the Bloch sphere. Hence, we have created a decoherence free subspace. Here we see a DD-like sequence has resulted in an increase in mutual information. However, the `resource' in the noise process was a symmetry in the Lindblad jump operator, as opposed to non-Markovianity.
\end{example}

This technique still relies on rapid sequences of operations akin to DD. However, unlike was seen in Sup~\ref{sup:markovianisation}, the benefit does not derive from the elimination of non-Markovian evolution. The pre-requisite for a DFS is symmetries in the underlying dynamics. However, we do not preclude that symmetries might also be helpful for DD. Furthermore, since $\mathsf{P}_{\hat{\mathbbm{W}}} \subset \mathsf{D}_{\hat{\mathbbm{W}}}$, there is no reason why an experimenter operating within $\mathsf{D}_{\hat{\mathbbm{W}}}$ cannot simultaneously harness both resources. This may also partly explain the significant improved performance of our pulse optimisation methods over traditional DD.

\subsubsection{Resource Theory for QZE}

One might want to know what an experimenter is capable given even more stringent constraints on the types of operations allowed, for example that all actions of the experimenter are destructive to quantum information. The allowed superprocesses in the sub-theory $\mathsf{C}_{\hat{\mathbbm{W}}} \subset \mathsf{Q}_{\hat{\mathbbm{W}}}$ act as
\begin{equation} \label{eq:0bprocesses}
    \llbracket \mathbf{T}_n | \mathbf{Z}_n =\sum_{\{i_k,i_l\}} \mathcal{M}^{s}_{n_k} \circ \mathcal{M}^{s}_{n_l} \circ \mathcal{T}^{se}_{n:n-1} \circ \mathcal{M}^{s}_{{n-1}_k} \circ_e \dots \circ_e \mathcal{M}^{s}_{1_l} \circ \mathcal{T}^{se}_{1:0} \circ \mathcal{M}^{s}_{0_k} \circ \mathcal{M}^{s}_{0_l},
\end{equation}
where $\mathcal{M}^{s}_{i_k}=\mathcal{M}^{s}_{j_k}$ for all $i,j$, and $\mathcal{M}^{s}_{i_l}=\mathcal{M}^{s}_{j_l}$ for all $i,j$  $\mathcal{M}^{s}_{k}$. All $\mathcal{M}^{s}_{a}$ are entanglement breaking channels $\mathcal{M}^{s}_{a}(\rho^{se}) = \pi_a \circ \text{tr}_a \{ \Pi_a \rho^{se} \}$, with $\Pi_a$ as a POVM measurement $\pi_a$ as a conditioned re-preparation pair. Utilising the quantum Zeno effect simply requires setting all $\mathcal{M}^{s}_{a}$ to be identical, and ensuring that the operators $\Pi_k$ are projectors onto re-prepared states $\pi_k$.

In $\mathsf{C}_{\hat{\mathbbm{W}}}$, choosing to act may be more detrimental than doing nothing for the preservation of quantum information, but a classical measurement outcome can be preserved via the QZE. As with $\mathsf{P}_{\hat{\mathbbm{W}}}$ set of free resources is still zero capacity channels.

\subsection{Notation summary} \label{sup:notationsummary}

In this section is a reference for the notation used throughout this work. The script used for upper case letters has a meaning: calligraphic letters refer to traditional quantum maps (superoperators), sans-serif letters indicate sets, boldface letters correspond to higher order maps/quantum combs, and regular math text is used for functions to real numbers. Lower case letters are typically used for indices, and hats are used to denote that the index is a set rather than a number. 
\begin{table*}[htbp] \label{tabnotationsummary}
\centering
\begin{tabular}{|l|l|}
\hline
\textbf{Object} & \textbf{Meaning}                                                              \\ \hline \hline
${\hat{n}}$    & Times for intermediate interventions                                                              \\ \hline
$\mathbf{T}_{\hat{n}}$    & Process tensor                                                              \\ \hline
$\mathbf{A}_{\hat{n}}$   & Control sequence                                                             \\ \hline
$\mathbf{Z}_{\hat{n}}$   & Superprocess                                                     \\ \hline
$\mathbf{I}_{\hat{n} \setminus \hat{m}}$  & Temporal coarse-graining                                                    \\ \hline
$\mathsf{S}_{\hat{\mathbbm{W}}}$   & Resource theory of temporal resolution                                          \\ \hline
$\mathsf{T}_{\hat{\mathbbm{W}}}$   & Process tensors from $\mathsf{S}_{\hat{\mathbbm{W}}}$                                             \\ \hline
$\mathsf{Z}_{\hat{\mathbbm{W}}}$   & Superprocesses from $\mathsf{S}_{\hat{\mathbbm{W}}}$                                             \\ \hline
$\mathsf{I}_{\hat{\mathbbm{W}}}$   & Coarse-grainings from $\mathsf{S}_{\hat{\mathbbm{W}}}$                                     \\ \hline
$\mathsf{S}_{\hat{n}}$   & Sub-theory for fixed $\hat{n}$/RTQP                                          \\ \hline
$\mathsf{T}_{\hat{n}}$   & Process tensors from $\mathsf{S}_{\hat{n}}$                                             \\ \hline
$\mathsf{Z}_{\hat{n}}$   & Superprocesses from $\mathsf{S}_{\hat{n}}$                                             \\ \hline

$\mathsf{Q}_{\hat{\mathbbm{W}}}$   & Theory for information preservation                                          \\ \hline
$\mathsf{D}_{\hat{\mathbbm{W}}}$   & Theory for DD                                          \\ \hline
$\mathsf{P}_{\hat{\mathbbm{W}}}$   & Theory for DFS inducement                                          \\ \hline
$\mathsf{C}_{\hat{\mathbbm{W}}}$   & Theory for QZE                                          \\ \hline

$\mathcal{V}_{t_i}$   & Pre-operation from superprocess at $t_i$                                            \\ \hline
$\mathcal{W}_{t_i}$   & Post-operation from superprocess at $t_i$                                              \\ \hline
$I$   & Total mutual information                                              \\ \hline
$M$   & Markov information                                              \\ \hline
$N$   & Multitime non-Markovianity                                             \\ \hline
$I_{\emptyset}$   & Channel-level mutual information                                              \\ \hline
\end{tabular}
\vspace{0.1cm}
\caption{Summary of notation for the most common objects featured in this work.
}
\end{table*}

\end{document}